\def\isReadyToSubmit{0}   
\def\includeAuthor{1}     
\def\sandy{0}             

\ifnum\sandy=0
  \documentclass[letterpaper,twocolumn,10pt]{article}
\else
  \documentclass[letterpaper,single,10pt]{article}
\fi

\usepackage{times}
\usepackage{graphicx}
\usepackage{subcaption}
\usepackage{bbold}
\usepackage{array} 
\usepackage{amsmath,amsfonts,amssymb,amsthm,thmtools}
\usepackage{color}
\usepackage{xspace} 
\usepackage{algorithm} 
\usepackage{algpseudocode} 
\usepackage{wrapfig} 
\usepackage{flushend} 
\usepackage{lipsum}
\usepackage{multirow} 
\usepackage{booktabs} 
\usepackage{anyfontsize} 
\usepackage{paralist} 
\usepackage[font=small]{caption}
\usepackage{courier} 
\usepackage{verbatim}
\usepackage{calc}
\usepackage{pbox}
\PassOptionsToPackage{table}{xcolor}
\ifnum\sandy=1
    \usepackage{setspace}
\fi
\usepackage{listings}
\usepackage{color,soul}
\usepackage{enumitem}
\usepackage{cancel}
\usepackage{hyperref}
\usepackage{mdframed}
\usepackage[normalem]{ulem}

\usepackage[compact,small]{titlesec}
\titlespacing{\section}{1pt}{*.4}{*.2}
\titlespacing{\subsection}{1pt}{*.4}{*.2}
\titlespacing{\subsubsection}{1pt}{*.4}{*.2}

\usepackage{etoolbox}
\makeatletter
\patchcmd{\ttlh@hang}{\parindent\z@}{\parindent\z@\leavevmode}{}{}
\patchcmd{\ttlh@hang}{\noindent}{}{}{}
\makeatother

\theoremstyle{definition}

\newtheorem{definition}{Definition}

\declaretheoremstyle[%
    spaceabove=-6pt,%
    spacebelow=3pt,%
    headfont=\normalfont\itshape,%
    postheadspace=1em,%
    qed=\qedsymbol%
]{mystyle}
\declaretheorem[name={Proof},style=mystyle,unnumbered,]{compactproof}


\newcommand{\ie}{{\em i.e.,~}}
\newcommand{\eg}{{\em e.g.,~}}

\def\F{Fig.~}
\def\T{Tab.~}

\newcommand{\ar}[3]{} 
\include{latexdefs}
\include{mathnotation}

\newcommand{\heading}[1]{{\vspace{2pt}\noindent\bf{#1}}} 

{\makeatletter
\gdef\xxxmark{%
    \expandafter\ifx\csname @mpargs\endcsname\relax 
        \expandafter\ifx\csname @captype\endcsname\relax 
            \marginpar{\textcolor{red}{xxx~}}
        \else
            \textcolor{red}{xxx~}
        \fi
    \else
        \textcolor{red}{xxx~}
    \fi}
\gdef\xxx{\@ifnextchar[\xxx@lab\xxx@nolab}
\long\gdef\xxx@lab[#1]#2{{\bf [\xxxmark \textcolor{red}{#2} ---{\sc #1}]}}
\long\gdef\xxx@nolab#1{{\bf [\xxxmark \textcolor{red}{#1}]}}
\ifnum\isReadyToSubmit=1
    \long\gdef\xxx@lab[#1]#2{}\long\gdef\xxx@nolab#1{}
\fi
}

{\makeatletter
\gdef\edit{\@ifnextchar[\edit@lab\edit@nolab}
\long\gdef\edit@lab[#1]#2{[\textcolor{red}{#2} ---{\sc #1}]}
\long\gdef\edit@nolab#1{[\textcolor{red}{#1}]}
\ifnum\isReadyToSubmit=1
    \long\gdef\edit@lab[#1]#2{[#2]}
\fi
}

\newcommand{\ignore}[1]{}

\newcommand{\code}[1]{\tt{#1}}

\definecolor{codegreen}{rgb}{0,0.6,0}
\lstdefinestyle{codestyle}{
    commentstyle=\color{codegreen},
    keywordstyle=\bfseries,
    basicstyle=\scriptsize\ttfamily,
    breakatwhitespace=false,
    captionpos=b,
    keepspaces=true,
    numbers=left,
    numbersep=4pt,
    showspaces=false,
    showstringspaces=false,
    showtabs=false,
    tabsize=2,
    abovecaptionskip=2pt,
    belowcaptionskip=-20pt,
    xleftmargin=2em,
    rulesepcolor=\color{white},
    rulecolor=\color{white},
    frame=single
}
\DeclareMathAlphabet\mathbfcal{OMS}{cmsy}{b}{n}

\newtheorem*{proposition*}{Proposition}
\newtheorem*{theorem*}{Theorem}
\newtheorem*{lemma*}{Lemma}
\newtheorem*{corollary*}{Corollary}
\theoremstyle{definition}

\def\e{\epsilon}  
\def\eg{\epsilon_g}  
\def\D{\mathcal{D}}
\def\M{\mathcal{Q}}

\def\Y{\mathcal{V}}

\def\ES{\mathcal{S}}

\def\sysname{PrivateKube\xspace}
\def\privacyresource{private block\xspace}
\def\PrivacyResource{Private Block\xspace}
\def\Privacyresource{Private block\xspace}
\def\privacyresources{private blocks\xspace}
\def\Privacyresources{Private blocks\xspace}

\newcommand{\todo}[1]{\textcolor{blue}{#1}}

\def\ie{{i.e.}\xspace}
\def\eg{{e.g.}\xspace}

\def\etc{etc.\xspace}

\newenvironment{denseenum}{
\begin{enumerate}[topsep=2pt, partopsep=0pt, leftmargin=1.5em]
  \setlength{\itemsep}{2pt}
  \setlength{\parskip}{0pt}
  \setlength{\parsep}{0pt}
}{\end{enumerate}}

\usepackage{amsfonts}
\usepackage{amssymb}
\usepackage{amsthm}
\usepackage{amssymb}
\usepackage{amsmath}
\usepackage{style/usenix-2020-09}
\usepackage[available,functional]{style/usenixbadges}

\newtheorem{proposition}{Proposition}
\newtheorem{theorem}{Theorem}

\def\E{\mathbb{E}}
\def\R{\mathbb{R}}

\def\cD{\mathcal{D}}

\def\cR{\mathcal{R}}

\begin{document}

\ifnum\sandy=1
    \doublespacing
\fi
\date{}  

\title{\Large \bf Privacy Budget Scheduling\thanks{This is an extended version of the OSDI 2021 paper with the same title. This version includes additional appendices and minor content modifications to the core of the paper to reference these appendices.}}   

\ifnum\includeAuthor=1

    \author{
        {\rm Tao Luo\thanks{First co-authors of the paper with equal, complementary contributions.}}\\
        Columbia University
        \and
        {\rm Mingen Pan\footnotemark[2]}\\
        Columbia University
        \and
        {\rm Pierre Tholoniat\footnotemark[2]}\\
        Columbia University
        \and
        {\rm Asaf Cidon}\\
        Columbia University
        \and
        {\rm Roxana Geambasu}\\
        Columbia University
        \and
        {\rm Mathias L\'ecuyer}\\
        Microsoft Research
    }

\else
    \author{{\rm Submission \#133}}
\fi

\maketitle
\captionsetup{font=footnotesize}


\begin{abstract}

    Machine learning (ML) models trained on personal data have been shown to leak information about users.
    Differential privacy (DP) enables model training with a guaranteed bound on this leakage.
    Each new model trained with DP increases the bound on data leakage and can be seen as consuming part of a {\em global privacy budget} that should not be exceeded.
    This budget is a scarce resource that must be carefully managed to maximize the number of successfully trained models.

    We describe {\em \sysname}, an extension to the popular Kubernetes datacenter orchestrator that adds privacy as a new type of resource to be managed alongside other traditional compute resources, such as CPU, GPU, and memory.
    The abstractions we design for the privacy resource mirror those defined by Kubernetes for traditional resources, but there are also  major differences.
    For example, traditional compute resources are replenishable while privacy is not: a CPU can be regained after a model finishes execution while privacy budget cannot.
    This distinction forces a re-design of the scheduler.
    We present {\em DPF} ({\em {\uline D}ominant {\uline P}rivate Block {\uline F}airness}) -- a variant of the popular Dominant Resource Fairness (DRF) algorithm --
    that is geared toward the non-replenishable privacy resource but enjoys similar theoretical properties as DRF.

    We evaluate \sysname and DPF on microbenchmarks and an ML workload on Amazon Reviews data.
    Compared to existing baselines, DPF allows training more models under the same global privacy guarantee.
    This is especially true for DPF over R\'enyi DP, a highly composable form of DP.  


\end{abstract}

\vspace{-0.3cm}
\section{Introduction}                      
\vspace{-0.3cm}
\label{sec:introduction}
Increasing evidence suggests that machine learning (ML) models trained on sensitive, personal information -- such as auto-complete models trained on users' emails -- expose individual entries from their training sets~\cite{carlini2018theSecretSharer,shokri2017membership}.
Despite the evidence, there is an increasing trend to push models to end-user devices for faster predictions~\cite{baylor2017tfx,hazelwood2018applied,ravi2017onDeviceMachineIntelligence}, share them across teams in a company~\cite{Li:michelangelo,twitter-embeddings} and even externally~\cite{modelzoo,aws-marketplace-sagemaker}.

Differential privacy (DP)~\cite{dwork2006differential} promises to enable safe sharing of models by providing solid guarantees regarding the exposure of individuals' data through these models.
DP randomizes a computation over a dataset (e.g. training one model) to bound the leakage of individual entries in the dataset through the output of the computation (the model).
Each new DP computation increases this bound over data leakage, and can be seen as consuming part of a {\em global privacy budget} that should not be exceeded.
DP is {\em mature algorithmically}: most popular ML algorithms have been adapted to {\em individually} enforce the DP guarantee.
There are also libraries that implement these algorithms, including TensorFlow Privacy~\cite{tensorflow-privacy}, Opacus for PyTorch~\cite{opacus}, and multiple libraries for statistics~\cite{idm-diffprivlib,google-dp,ms-harvard-opendp}.  

Comparatively, DP research is {\em primitive on systems} that enforce a global DP guarantee across {\em multiple} DP algorithms.
Indeed, enforcing a global DP guarantee creates scheduling challenges that have never been addressed in the literature.
For example, given a dynamic ML workload of multiple models trained on the same user data stream, how should the global privacy budget be allocated to maximize the number of models that are successfully trained with DP?
Recently, we presented Sage, an incipient design of an ML training platform that maintains a global DP guarantee for a dynamic workload of ML pipelines operating on a continuous data stream~\cite{sage}.
Our key contribution was to show that by splitting the data stream into {\em blocks} (for example by time), enforcing a global DP guarantee over the entire stream reduces to enforcing the guarantee on each block.
This showed at a basic level how to operationalize a global DP guarantee for a dynamic ML workload. but left the challenging questions related to scheduling unresolved.  Moreover, our block notion was rudimentary, supporting only limited DP semantics (Event DP, which offers non-ideal protection~\cite{mir2011pan,Kifer2020GuidelinesFI}) and basic DP composition methods (which scale poorly with the number of models).

In this paper, we present {\em \sysname}, a plug-in extension to the popular Kubernetes workload orchestrator that can be used to schedule global privacy budgets for a dynamic workload of DP ML pipelines akin to Sage's.
The key insight is to (1) generalize the notion of private blocks to support a wider range of DP semantics and composition methods, and (2) incorporate private blocks as {\em a new, native resource} into Kubernetes, alongside traditional compute resources (such as CPU, GPU, and RAM), so they can be scheduled uniformly.
Despite intuitive correspondence of our privacy abstraction to Kubernetes abstractions for traditional resources, there are also significant semantic differences that force us to redesign the scheduling at a fundamental, algorithmic level.

Specifically, \privacyresources differ from traditional computing resources in two key dimensions.
First, once a portion of a \privacyresource is allocated to a task, it can never be recuperated.
Second, in many use cases, the utility of using \privacyresources is a step function: if a task has enough privacy budget it can make progress, but if it does not have sufficient budget, its accuracy can be affected in complex ways and it is often preferable to wait to accumulate enough budget before proceeding.
These two properties invalidate assumptions typically made by scheduling algorithms for traditional computing resources, such as the popular DRF~\cite{drf}, which we show loses the max-min fairness property if applied directly to \privacyresources.
In fact, we find that the very definitions of standard game-theoretical scheduler properties require change to apply to the characteristics of the privacy resource.

We develop a new algorithm for scheduling \privacyresources, called DPF (Dominant \Privacyresource Fairness).
DPF treats each \privacyresource as a
    {\em separate resource} that can be demanded (or not) by tasks.
Different tasks can demand different private blocks, creating heterogeneous resource demands and pointing to multi-resource scheduling algorithms, such as DRF~\cite{drf}, as a basis for DPF.
Similar to DRF, DPF allocates \privacyresources to the user that has the minimal \emph{dominant \privacyresource share} -- the maximum privacy budget requested by a user across the \privacyresources.
Different from DRF and other related scheduling algorithms~\cite{dynamicdrf,parkes2015beyond}, DPF releases privacy budgets progressively into the blocks, to ensure that future pipelines have access to the privacy resource in accordance to a fairness policy.
Moreover, DPF allocates requested budgets all-or-nothing to ensure that pipelines can achieve their accuracy goals.
We prove that DPF satisfies several important game-theoretic properties: sharing incentive, strategy-proofness, dynamic envy-freedom (a variant of traditional envy-freedom), and Pareto efficiency.

We evaluate \sysname on microbenchmarks and a workload on Amazon Reviews data.
We find that: (1) DPF grants more pipelines than baseline policies at a small cost in delay; (2) stronger DP semantics (such as User DP) require more budget and data,  increasing the need for judicious budget allocation as with DPF; (3) adapting DPF to R\'enyi DP~\cite{8049725}, the state-of-the-art composition method, enables allocation of either many more or much larger pipelines, and (4) our native integration of the privacy resource into Kubernetes lets us easily adapt the Grafana compute resource monitor to track privacy usage on par with compute usage.

Overall, this paper is the first to pose these questions:
(1)~what are the characteristics of the ``privacy resource'' in ML workloads,
(2)~how should scheduling algorithms support this resource, and
(3)~what kinds of game-theoretical properties can be guaranteed for this resource?
The answers, which form our primary contributions, are: (1)~the abstraction of the privacy resource as dynamically-arriving, non-replenishable \privacyresources, (2)~the DPF algorithm, and (3)~the theoretical properties of DPF.
All these are integrated into real systems, Kubernetes and Kubeflow, in a prototype that we have open-sourced: \url{https://github.com/columbia/privatekube}.

\vspace{-0.3cm}
\section{Threat Model and Background}       
\vspace{-0.3cm}
\label{sec:threat-model-and-background}
\subsection{Threat Model}
\label{sec:threat-model}

We are concerned with the sensitive data exposure that may occur when pushing models trained over user data to untrusted locations, such as mobile devices~\cite{baylor2017tfx,hazelwood2018applied,ravi2017onDeviceMachineIntelligence}, model stores that are widely shared among teams in a company~\cite{Li:michelangelo,twitter-embeddings}, or even opened to the world via prediction APIs~\cite{modelzoo,aws-marketplace-sagemaker}.
Our focus is not on singular models, pushed once, but rather on workloads of many models, trained periodically over increasing data from user streams.
For example, a company may train an auto-complete model daily or weekly to incorporate new data from an email stream, distributing the updated models to mobile devices for fast predictions.
Moreover, the company may use the same email stream to periodically train and disseminate multiple types of models, for example for recommendations, spam detection, and ad targeting.
This creates ample opportunity for an adversary to collect models and perform {\em privacy attacks} to siphon personal data.

Two classes of privacy attacks are particularly relevant: (1) {\em membership inference}, in which the adversary infers whether a particular entry (e.g., user) is in the training set based on either white-box or black-box access to the model and/or predictions~\cite{backes2016membership,dwork2015robustTraceability,homer2008resolving,shokri2017membership}; and (2) {\em reconstruction attacks}, in which the adversary infers unknown sensitive attributes about entries in the training set based on similar white-box or black-box access~\cite{carlini2018theSecretSharer,dinurNissim2003revealing,dwork2017exposed}.
We aim to ensure that an entry's {\em participation} in a company's model {\em does not increase the risk} of an adversary learning something about that entry.

Of particular concern are attacks that can access {\em multiple} models or statistics trained on the same or overlapping portions of a data stream.  While individually these may leak limited information about specific entries, together they may leak significant information, especially when combined with side information about an entry.  Consider two statistics: (1) average value of a sensitive column s (say representing user salary); and (2) average value of column s across entries whose ID differs from ``1234.'' Individually, they reveal nothing specific about any entry in a dataset.  Together, they reveal the value of sensitive column s for entry ``1234.''
This is a trivialized example in which the queries are ideally chosen and the adversary has access to ideal side-information about their target: the ID.
However, research in more practical settings has shown that releasing multiple (versions of) ML models trained over overlapping datasets increases the attacker's membership inference power compared to releasing just one~\cite{beguelinccs2020}.
Moreover, many pieces of information, such as demographic traits and locations, can be pieced together to uniquely identify individuals and used as side information in such attacks~\cite{demontjoye2013unique,Narayanan:2008:RDL:1397759.1398064,backes2016membership}.
Thus, a significant data exposure threat stems from the repeated release of models/statistics from overlapping portions of a stream.

\subsection{Differential Privacy}
\label{sec:dp}

DP is known to address the preceding attacks~\cite{shokri2017membership,dwork2017exposed,carlini2018theSecretSharer,236254}.
At a high level, membership and reconstruction attacks work by finding data
points (which can range from individual events to entire users) that make the observed model more likely: if those points were in the
training set, the likelihood of the observed output increases.  DP prevents
these attacks by ensuring that no specific data point can drastically increase
the likelihood of the model outputted by the training procedure.

To prevent such information leakage, DP introduces {\em randomness} into the computation to hide details of individual entries.
A randomized algorithm $\M : \D \rightarrow \Y$ is $(\epsilon, \delta)$-DP if for any
neighboring datasets $\D, \D'$ that differ in one row and for any $\ES \subseteq \Y$, we have:
$P(\M(\D) \in \ES) \leq e^\epsilon P(\M(\D') \in \ES) + \delta .$
Parameters $\e>0$ and $\delta \in [0,1]$ quantify the strength of the privacy guarantee: small values imply that one draw from such an algorithm's output gives little information about whether it ran on $D$ or $D'$.
The {\em privacy budget} $\e$ upper bounds an $(\epsilon, \delta)$-DP computation's privacy loss with probability (1-$\delta$).

A key strength of DP is its {\em composition} property, which in its basic form, states that the process of running an $(\e_1,\delta_1)$-DP and an $(\e_2,\delta_2)$-DP computation on the same dataset is $(\e_1+\e_2,\delta_1+\delta_2)$-DP.
Therefore, privacy loss accumulates linearly with the privacy loss of each computation.
Composition lets one account for the privacy loss resulting from a sequence of DP-computed outputs, such as the release of multiple models.  It is thus critical for enforcing a global DP guarantee.
There are more advanced forms of composition, such as R\'enyi DP~\cite{8049725}, which permit much tighter analysis of cumulative privacy loss (sublinear).  We discuss those in the latter parts of the paper, because they are vital to a well-performing globally DP system, but for the next two sections we assume basic composition for simplicity.

Multiple DP mechanisms exist, such as the Laplace and Gaussian mechanisms.
They add noise to the computation from a Laplace/Gaussian distribution scaled by a function of $\epsilon$, $\delta$, and the sensitivity of the computation.
The noise scale depends linearly in $1/\epsilon$ and at most logarithmically in $1/\delta$.
When enforcing a global DP guarantee, which we denote in this paper as $(\epsilon^G, \delta^G)$, both parameters become ``resources'' that must be allocated among the individual computations to ensure that cumulatively the computations do not exceed either.
However, because individual computations are much more sensitive to the allocated $\epsilon$ than to $\delta$, throughout this paper we will focus on $\epsilon^G$ as the sole global resource to schedule.
In evaluation, we set the individual $\delta$ requested by each pipeline small enough in comparison to $\delta^G$ ($10^{-9}$ and $10^{-7}$, respectively) such that $\epsilon^G$ is always the bottleneck.

The DP semantic can be instantiated at multiple granularities, the difference being what a ``row'' corresponds to.
    {\em Event DP} enforces DP on individual data points (e.g., individual clicks).
    {\em User DP} enforces DP on all data points contributed by a user.
It is stronger but challenging to sustain when new models must keep training on new data from the same users.
    {\em User-Time DP} is a middle-ground that enforces DP on all data points contributed by a user in a given period (\eg, one day).

\subsection{Assumptions}
\label{sec:threat-model-assumptions}

Our overarching goal is to {\em develop infrastructural support for organizations to enforce a global DP guarantee -- at Event, User, or User-Time level -- across the entire ML workload they operate on sensitive data streams}.
This would let organizations control the leakage of personal information through the models.
The focus of this paper is on how to orchestrate the global privacy budget across {\em competing} but {\em trusted} ML training processes, each of which is assumed to be coded by their programmers to enforce DP.
We assume that the programmers are trusted to correctly implement DP training processes and to adhere to the protocols we establish for them.
Moreover, we assume that the training processes themselves, plus the compute infrastructure, are trusted.
For example, if our scheduler refuses to allocate a requested privacy budget to a training task, the task will not access the data.
If the scheduler allocates the task's requested budget, $\epsilon$, then the training process will not attempt to use more than $\epsilon$.
On the other hand, programmers may be incentivised to achieve higher accuracy for their models by requesting more $\epsilon$. Therefore, we must provide users with strong incentives to fairly share $\epsilon^G$.

\vspace{-0.3cm}
\section{\sysname Architecture}             
\vspace{-0.3cm}
\label{sec:privatekube-architecture}

{\em \sysname} is a plug-in extension to the popular Kubernetes workload orchestrator.
It can be used to allocate privacy budgets for a dynamic workload of ML pipelines to enforce a global ($\epsilon^G, \delta^G$) DP semantic.
Our key insight is to incorporate the privacy budget as {\em a new, native resource} alongside traditional compute resources so developers can manage compute and privacy uniformly.
Despite one-to-one correspondence of our privacy resource abstractions to traditional Kubernetes abstractions, there are also significant semantic differences that cause us to re-think scheduling for the privacy resource.
This section gives an architectural view of our privacy resource abstraction, with the similarities and differences from Kubernetes' abstractions.
\S\ref{sec:dpf-algorithm} then describes {\em DPF}, the first scheduling algorithm suitable for the privacy resource.  \S\ref{sec:dpf-extensions} presents extensions of DPF to support both R\'enyi composition and all three DP semantics: Event, User, User-Time.  These, too, constitute firsts for the DP systems literature.

\begin{figure}[t]
    \includegraphics[width=\linewidth]{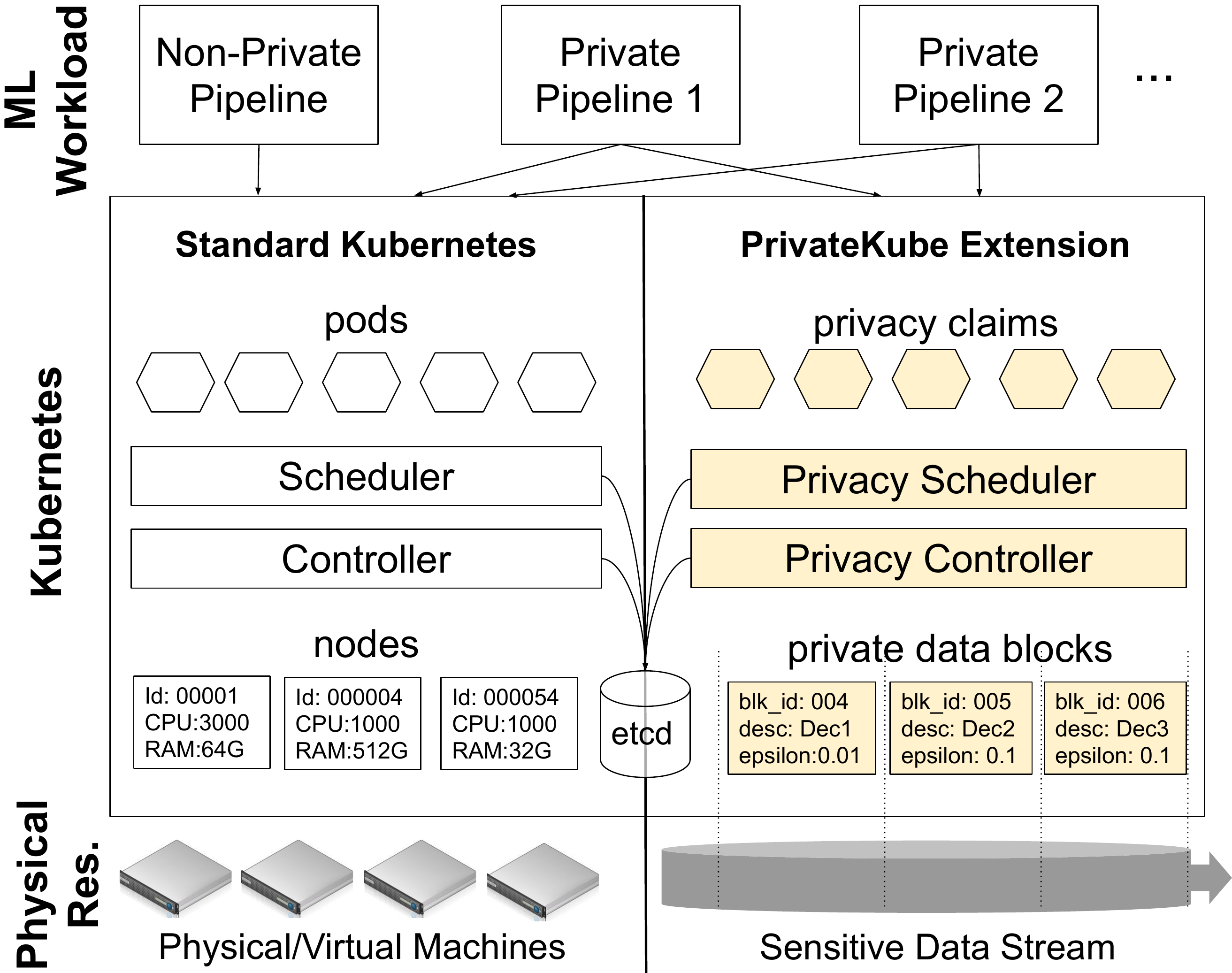}
    \caption{{\bf \sysname architecture.} Clear components are standard Kubernetes.
        Highlighted components (yellow) are added by \sysname.}
    \label{fig:architecture}
\end{figure}

\subsection{Overview}
\label{sec:overview}

\F\ref{fig:architecture} shows the \sysname architecture alongside the main components of a standard Kubernetes deployment.
It underscores the correspondence between traditional and privacy abstractions.
Kubernetes orchestrates the execution of a {\em workload} -- in our case an {\em ML workload} consisting of multiple training pipelines -- onto the {\em physical resources} available to the Kubernetes deployment.
In standard Kubernetes, the physical resources are physical or virtual machines.
The main abstractions that standard Kubernetes provides are: (1) {\em node}, an abstract representation for a physical or virtual machine; and (2) {\em pod}, a containerized unit of execution.
A pod specifies the container image to execute, plus the type and quantity of compute resources it demands, such as CPU, GPU, RAM, SSD.
A node specifies the type and quantity of compute resources it has available.
The primary functions of Kubernetes are to: (i) monitor for pods with unsatisfied resource demands (component {\em Controller} in \F\ref{fig:architecture}) and (ii) {\em bind} each pod to one node that has the demanded resources (component {\em Scheduler}).
Once a pod is bound to a node, the pod's image is executed.

\sysname extends Kubernetes to add a new type of physical resource: sensitive data streams.
We correspondingly add two new abstractions to Kubernetes: (1) {\em private data block} and (2) {\em privacy claim}.
Private data blocks (or {\em \privacyresources} for short) constitute non-overlapping portions of a sensitive data stream, such as daily windows of data from that stream.
\Privacyresources are the finest granularity at which data can be requested by a training pipeline, and the level at which \sysname keeps track of the total privacy loss incurred by an ML workload of multiple pipelines.
\Privacyresources specify the portion of the data they represent (e.g., the start and end times of the corresponding window), plus the privacy budget still available for use in that window.
    {\em Privacy claims} are used by training pipelines to demand privacy budget for the \privacyresources they are interested in.
A pipeline specifies in its privacy claims a selector for the \privacyresources it is requesting (such as the window of time from which they want data), plus the privacy budget it demands for these blocks.
The primary functions of \sysname are to: (i) monitor for privacy claims with unsatisfied \privacyresource demands (component {\em Privacy Controller} in \F\ref{fig:architecture}) and (ii) {\em bind} each privacy claim to the \privacyresources it demands (component {\em Privacy Scheduler}).

In a Kubernetes deployment with \sysname enabled, the workload may consist of a mix of non-private pipelines (which interact with insensitive data) and private pipelines (which interact with sensitive data).
Each pipeline has multiple steps organized in a directed acyclic graph, including steps that read the data, transform it, train models, etc.
The non-private pipeline interacts with standard Kubernetes to schedule its steps for execution by registering a pod for each step as soon as the step's inputs are available.
The private pipeline 
interacts not only with standard Kubernetes (to allocate compute resources for each step) but also with \sysname (to allocate and consume privacy budget needed to execute the steps on the sensitive data in a privacy preserving way).

\subsection{\sysname Abstractions}
\label{sec:architecture:private-data-block}

\begin{figure}[t]
    \centering
    \includegraphics[width=\linewidth]{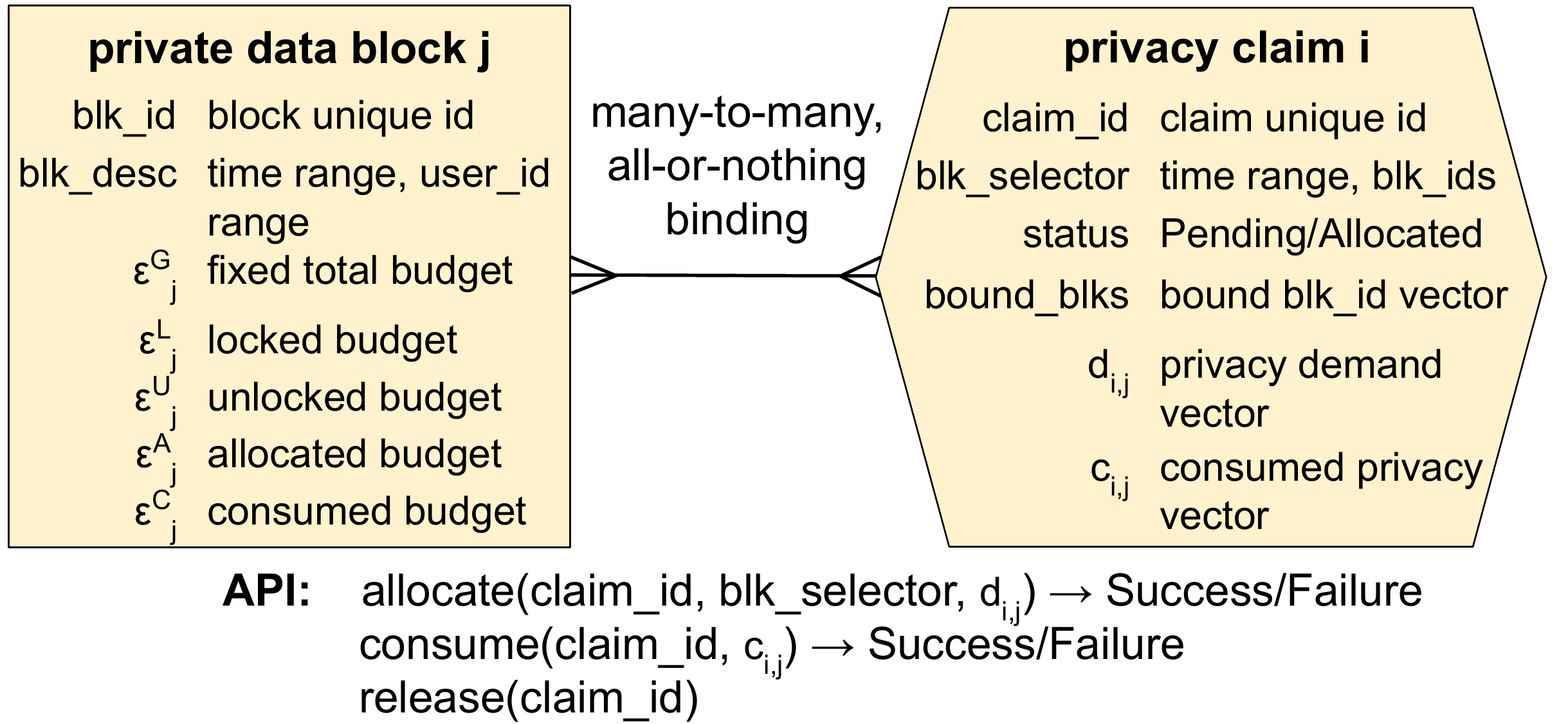}
    \caption{\footnotesize {\bf \sysname abstractions and API.} Some variables are indexed by block ($j$) or claim ($i$) for consistency with notation needed in \S\ref{sec:dpf-algorithm}.}
    \label{fig:abstractions-api}
\end{figure}

\sysname's abstractions are implemented {\em natively} in Kubernetes using its Custom Resource Definition extension API.
\F\ref{fig:abstractions-api} shows the state maintained for each abstraction.
As with standard abstractions, state for custom resources is stored in the fault-tolerant, strongly consistent etcd store.

\heading{\PrivacyResource} (\F\ref{fig:abstractions-api}, left):
This abstraction has three constant fields: a globally unique block id ({\code blk\_id}), a descriptor specifying the portion of the sensitive data stream it represents ({\code blk\_desc}), and the global privacy guarantee \sysname is configured to enforce against the entire stream ($\epsilon_j^G=\epsilon^G$).
\sysname supports multiple ways of splitting the stream into \privacyresources, and splitting determines the type of DP guarantee \sysname enforces: Event, User, or User-Time DP.
\S\ref{sec:dpf-extensions} shows how splitting works for each.

Each block $j$ also maintains four variable fields.
(1)~$\epsilon^C_j$ denotes the budget that has been consumed for the block. We leverage the theory we developed for Sage~\cite{sage} to justify that enforcing a global $\epsilon^G$ privacy guarantee over the entire stream reduces to ensuring that $\epsilon^C_j \le \epsilon^G_j=\epsilon^G$ for all blocks $j$ at all times.
Thus, when $\epsilon^C_j$ reaches $\epsilon^G$, we remove private block $j$ from Kubernetes and it no longer represents a resource.
(2)~$\epsilon^A_j$ denotes the part of block j's budget that has been allocated to some claims but not yet consumed.
(3)~$\epsilon^U_j$, called {\em unlocked budget}, is the unallocated and unconsumed budget made presently available for allocation to privacy claims.
(4)~$\epsilon^L_j$, called {\em locked budget}, is the unconsumed and unallocated budget not yet made available for allocation.
Our DPF algorithm (\S\ref{sec:dpf-algorithm}) leverages the last two fields to unlock budget from $\epsilon^G_j$ progressively to ensure that future pipelines have access to the privacy resource in accordance to a fairness policy.
Among all fields, the invariant is: $\epsilon^G_j=\epsilon^L_j+\epsilon^U_j+\epsilon^A_j+\epsilon^C_j$.

\heading{Privacy Claim} (\F\ref{fig:abstractions-api}, right):
This abstraction is used by pipelines to allocate and consume privacy budget from one or more \privacyresources.
When creating a privacy claim, the programmer specifies a selector for the data blocks relevant for their pipeline ({\code blk\_selector}).
Typically, this means specifying a time range from which the programmer wishes to obtain data samples (e.g., the past year).
\sysname then maps this descriptor onto the \privacyresources that contain data samples from that time range.
In addition to the block selector, the programmer also specifies the demanded privacy budget for each of the blocks that match the selector.
While often the demanded privacy budget will be uniform across all selected blocks, we allow the programmer to specify a {\em demand vector}, $d_{i,j}$, with one separate entry for each selected block.

\heading{API} (\F\ref{fig:abstractions-api}, bottom):
We implement three functions on privacy claims: {\code allocate}, {\code consume}, and {\code release}.
A pipeline can invoke them multiple times on the same claim, and they will be executed sequentially.
    {\code allocate} invokes the Privacy Scheduler to allocate privacy demand, $d_{i,j}$, to blocks that match the {\code blk\_selector}.
The scheduler will perform the selection, verify that every matching block has sufficient unconsumed and unallocated budget to potentially honor $d_{i,j}$, and if so, binds the matching blocks to the claim.
It then adds the claim to its internal list of claims to schedule with the DPF algorithm.
The scheduler will ultimately decide to allocate the request, or not.
If it does, {\code allocate} succeeds and the caller is guaranteed that the entire demand vector $d_{i,j}$ has been allocated to the bound blocks.
If it does not, the blocks are unbound, and the caller can assume that none of the requested budgets in its demand vector were allocated.
    {\code consume} invokes the Privacy Controller to deduct a part of previously allocated budget, $c_{i,j}$, from blocks already bound to the claim.
The function is similarly not guaranteed to succeed, for example if the caller is asking to consume more than the budget it has left for a block.
    {\code release} invokes the Privacy Controller to reclaim a previous unconsumed allocation to a claim.
For example, a pipeline invokes {\code release} if it decides to stop early and not execute some steps.
The Privacy Controller can also invoke {\code release} if the pipeline that owns the claim fails.

\subsection{Example Pipeline}
\label{sec:architecture:example-pipeline}

\begin{figure}[t]
    \centering
    \begin{subfigure}[t]{0.5\linewidth}
        \centering
        \includegraphics[width=\linewidth]{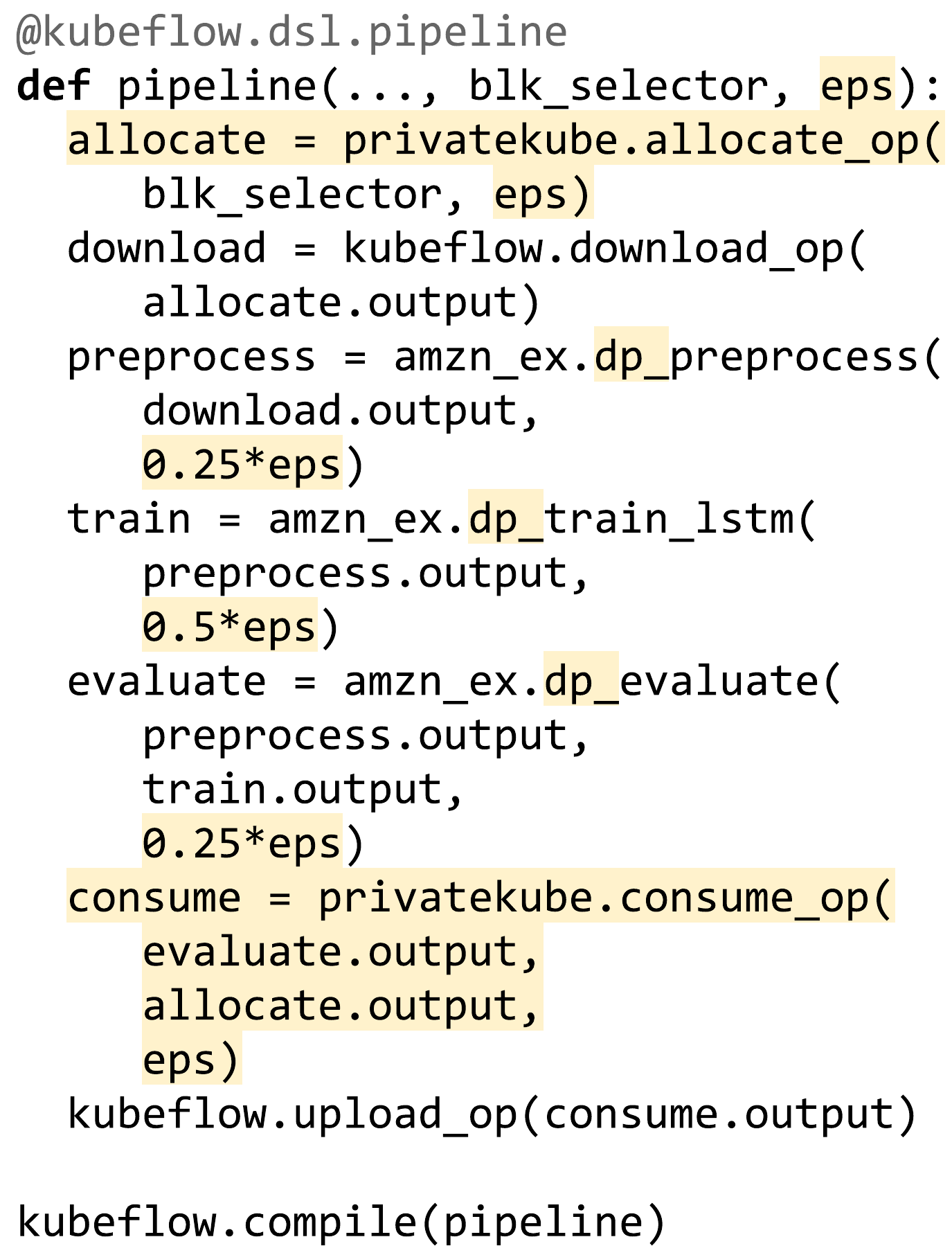}
        \caption{\footnotesize{\bf Pseudocode}}
        \label{fig:example-pipeline-pseudocode}
    \end{subfigure}%
    \begin{subfigure}[t]{0.5\linewidth}
        \centering
        \includegraphics[width=\linewidth]{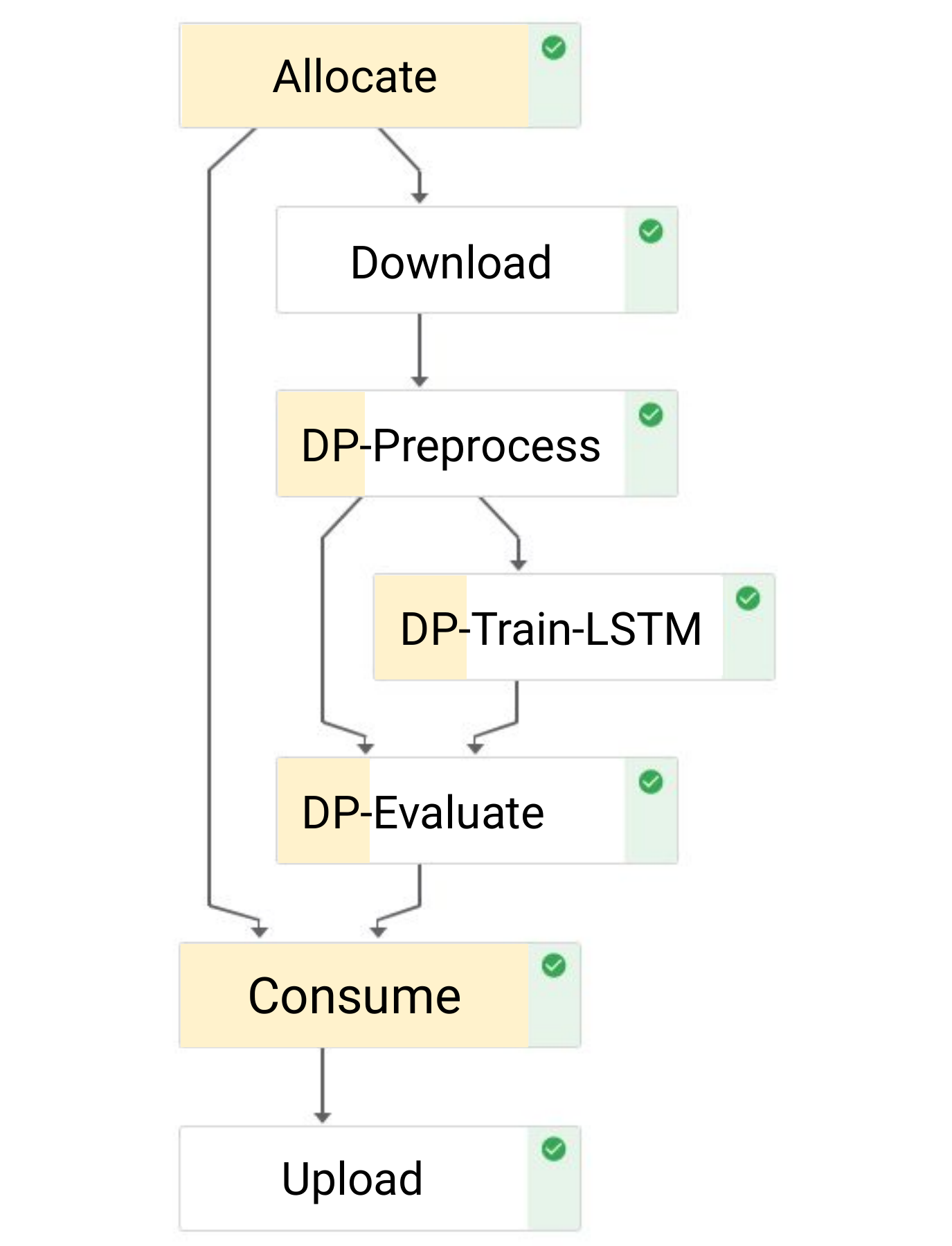}
        \caption{\footnotesize{\bf Execution Graph}}
        \label{fig:example-pipeline-execution-graph}
    \end{subfigure}
    \caption{\footnotesize {\bf Example private Kubeflow pipeline.}
        Distinctions from the non-private version are highlighted in yellow background.
    }
    \label{fig:example-pipeline}

\end{figure}

To exemplify usage of \sysname's abstractions and API, we describe a pipeline from our evaluation (Product/LSTM in \S\ref{sec:evaluation:macrobenchmarks}).
It is built in Kubeflow, an ML pipeline orchestrator for Kubernetes, and trains an NLP model on Amazon Reviews to predict a product category. 
\F\ref{fig:example-pipeline} shows (a) our code in Kubeflow DSL and (b) the pipeline's execution graph.
Highlighted are the distinctions between private and non-private versions.

The pipeline has three processing steps: {\code Preprocess} tokenizes the reviews; {\code Train-LSTM} trains an LSTM model with stochastic gradient descent (SGD); {\code Evaluate} validates that the model passes a baseline accuracy. The Kubeflow runtime executes each step in a separate pod and passes artifacts along the computation graph~\cite{kubeflow}.
If a step fails, its children in the graph will not be launched.
An important note for \sysname is that in Kubeflow, most steps of a pipeline are pure functions and do not communicate with the outside. Only a few well-defined Kubeflow components do, including: {\code Download} (loads data from an external source) and {\code Upload} (pushes an artifact to the serving infrastructure).

Focusing on the {\em private version} (highlighted parts of \F\ref{fig:example-pipeline}), the distinctions from a non-private pipeline are two-fold.
First, each step is coded by the programmer to enforce DP.
For example, the training step uses DP SGD instead of SGD.
The DP steps take an additional parameter: privacy budget ({\code eps}).
The programmer splits {\code eps} among the steps to enforce {\code eps} DP at pipeline level.
In the example, {\code dp\_preprocess} gets 25\% of {\code eps}, {\code dp\_train} 50\%, {\code dp\_evaluate} 25\%  (\F\ref{fig:example-pipeline-pseudocode}).

Second, the private pipeline interacts with \sysname to demand and consume {\code eps}.
This interaction is through drop-in Kubeflow components that we created to wrap \sysname's API.
This example highlights two such components: (1) {\code Allocate} and (2) {\code Consume}, wrappers around {\code allocate} and {\code consume}, respectively (\F\ref{fig:example-pipeline-execution-graph}).
The protocol is simple: place {\code Allocate} before any component accessing sensitive data (e.g., {\code Download}); place {\code Consume} before any component with externally visible side-effects (e.g., {\code Upload}).
(1) {\code Allocate} creates a privacy claim and invokes {\code allocate} on it with a block selector and {\code eps} privacy budget.
If {\code allocate} succeeds, then {\code Download} reads the data of the blocks bound to the claim ({\code bound\_blks}) and the training process begins.
If {\code allocate} fails, then {\code Download} is never launched and the sensitive data never accessed.
(2) {\code Consume} receives the privacy claim from {\code Allocate} and invokes {\code consume} on it with a privacy budget equal to the one that was consumed. 
If {\code consume} succeeds, then {\code Upload} runs and outputs the model artifact.
If {\code consume} fails, then {\code Upload} is never launched and the model never externalized.
Assuming programmers adhere to this protocol (\S\ref{sec:threat-model-assumptions}), the above ensures that \sysname controls the privacy loss resulting from externalizing ML artifacts.

\subsection{Kubernetes -- \sysname Distinctions}
\label{sec:architecture:kubernetes-kubeflow-distinctions}

Despite one-to-one mapping of our abstractions with Kubernetes' -- node::private block, pod::privacy claim -- there are also semantic differences.
First is the level at which we make scheduling decisions.
Consider the pipeline from \S\ref{sec:architecture:example-pipeline}.
The Kubernetes Scheduler performs a scheduling decision for each step.
It schedules our {\code Allocate} and {\code Consume} pods, as well as the functional pods.
In \sysname, we decided to {\bf \em allocate privacy at the level of entire pipelines}.
Indeed, after being allocated compute resources, the {\code Allocate} pod creates a privacy claim and invokes {\code allocate} on it.  This is when the Private Scheduler makes a scheduling decision for the privacy resource.
The privacy claim is then kept for the entirety of the pipeline and passed among its components as needed.
%

Second, in Kubernetes, the binding of pod to node is many-to-one: one pod can be bound only to one node, but the same node can be bound to multiple pods.
In \sysname, the binding is many-to-many: a privacy claim can be bound to many private blocks, and the same block can be bound to multiple claims.
This leads to a question of atomicity for the binding across multiple blocks.
A critical design decision we have made is an {\bf \em all-or-nothing semantic} for scheduling: a pipeline can expect {\code allocate} on its privacy claim to either fail or guarantee that (1) all the blocks matching the claim's selector were bound to the privacy claim, and (2) for each block, the demanded privacy budget was allocated in full.
This decision, which has significant impact on the scheduling algorithm (\S\ref{sec:dpf-algorithm}), should be thought of as a plausible assumption, though not the only reasonable one.
Multiple use cases justify all-or-nothing.
Many DP algorithms have complex interactions with hyper-parameters, such as learning rate and batch size; programmers may want to run on the budget for which those were tuned.
Other use cases include the need for comparable models and DP budget searches on a fixed schedule (as proposed in Sage~\cite{sage}).
Furthermore, the {\bf \em non-replenishable} nature of the privacy budget suggests that the scheduler should grant no more budget than a pipeline demanded, to keep as much budget available for future pipelines.

\vspace{-0.3cm}
\section{DPF Algorithm}                     
\vspace{-0.3cm}
\label{sec:dpf-algorithm}

Given the preceding integration of private blocks as a new resource in Kubernetes, we now explore how scheduling should work for this resource.
Can we achieve for privacy the same types of theoretical guarantees that compute schedulers often achieve?  How should scheduling algorithms change given the semantic differences between privacy and compute resources?
To obtain initial answers, we focus on max-min fairness guarantees and algorithms that support them.

Our idea is to model each \privacyresource as a {\em separate resource} that must be allocated to different pipelines based on their demands.
Demands will differ across pipelines, both in the blocks they select and in the privacy budgets they request for selected blocks.
Consider four blocks ($B0, B1, B2, B3$) and three pipelines requesting: $d_1=(0.5, 0.5, 0.5, 0.0)$; $d_2=(0.0, 0.1, 0.1, 0.1)$; and $d_3=(0.0, 0.0, 0.0, 0.01)$.
The pipelines could be: a large model (user embedding) registered before block B3 appeared; a smaller model that needs recent data (news recommendation) registered after B3 appeared; and a daily statistic invoked on B3.
Privacy demands being heterogeneous, the four blocks will have heterogeneous capacities left after the pipelines complete.

The preceding formulation points to DRF (Dominant Resource Fairness)~\cite{drf} -- an algorithm that achieves max-min fairness for multiple, heterogeneous compute resources (\eg, CPU, memory) -- as a basis for scheduling privacy.
However, as we will show, DRF's max-min fairness guarantees do not hold for scheduling privacy.
We next describe the limitations of DRF and several variations for privacy scheduling, after which we present the design and analysis of our new algorithm, {\em DPF} ({\em Dominant \Privacyresource Fairness}).

\subsection{Limitations of DRF and Variations}
\label{sec:dpf-algorithm:motivation}

We identify three limitations of DRF with respect to the privacy resource.
First, DRF assumes static resources and sometimes even static workloads.
In \sysname, we focus on a {\bf \em dynamic setting}: both pipelines and \privacyresources arrive to the system dynamically.
If we applied DRF on private blocks, at every point in time, DRF would try to consume the entire available budget to satisfy the demands of all present tasks. This would make it violate the sharing incentive guarantee of max-min fairness. A new task arriving to the system that asks for its fair share of privacy budget might not be able to get it, since DRF had already allocated the budget to previous tasks.

Second, DRF, like most scheduling algorithms for compute resources~\cite{hug,carbyne,pisces,tetris,graphene}, assumes these resources are \emph{replenishable}: a resource can grant utility (\ie via CPU cycles, network bandwidth) indefinitely.
For instance, if multiple pipelines need to time-share a CPU core, prior work assumes that
if a pipeline was assigned to the core in time interval $T_1$, the core will naturally be
available for other pipeline in time intervals $T_2$, $T_3$, \etc and provide them
with the same amount of CPU cycles per time slot.
In contrast, an individual \privacyresource is a {\bf \em non-replenishable resource}. If a
pipeline is assigned a budget for a particular \privacyresource, that budget is consumed forever, and there may not be sufficient budget remaining for another pipeline in that particular block.
    {\em Dynamic DRF}~\cite{dynamicdrf}, a more recent extension of DRF, considers both dynamic settings and non-replenishable resources.  Unfortunately, Dynamic DRF has its own limitation, as follows.

Third, as discussed in \S\ref{sec:architecture:kubernetes-kubeflow-distinctions}, \sysname adopts an {\bf \em all-or-nothing semantic}: a pipeline is either allocated all of its demanded budget, or none at all.
Therefore, pipelines have an all-or-nothing utility function, where they can only be scheduled (with a utility of 1) if
\emph{their entire demand vector is allocated}, otherwise their utility is 0.
Once a pipeline is allocated its entire demand vector, it leaves the system.
Having an all-or-nothing utility function departs from both Dynamic DRF and DRF, which assume compute resources with continuous utility. In fact, an all-or-nothing utility function would break the Pareto efficiency of Dynamic DRF and DRF alike, which allocate resources proportionally based on demand (see \S\ref{sec:related-work}).

\subsection{DPF}
\label{sec:dpf-design}

Due to the {\bf \em dynamic} arrival of pipelines and the {\bf \em non-replenishable} nature of \privacyresources, we need
to {\em gradually unlock} privacy budget as pipelines arrive to the system, in order to award those pipelines their fair share.
Therefore, we need to define a more constrained notion of a {\em fair share} that divides the budget of \privacyresources over some particular number of pipelines, or a particular time period.
This section presents a version of DPF that defines a fair share over the \emph{first $N$ pipelines that select particular \privacyresources}, and provides formal fairness guarantees {\em for those first $N$ pipelines}.
For any subsequent pipelines (after the first $N$) that request a budget for those particular blocks,
\sysname will {\em not} guarantee them a fair share, but will make a best-effort to schedule them with leftover budget.
\S\ref{sec:DPF-extensions} discusses a version of DPF that instead of dividing resources by pipelines, divides
resources by time intervals, and has weaker fairness guarantees.
In both cases we ensure that DPF schedules budget {\bf \em all-or-nothing}, so that no budget is wasted on tasks that will not end up being scheduled, thus violating Pareto efficiency.

Algorithm~\ref{alg:dpf} gives pseudocode for DPF.
When a new block $j$ is created ($\textproc{OnDataBlockCreation}$), its per-block budget, $\epsilon^G_j$, is determined by the fixed global privacy budget $\epsilon^G$.
To ensure that the first $N$ tasks that request $j$ get their fair share, $j$'s budget is initially completely {\em locked} ($\epsilon^U_j=0$).

Recall that each pipeline in \sysname has in its privacy claim a privacy demand vector, $d$, whose entries represent the epsilon demand for the \privacyresources matching the claim's selector.
We define the {\em privacy budget fair share} of each \privacyresource $j$ as: $\epsilon^{FS}_j=\epsilon^G_j / N$.
DPF guarantees the fair share of a given \privacyresource $j$ to the first $N$ pipelines that arrive to the system that have a non-zero demand for $j$.

We unlock the budget as pipelines arrive (function $\textproc{OnPipelineArrival}$): a new pipeline $i$ that requests budget from a particular block $j$ unlocks $\epsilon^{FS}_j$ of that block's budget, up until all the block's budget is unlocked.
The scheduler's responsibility is to allocate the total unlocked budget ($\epsilon^U$) among the different pipelines.

\begin{algorithm}[!t]
    \caption{{\bf DPF} (max-min fairness for first $N$ pipelines).}
    \begin{algorithmic}
        \State {\color{gray} \# Config.: ($\epsilon^G, \delta^G$) global DP guarantee to enforce.}
        \Function{onDataBlockCreation}{block index $j$}
        \State $\epsilon^G_j \gets \epsilon^G, \epsilon^U_j\gets 0, \epsilon^A_j \gets 0, \epsilon^C_j \gets 0$
        \EndFunction
        \Function{onPipelineArrival}{demand vector $d_{i}$}
        \For{$\forall j: d_{i,j}>0$}
        \State $\epsilon^U_j \gets \min(\epsilon^G_j, \ \epsilon^U_j + \frac{\epsilon^G_j}{N})$
        \EndFor
        \EndFunction
        \Function{onSchedulerTimer}{waiting pipelines wp}
        \State $\textrm{sorted\_pipelines} \gets$ sortBy(\textproc{DominantShare}, wp)
        \For{$i$ in  $\textrm{sorted\_pipelines}$}
        \If{\Call{CanRun}{$d_i$}}
        \State \Call{Allocate}{$d_i$}
        \State Run task $i$, which either consumes $d_{i,j}$ (moving
        \State it to $\epsilon^C_j$) or releases it (moving it back to $\epsilon^U_j$).
        \EndIf
        \EndFor
        \EndFunction
        \Function{DominantShare}{demand vector $d_i$}
        \State return $\max_{j: \ d_{i,j} > 0} \ \frac{d_{i,j}}{\epsilon^G_j}$
        \EndFunction
        \Function{CanRun}{demand vector $d_i$}
        \State return $\forall j: \ d_{i,j} \leq \epsilon^U_j$
        \EndFunction
        \Function{Allocate}{demand vector $d_i$}
        \For{$\forall j$}
        \State $\epsilon^U_j \gets \epsilon^U_j - d_{i,j}$
        \State $\epsilon^A_j \gets \epsilon^A_j + d_{i,j}$
        \EndFor
        \EndFunction
    \end{algorithmic}
    \label{alg:dpf}
\end{algorithm}

To determine which pipeline gets scheduled first, the scheduler maintains a sorted list of the waiting pipelines,
based on their \emph{dominant \privacyresource share}.
This is defined as the maximum demand within each pipeline's demand vector:
\vspace{-.2cm}
\begin{equation}
    \label{eq:dominantshare}
    \textrm{DominantShare}_i = \max_{j} \frac{d_{i,j}}{\epsilon^G_j} ,
    \vspace{-.2cm}
\end{equation}
where $d_{i,j}$ is the demand for block $j$ of pipeline $i$ and $\epsilon^G_j$ is the total budget of \privacyresource $j$.
The scheduler sorts pipelines by their dominant \privacyresource share, with the smallest
share ranked first (function $\textproc{OnSchedulerTimer}$).
If there are one or more pipelines that have the same
dominant \privacyresource share, DPF will sort them by taking the smallest of the
second-most dominant \privacyresource share of each pipeline,
followed by the smallest third-most dominant share, \etc

DPF tries to allocate pipelines based on their order in the list. It tries to allocate \emph{all of the demanded privacy budget vector of the pipeline at once}. If it cannot allocate the pipeline fully (function $\textproc{CanRun}$ returns false), then it moves to the next one in the list, until it reaches the end of the list.

\begin{figure}[t]
    \centering
    \includegraphics[width=0.8\columnwidth]{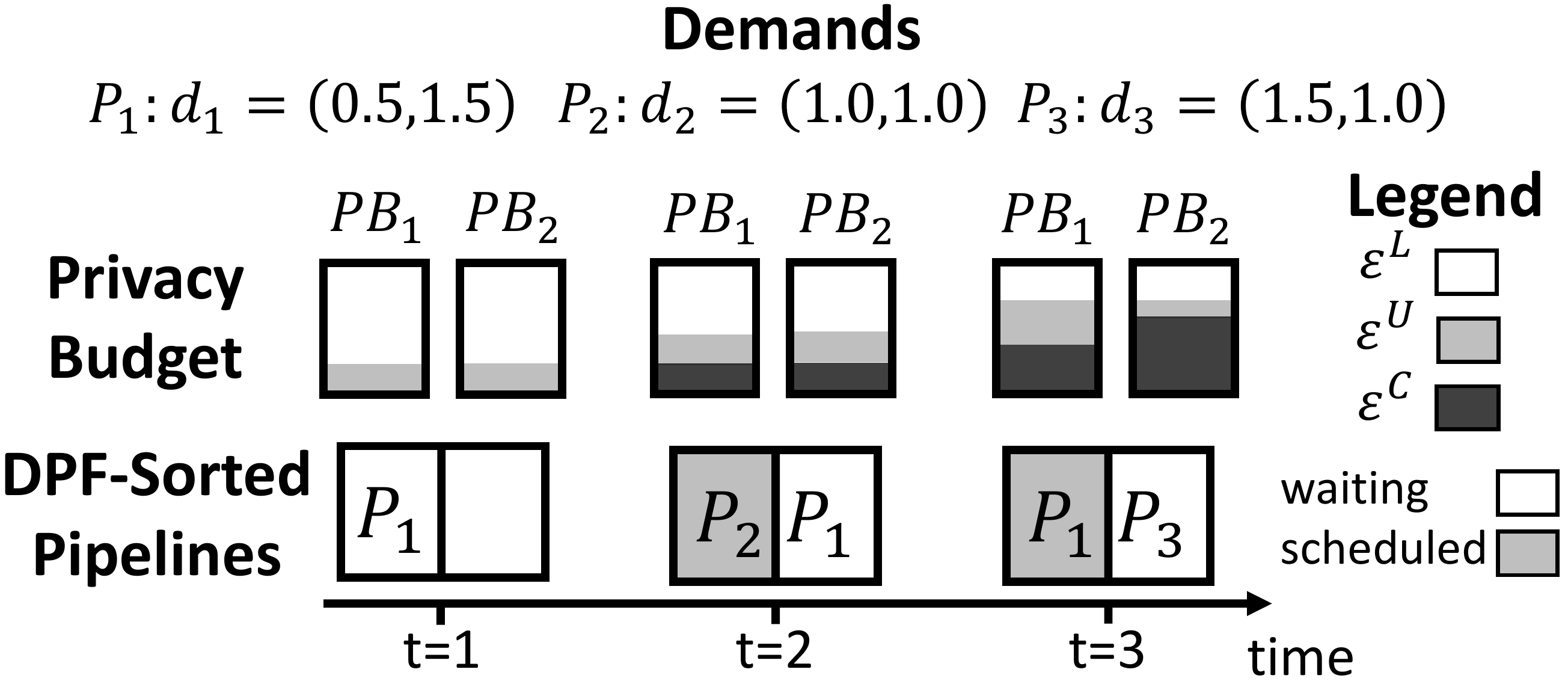}
    \caption{\footnotesize {\bf DPF example.} DPF is scheduling three pipelines ($P_1, P_2, P_3$) over two \privacyresources ($PB_1$, $PB_2$), over time. Shows the state of DPF's sorted list, and what portion of each \privacyresource is locked ($\epsilon^L$), unlocked ($\epsilon^U$), and consumed ($\epsilon^C$).  Assumes budget is consumed instantaneously ($\epsilon^A=0$).}
    \label{fig:dpf_example}
\end{figure}


\paragraph{Example.}
\F\ref{fig:dpf_example} shows an example run of DPF with three pipelines and two \privacyresources.
Suppose the fair share ($\epsilon^{FS}$) of each block is equal to 1.
Pipeline 1 ($P_1$) arrives at $t=1$, then $P_2$ and $P_3$ at each time unit.
The demand vector of $P_1$ is $d_1=(0.5, 1.5)$, while the vector of $P_2$ is $d_2=(1.0, 1.0)$ and $P_3$'s demand is $d_3=(1.5, 1.0)$.
The bottom of the figure depicts the state of of DPF's sorted list at each time unit, where the shaded
pipeline in the list is the one that is scheduled at that time unit, while the unshaded one remains waiting.

When $P_1$ arrives it unlocks a privacy budget of 1 in each block. Since it is the only pipeline in the system
(and therefore has the minimum dominant resource), the scheduler tries to allocate it a budget.
However it is unable to do so, since $P_1$ requires a budget of 1.5 from $PB_2$ but only 1 is unlocked.

When $P_2$ arrives, more budget is unlocked. The dominant resource of $P_1$ is then the second block (with a demand of 1.5)
and the dominant resource of $P_2$ is either block 1 or 2, each of which has a share of 1.
Therefore, the scheduler tries to allocate budget to $P_2$, and does so successfully. It then tries to allocate budget to $P_1$, but is
unable to (since there is only a budget of 1 left in $PB_2$). $P_1$ will have to keep waiting.
When $P_3$ arrives, its dominant share is for block 1 (1.5), while the dominant share for $P_1$ is block 2 (1.5).
Since their dominant share is the same, DPF orders them based on their second highest share, which is 0.5 for $P_1$ and 1.5 for $P_2$.
Therefore, the scheduler allocates the budget for $P_1$. $P_3$ must wait, since the remaining unlocked budget for block 2 is only 0.5.

\subsection{DPF Analysis}
\label{sec:analysis}

We prove four properties of DPF: {\em sharing incentive}, {\em strategy-proofness}, {\em dynamic envy-freeness}, and {\em Pareto efficiency}.
We use the same definitions for these properties defined for dynamic environments based on Kash, et.al.~\cite{dynamicdrf}.

\begin{definition}[{\em fair demand pipeline}]
    \label{def:fair-demand}
    A fair demand pipeline has two properties:
    (a) the pipeline is within the first N pipelines that requested some budget for all its requested blocks,
    and (b) its demand for each one of the blocks is smaller or equal to the fair share (\ie for pipeline $i$, $\forall j: d_{i,j} \leq \epsilon^{FS}_j$).
\end{definition}

\begin{theorem}[{\em sharing incentive}]
    \label{thm:sharing-incentive}
    A fair demand pipeline is granted immediately.
\end{theorem}
\begin{compactproof}
    Consider a fair demand pipeline $i$ with demand $d_i$.  We proceed by induction over the number of waiting pipelines.
        {\em Base case:} no waiting pipelines.  $d_{i,j} > 0 \Rightarrow \epsilon^{FS}_j \leq \epsilon^U_j$, since $\epsilon^{FS}_j$ is unlocked by $d_i$. $d_i$ is fair so $d_{i,j} \leq \epsilon^{FS}_j \leq \epsilon^U_j$. The pipeline is granted, and no fair pipeline is waiting.
        {\em Induction step:} Consider any waiting pipeline $k$ with demand $d_k$ and dominant share $\textrm{DominantShare}_k$.
    By the induction assumption no fair pipeline is waiting, so $\textrm{DominantShare}_k > \epsilon^{FS}_j \geq \textrm{DominantShare}_i$.
    As before, $d_{i,j} > 0 \Rightarrow d_{i,j} \leq \epsilon^{FS}_j \leq \epsilon^U_j$, and $d_i$ can be granted.
    $d_i$ is ordered first so it is granted.
\end{compactproof}

\begin{theorem}[{\em strategy-proofness}]
    \label{thm:strategy-proofness}
    A pipeline has no incentive to misreport its demand.
\end{theorem}
\begin{compactproof}
    A pipeline has no incentive to ask for more budget than its real demand, because: (a) its utility would
    not increase if it obtains more budget than it needs, (b) its dominant share will be greater or equal so it can only become less likely to get scheduled.
    A pipeline also has no incentive to ask for less budget than its real demand, because its utility will drop
    to zero if it is not allocated its demanded budget.
\end{compactproof}

\begin{theorem}[{\em dynamic envy-freeness}]
    \label{thm:dynamic-envy-freeness}
    A pipeline present at time $t$ cannot envy the allocation of another pipeline present at time $t$, except if their $\textrm{DominantShare}$s are identical.
\end{theorem}
\begin{compactproof}
    Consider pipeline $i$. There are two cases.
    Case 1: $i$ was granted. Its utility cannot improve due to all-or-nothing utility, there is no envy.
    Case 2: $i$ is waiting. Consider any pipeline $j$ that $i$ envies and is non identical ($i$ and $j$ are strictly ordered by DPF).
    We show by contradiction that $j$ was granted before $i$ entered the system.  Suppose that was not the case.
    When $j$ was granted: either $\textrm{DominantShare}_j < \textrm{DominantShare}_i$ and $j$ could be granted; or $\textrm{DominantShare}_j > \textrm{DominantShare}_i$ but $i$ could not be granted while $j$ could.
    In both bases $i$ cannot be granted from $j$'s allocation, which would give $i$ a utility of zero. $i$ cannot envy $j$, which is a contradiction.
\end{compactproof}

\begin{theorem}[{\em Pareto efficiency}]
    \label{thm:Pareto-efficiency}
    No allocation from unlocked budget can increase a pipeline's utility without decreasing another pipeline's utility.
\end{theorem}
\begin{compactproof}
    Consider pipeline $i$.
    If $d_i$ was already allocated, its utility cannot improve due to all-or-nothing utility.
    If $i$ is waiting, it cannot be allocated from unlocked budget as DPF grants pipelines until no pipeline can be allocated.
    Allocating $d_i$ would require extra budget, which can only come from another allocated pipeline. Since each allocated pipeline has exactly its requested budget this would decrease its utility from one to zero, which is not Pareto-improving.
\end{compactproof}




\subsection{Best-effort Scheduling for Higher Demands}

While DPF only guarantees immediate allocation for fair demand pipelines,
the algorithm has a best-effort approach to schedule pipelines that do not have a fair demand.
There are two scenarios where pipelines do not have a fair demand.
First, a pipeline's demand may be higher than its fair share for at least one block.
From Theorem~\ref{thm:sharing-incentive}, fair demand pipelines always get immediately scheduled. Therefore, if there is any
leftover unallocated budget after a fair demand pipeline gets scheduled, that
budget can be used to schedule pipelines with higher demands.
This budget will not be needed by any future fair demand pipeline, since they
unlock a budget equal to the fair share.
In \F\ref{fig:dpf_example}, even though pipeline 1 has a higher demand than its
fair share for block 1, it still gets scheduled.
Second, for the same reason, DPF can safely schedule pipelines that are not
among the first $N$ to request budget from some blocks, if there is
leftover unallocated budget in those blocks.

\subsection{Scheduling Compute Alongside Privacy}
\label{sec:otherresources}

DPF only schedules \privacyresources. However, a pipeline will also need computing resource.
Currently, our \sysname prototype implements two schedulers: the privacy scheduler (based on DPF) schedules \privacyresources to private pipelines. The default Kubernetes scheduler schedules traditional computing resources for non-private pipelines, and for private pipelines that have been allocated their privacy budget.
DPF's game theoretic properties hold {\em if} the system is bottlenecked by privacy budget, rather than computing resources.
We leave open the problem of scheduling privacy together with computing resources while guaranteeing game theoretic properties.

\vspace{-0.3cm}
\section{DPF Extensions}                      
\vspace{-0.3cm}
\label{sec:dpf-extensions}
\label{sec:DPF-extensions}
We have focused so far on the core version of DPF that unlocks budget based on pipeline arrival, and uses basic DP composition and Event DP.
We consider three extensions of DPF to address limitations of this core version: unlocking budget by time, using a stronger DP composition (R\'enyi) and stronger DP semantics (User and User-Time DP).
%
%
%

\subsection{Time-based DPF}
\label{sec:dpf-t}

\begin{algorithm}[!t]
	\caption{{\bf DPF-T} (shows what changes in Alg.~\ref{alg:dpf}).}
	\begin{algorithmic}
		\State {\color{gray} \# Replace \textproc{onPipelineArrival} with:}
		\Function{onPrivacyUnlockTimer}{data lifetime L}
		\For{$\forall j$}
		\State $\epsilon^U_j \gets \min(\epsilon^G_j, \ \epsilon^U_j + \frac{\epsilon^G_j}{L})$
		\EndFor
		\EndFunction
	\end{algorithmic}
	\label{alg:dpf-t}
\end{algorithm}

Gradually unlocking privacy budget is key to dealing with a non-replenishable resource and a dynamic workload.
The preceding DPF algorithm unlocks $\epsilon^{FS}_j$ for each requested block $j$, whenever a new pipeline arrives.
We also define a version of DPF that unlocks budget {\em over time}, regardless of workload.
Many organizations already enforce an expiration period, $L$, for collected data.
In time-based DPF (Algorithm~\ref{alg:dpf-t}), each block gradually unlocks its budget over its lifetime $L$, and the fair share
is defined as $\epsilon^{FS}_j=\frac{t}{L}\epsilon^G_j$, where $t$ is the interval of time at which \privacyresource budgets are unlocked.
The advantage of this version is the budget unlocking is predictable and independent of the pipeline arrival patterns. Moreover, by pacing budget unlocking over the data's lifetime, we ensure that the data will have DP budget remaining while still accessible.

Unfortunately, time-based DPF does not guarantee the sharing incentive.
A fair-share pipeline may overlap with many other, smaller pipelines that are ordered first and consume budget when it becomes available, forcing it to wait longer than $t$ or even never be granted.

However, the other three properties are guaranteed by this policy.
We briefly sketch out the proofs for each.
Strategy-proofness is guaranteed because there is no advantage in demanding more than the real demand,
since the pipeline will need to wait longer for the budget to be unlocked.
Envy-freeness is guaranteed for the same reason as in the base version of DPF.
At any given time DPF will prioritize the pipeline with minimum dominant \privacyresource, so a pipeline with a higher dominant resource can only be scheduled earlier than another pipeline by being granted before the other pipeline arrives.
Finally, Pareto efficiency is guaranteed by the combination of all-or-nothing utility and allocation.

\subsection{DPF with R\'enyi DP}
\label{sec:dpf-extensions:renyi}

R\'enyi DP~\cite{8049725} is an alternative DP definition that is stronger than $(\epsilon, \delta)$-DP for $\delta \in (0,1]$ (in the sense that R\'enyi DP always implies $(\epsilon, \delta)$-DP but the converse is not true) and is weaker than $(\epsilon, 0)$-DP ($(\epsilon, 0)$-DP always implies R\'enyi DP).
The great benefit of R\'enyi DP is that it permits convenient composition of multiple mechanisms that scales much better than the basic composition we have been assuming so far.
We thus believe it is important for any globally DP system to support R\'enyi DP, and for this reason we describe our integration of it in \sysname.
However, the definition and formulas of R\'enyi DP are more complex than those of $(\epsilon, \delta)$-DP, so we will not attempt to detail them here.
Instead, we include a R\'enyi DP primer in Appendix~\ref{apx:rdp} and only state here a few facts needed to understand our method.

\heading{R\'enyi DP Facts.}
As described in \S\ref{sec:dp}, DP in general upper bounds the change in the output distribution of a randomized algorithm that can be triggered by a small change in its input.
Making $\delta=0$ in the DP definition in \S\ref{sec:dp}, we see that $(\epsilon, 0)$-DP puts a {\em multiplicative bound} on the change in the output distribution: $\forall \ES . \frac{P(\M(\D) \in \ES)}{P(\M(\D') \in \ES)} \leq e^\epsilon$.
$(\epsilon, \delta)$-DP loosens this multiplicative bound with an {\em additive factor}, $\delta$.
In contrast to these definitions, R\'enyi DP puts an upper bound on the {\em R\'enyi divergence}, a particular measure of distance between the output distributions: $\textrm{R\'enyiDivergence}_\alpha(\M(\D),\M(\D')) \leq \epsilon$.
We state three facts about this distance and R\'enyi DP.

First, R\'enyi divergence is parameterized by a parameter, $\alpha > 1$, hence R\'enyi DP is expressed in terms of two parameters: $(\alpha, \epsilon)$.
Second, for every value of $\alpha$, there is a direct translation from R\'enyi DP to $(\epsilon, \delta)$-DP.
The formula is: $(\alpha, \epsilon-\frac{log(1/\delta)}{\alpha-1})$-R\'enyi DP implies $(\epsilon, \delta)$-DP for any value of $\epsilon>0$, $\delta \in (0,1]$, and $\alpha>1$.
Also, $(\infty,\epsilon)$-R\'enyi DP is equivalent to $(\epsilon, 0)$-DP for any value of $\epsilon>0$.
Thus, the $\alpha$ parameter can be seen as adding a spectrum between pure $(\epsilon, 0)$-DP and $(\epsilon, \delta)$-DP; from any point on that spectrum, one can reconstruct back the traditional $(\epsilon, \delta)$-DP guarantee.
For our work, this means that \sysname can use R\'enyi DP internally while exposing the same $(\epsilon^G, \delta^G)$-DP guarantee externally.

Third, R\'enyi DP allows tighter analysis of the privacy loss from multiple mechanisms.
For example, the scale of the Gaussian distribution required to achieve $(\epsilon, \delta)$-DP depends linearly on $1/\epsilon$.
The scale of the Gaussian required to achieve $(\alpha, \epsilon)$-R\'enyi DP depends on $1/\sqrt{\epsilon}$ (and on $\alpha$).
In traditional DP, when composing (summing the $\epsilon$'s of) $k$ Gaussian mechanisms with the same scale, $\sigma$, the composite mechanism is equivalent to a Gaussian mechanism with $\sigma/k$ scale, so it's ``k times less private.''
But in R\'enyi DP, when composing the same $k$ Gaussian mechanisms, the composite mechanism is equivalent to a Gaussian mechanism with $\sigma/\sqrt{k}$ scale, so it's just ``$\sqrt{k}$ less private.''
Thus, R\'enyi DP scales much better in the number of computations and should enable more pipelines to share the global budget.

\paragraph{DPF with R\'enyi DP.}
Our goal is to take advantage of R\'enyi composition without sacrificing DPF's game-theoretical properties.
One option is to pick one point in the R\'enyi DP spectrum (one value of $\alpha>1$) and apply DPF as is, internally using R\'enyi to analyze and compose privacy loss, and ultimately translating the R\'enyi guarantee back into traditional DP.
Unfortunately, when composing multiple, heterogeneous mechanisms (think different $\sigma$ for Gaussian) in R\'enyi DP, it is unclear \textit{a priori} which parameter $\alpha$ will ultimately give the best traditional-DP guarantee; this is because both the R\'enyi analysis of privacy loss and the translation to traditional DP depend on $\alpha$, in inverse directions (see Appendix~\ref{apx:rdp} for details).
In \sysname, we thus choose to track a set $A$ of $\alpha>1$ values, and to use one that ultimately gives the best traditional-DP guarantee.
As the R\'enyi DP author shows~\cite{8049725}, and as we observed experimentally, fine-grained choice of values is not important, so we select several values based on recommendations from~\cite{8049725}: $A=\{2,3,4,8,...,32,64\}$.

Algorithm~\ref{alg:dpf-rdp} summarizes the changes DPF requires to support R\'enyi.
For each private block, $j$, \sysname initializes a {\em vector of R\'enyi budgets}, with one entry for each value of $\alpha \in A$, based on the preceding translation formula (function $\textproc{OnDataBlockCreation}$).
Other privacy variables maintained in the block similarly become vectors in $\alpha$ ($\epsilon^U$, $\epsilon^A$, etc.).
Moreover, a pipeline's privacy demand also becomes a vector {\em for each block}: $d_{i,j}(\alpha)$.
In practice, a developer will decide on the mechanism and noise scale to use (e.g. Gaussian mechanism with scale $\sigma$), based on which a library can compute the R\'enyi privacy demand vector for the tracked $\alpha$'s.
When a pipeline is allocated (function $\textproc{Allocate}$), the requested budget is deducted from each block, {\em and for each $\alpha$}.

With these changes, the question becomes how to schedule over the $\alpha$ vectors.
One approach is to treat each $(block, \alpha)$ tuple as a separate resource.
Since DPF already supports multiple resources, its game-theoretical guarantees should hold.
Indeed, this is how we compute the $\textproc{DominantShare}$ under R\'enyi: return the maximum demand over all requested blocks and $\alpha$ orders.
However, treating each $(block, \alpha)$ tuple as a separate resource does not work when deciding if a pipeline $\textproc{CanRun}$.
Indeed, doing so would allocate pipelines only when enough budget is unlocked for {\em all} $\alpha$ values.
However, recall that in R\'enyi DP, {\em any} $\alpha$ with sufficient privacy budget can be translated to an $\epsilon,\delta$-DP guarantee.  Requiring {\em all} to have that would just block progress until the largest $\alpha$ acquires sufficient budget, which removes the benefits of R\'enyi composition.
Instead, we allow allocation of any pipeline in which each requested block has enough unlocked budget $\epsilon^{U}_j(\alpha)$ for {\em any} $\alpha$ (potentially at different $\alpha$ across blocks).

\heading{Analysis.}
Under this behavior, the consumed budget at some $\alpha$ values may be higher than the unlocked budget, and even the global one.
However, for each block $j$ there will always remain one $\alpha$ such that $0 \leq \epsilon^{U}_j(\alpha) \leq \epsilon^{G}_j(\alpha)$.
The global $(\epsilon^G, \delta^G)$-DP guarantee is thus preserved (proof in Appendix~\ref{apx:rdp-dpf-proofs}).
Moreover, DPF's four properties (\S\ref{sec:analysis}) can be proven to hold under the following definition of a fair pipeline: $\forall j, d_{i,j}(\alpha) \leq \epsilon^{FS}_j(\alpha)$, where $\epsilon^{FS}_j(\alpha) = \frac{\epsilon^G_j(\alpha)}{N}$ (proofs in~\cite{privatekube_arxiv}).

\begin{algorithm}[!t]
    \caption{{\bf DPF-R\'enyi} (shows what changes in Alg. \ref{alg:dpf}).}
    \begin{algorithmic}
        \State {\color{gray} \# Config.: ($\epsilon^G, \delta^G$): global DP guarantee to enforce;}
        \State {\color{gray} \# \ \ \ $A$: R\'enyi parameters (default: $\{2,3,4,8,...,64\}$).}
        \Function{onDataBlockCreation}{block index $j$}
        \State $\forall \alpha \in A: \ \epsilon^G_j(\alpha) \gets \epsilon^G - \frac{\log(1/\delta^G)}{\alpha-1}$
        \EndFunction
        \State {\color{gray} \# Either \textproc{onPipelineArrival} or \textproc{onPrivacyUnlock}- \# \textproc{Timer}, modified to unlock budget for each alpha.}
        \Function{DominantShare}{demand vector $d_i(\alpha)$}
        \State return $max_{j:d_{i,j}>0}\max_{\alpha \in A}\frac{d_{i,j}(\alpha)}{\epsilon^G_j(\alpha)}$
        \EndFunction
        \Function{CanRun}{demand vector $d_i(\alpha)$}\label{line:can-run}
        \State return $\forall j: \ \exists \alpha \ \textrm{s.t.} \ d_{i,j}(\alpha) \leq \epsilon^U_j(\alpha)$
        \EndFunction
        \Function{Allocate}{demand vector $d_i(\alpha)$}
        \For{$\forall j \ \textrm{and} \ \forall \alpha \in A$}
        \State $\epsilon^U_j(\alpha) \gets \epsilon^U_j(\alpha) - d_{i,j}(\alpha)$
        \State $\epsilon^A_j(\alpha) \gets \epsilon^A_j(\alpha) + d_{i,j}(\alpha)$
        \EndFor
        \EndFunction
    \end{algorithmic}
    \label{alg:dpf-rdp}
\end{algorithm}


\subsection{Supporting Varied DP Semantics}
\label{sec:varied-dp-semantics}

\begin{figure*}[t!]
    \centering
    \hfill
    \begin{subfigure}{0.29\linewidth}
        \centering
        \includegraphics[width=\linewidth]{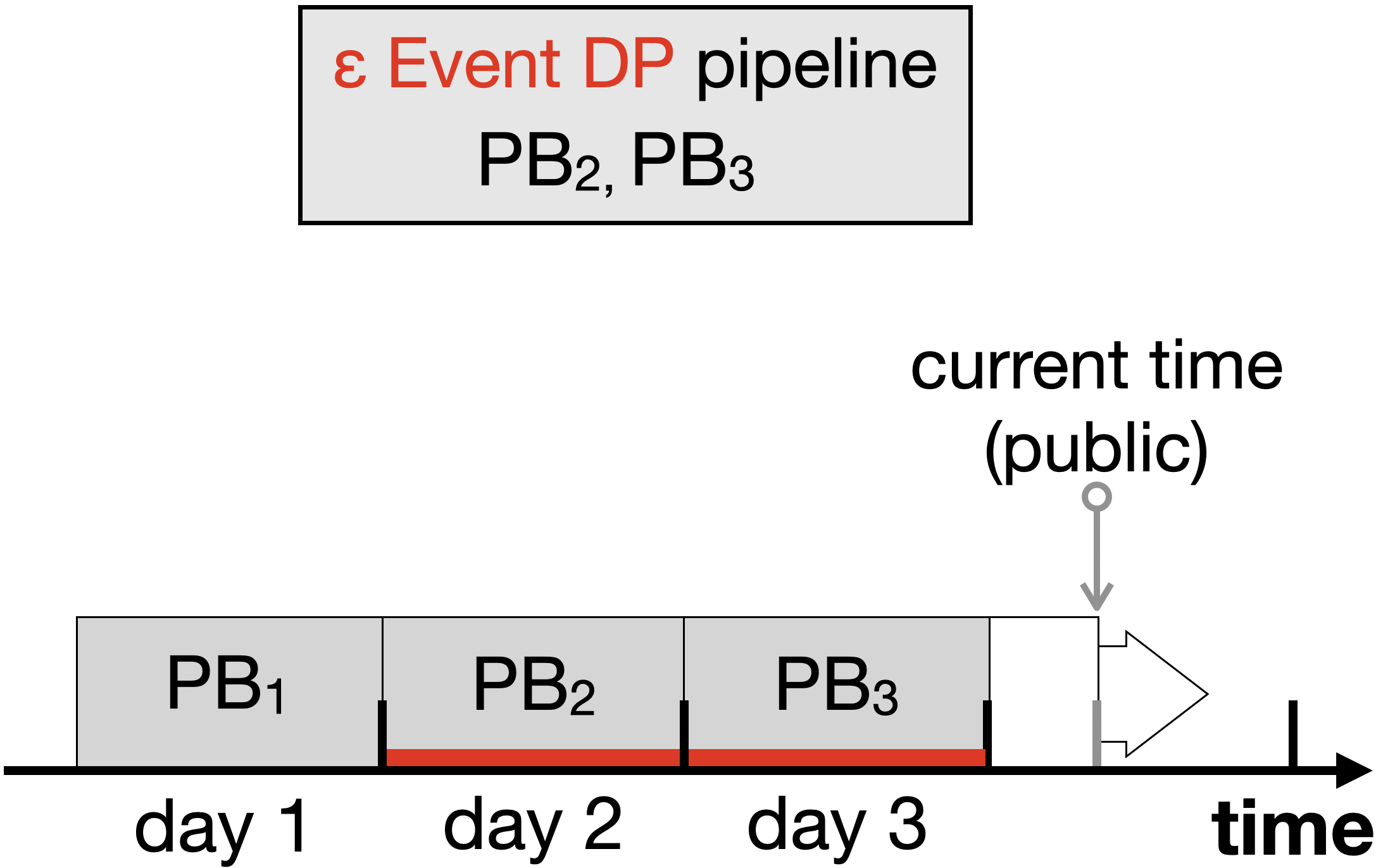}
        \caption{\footnotesize {\bf Event DP.}
            Same as Sage~\cite{sage}.
        }
        \label{f:time_splitting}
    \end{subfigure}
    \hfill
    \begin{subfigure}{0.29\linewidth}
        \centering
        \includegraphics[width=\linewidth]{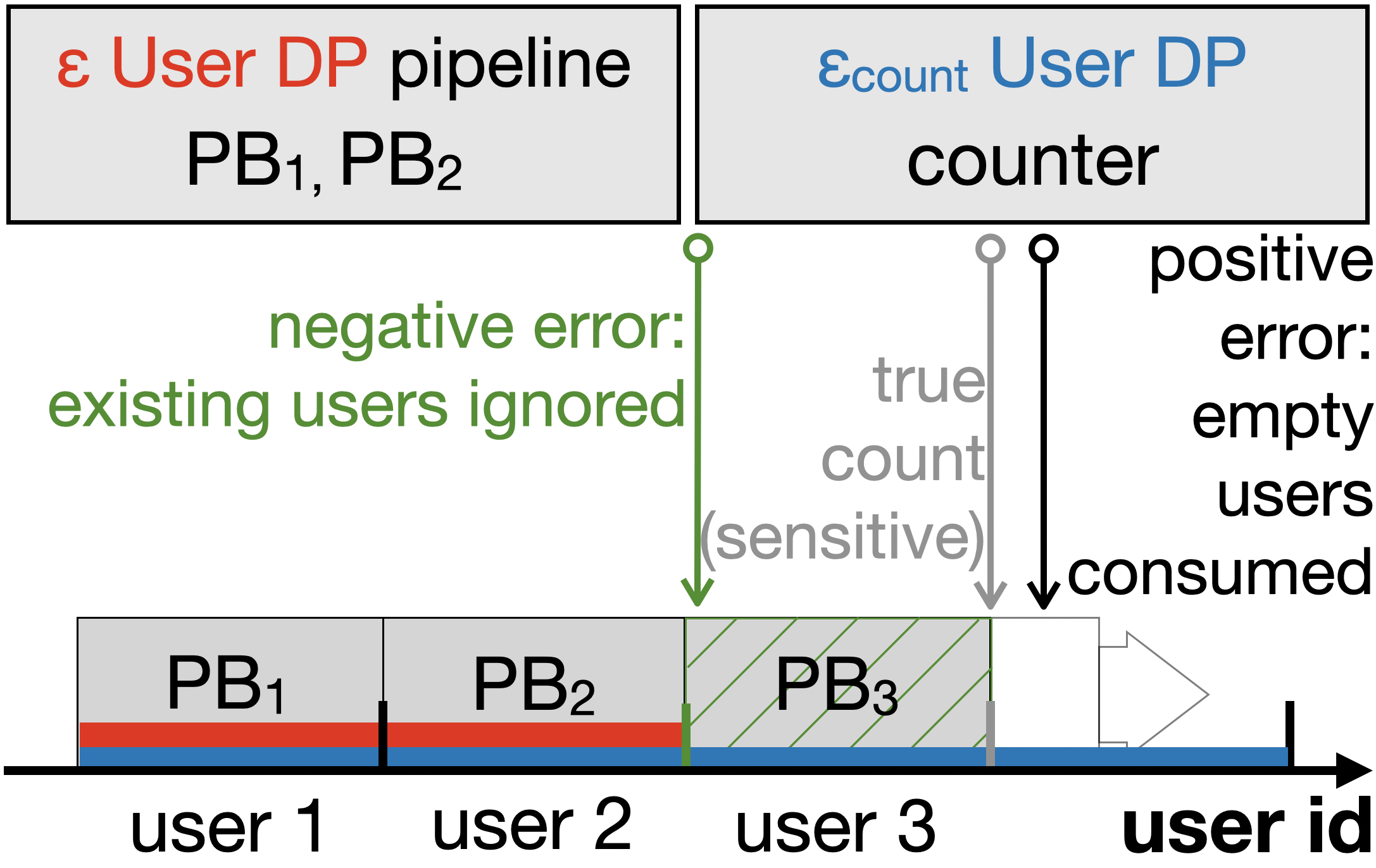}
        \caption{\footnotesize {\bf User DP.}
            New in \sysname.
        }
        \label{f:user_splitting}
    \end{subfigure}
    \hfill
    \begin{subfigure}{0.331\linewidth}
        \centering
        \includegraphics[width=\linewidth]{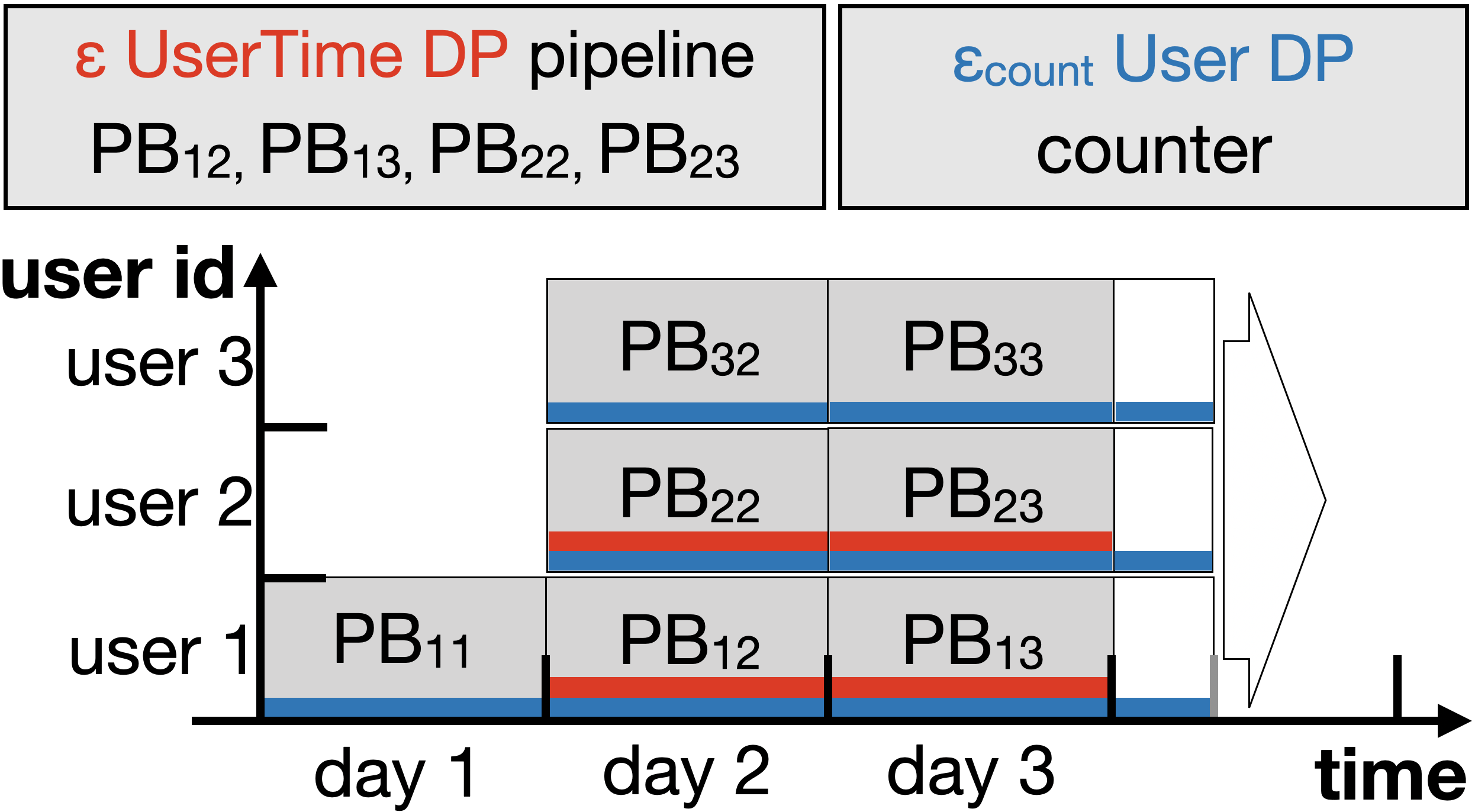}
        \caption{\footnotesize {\bf User-Time DP.}
            New in \sysname.
        }
        \label{f:usertime-splitting}
    \end{subfigure}
    \hfill
    \caption{\footnotesize {\bf \sysname's support for diverse DP semantics.}
        Shows how the data is split into blocks and how pipelines request them.
        Light-gray private blocks can be requested by pipelines, white blocks are in-progress.
        A block's area represents its $\epsilon_j^G$ budget. Red portions are consumed by pipelines, blue by counters.
    }
    \label{f:block_splitting}
\end{figure*}

Finally, we detail how we incorporate support for all three DP semantics -- Event, User, and User-Time DP -- in our private block abstraction. To our knowledge, no one has shown how to support all three with one abstraction, and since we believe that a DP system should support diverse semantics, suitable for different cases, we describe here how we do so.

DP conceals a change between neighboring $D$ and $D'$ that are identical with a row added or removed.
This neighboring definition, or what we treat as a row that is added or removed, defines the protection semantic.
In {\em Event DP}, the most common but weakest semantic, $D$ and $D'$ differ in one event (\eg, one click).
DP thus conceals the impact of adding or removing one such event (e.g. yesterday's click on a health related post about a specific condition), but since one user can contribute a large number of such events, important aspects of a user's behavior can still leak though DP computations (e.g. repeated clicks related to said medical condition).
In {\em User DP}, the strongest semantic, neighboring datasets differ by all the data of one user.
User DP conceals the entire contribution of a user regardless of the amount of data (\eg, many clicks in a health app).
This semantic can be challenging to enforce on streams, since users with an exhausted privacy budget cannot contribute to new computations, even if they generate new data.
    {\em User-Time DP} is a middle-ground~\cite{Kifer2020GuidelinesFI}, in which neighboring datasets $D$ and $D'$ differ by the addition or removal of all data from one user in a given time period (\eg, one day). Repeated actions of a user in that time period are protected (\eg, a browsing session with repeated clicks), and newly generated data in the next period can still be used.

\F\ref{f:block_splitting} illustrates how we support all three DP semantics in our private block abstraction.
It requires instantiating two aspects: (1) how data is split into private blocks and (2) how blocks are requested by the pipelines.

\heading{Event DP (\F\ref{f:time_splitting}).}
{\em (1)~Splitting data:} At pre-set time intervals (\eg, a day), the data collected in this interval forms a new private block with a total of $\epsilon^G$ privacy budget.
    {\em (2)~Requesting blocks:} Because time is public, we always know which past blocks have been created and filled with data.
Pipelines registered on \sysname can thus request blocks from a time range of interest without risking consuming budget from an empty block.
In \F\ref{f:time_splitting}, blocks for the first three days are available. The pipeline requests data from the last two days, thereby consuming budget only for those.
This design is identical to Sage~\cite{sage}, which supports {\em only} Event DP.

\heading{User DP (\F\ref{f:user_splitting}).}
{\em (1)~Splitting data:}
Computing on any user's event must consume DP budget for the entire user; time-based splitting is therefore insufficient because a user's clicks can span large time intervals.
Instead, \sysname maintains a private block for each (group of) user id(s) that will ever exist in the system, lazily instantiated.
New data is added to the block responsible for the corresponding user without changing its remaining DP budget, or to a newly created block if this user is new.
For instance, in \F\ref{f:user_splitting}, only the first three users contributed data so far.

    {\em (2)~Requesting blocks:}
This raises a challenge.  Unlike in Event DP, where we know which past blocks have been created and filled with data, in User DP we do not know which users exist in the system at a given time.
Knowing that would leak information about which users join when, violating User DP.
Instead, \sysname maintains a DP counter that estimates, in a user-DP way, the number of users in the system at any time.
The counter is updated periodically (e.g., daily) and consumes a bit of DP budget from every block (in blue on \F\ref{f:user_splitting}).
Since the count is noisy, pipelines requesting user blocks may sometimes overshoot and consume budget from users that do not yet exist (and therefore cannot possibly supply any data).
To avoid consuming budget from empty user blocks, our design has pipelines request user blocks based on a {\em high probability lower-bound of the true count}.  This ensures the true count is under-estimated with high probability, so no empty user is wastefully requested.
Appendix~\ref{apx:counter} gives the specific formulas to obtain this lower bound.

The counter does consume some $\epsilon_{count}$-DP budget, which is a configuration parameter of \sysname, fixed when \sysname is deployed.
The budget is deducted once for each data block, upon the block's creation.
For example, for R\'enyi-DP, $\textproc{OnPrivateBlockCreation}(j)$ initializes j's global R\'enyi budget vector to: $\epsilon_j^G(\alpha) = \epsilon^G - \frac{\log(1/\delta^G)}{\alpha-1} - 2\epsilon_{count}^2 \alpha$, where the last term corresponds to the Renyi consumption of the $\epsilon_{count}$-DP counter.
Since DPF always works from this $\epsilon_j^G(\alpha)$, all DPF properties are preserved.

\heading{User-Time DP (\F\ref{f:usertime-splitting}).}
A middle-ground between Event and User DP, User-Time DP combines both mechanisms.
    {\em (1)~Splitting data:}
Data is split over both user and time; newly collected data is assigned to the block managing the corresponding user and the time range that includes the data creation.
Some of the blocks may be empty (e.g. user 1, day 2), but since no new data can ever be added to them once their timeframe passes, there is no cost to the future of using their DP budget now.
    %
    {\em (2)~Requesting blocks:}
Blocks are requested on both time and a continuous DP counter of the number of users.
The counter works similarly to User DP, except that the first (smallest time) block for a user id is created when the upper-bound of the user counter reaches this user id.
This corresponds to the first time a user may have contributed data.

\vspace{-0.3cm}
\section{Evaluation}                        
\vspace{-0.3cm}
\label{sec:evaluation}
We implemented \sysname on Kubernetes 1.17. Our experiments run on Google Cloud with managed GKE on two pools of CPU (n1-standard8 machines) and GPU (n1-standard8 machines with one Tesla K80 GPU) servers. Each pool is autoscaled by Kubernetes up to a cap of 10 servers per pool.

Our evaluation seeks to answer six questions: 
\begin{denseenum}
    \item[{\bf Q1:}] How does DPF compare to baseline scheduling policies?
    \item[{\bf Q2:}] How do workload characteristics impact DPF?
    \item[{\bf Q3:}] How does R\'enyi DP impact DPF?
    \item[{\bf Q4:}] How does the DP semantic impact model accuracy?
    \item[{\bf Q5:}] How does the DP semantic impact DPF?
    \item[{\bf Q6:}] Does native integration facilitate tool reuse?
\end{denseenum}

We develop two methodologies.
First, we create a simple, controlled {\em microbenchmark} that helps us explore DPF under varied workload characteristics (Q1, Q2, Q3).
Second, we create a {\em macrobenchmark} consisting of multiple ML pipelines trained on Amazon Reviews~\cite{amazon-reviews} to investigate Q1, and Q4-6.

\heading{Metrics and Baselines.}
Across our experiments, we use the following metrics.
    {\em Number of allocated pipelines} is the number of pipelines that were successfully allocated their privacy budget throughout the experiment.
    {\em Scheduling delay} is the time measured from when a pipeline arrives to the point where it is allocated its privacy budget.
    {\em Accuracy} is the percentage of correct classification of a model.

We compare DPF to two baseline scheduling algorithms.
    {\em First-come-first-serve (FCFS)} tries to allocate pipelines by their order of arrival on available privacy budget. All the budget is immediately available to pipelines (\ie unlocked) from the outset.
    {\em  Round robin (RR)} allocates budget evenly among pipelines that are currently in the system. We implement two versions of RR that correspond to the two versions of DPF. The first one unlocks $\epsilon^{FS}_j$ of budget for each pipeline that arrives that demands a block $j$, and the second one unlocks budget in the block over time in proportion to its lifetime.
For example, if the data lifetime is a year, a third of the budget of a block will be released after four months.
This latter policy is similar to the one used by the Sage system~\cite{sage}.

\heading{Evaluation Highlights.}
DPF is able to grant more pipelines than the baselines at the cost of a small delay (Q1), especially over heterogeneous workloads (Q2). R\'enyi DP enables allocation of either many more or much larger pipelines (Q3). Stronger DP semantics require more DP budget and data (Q4), which increases the need for judicious budget allocation as with DPF (Q5). Our native integration enables reuse of existing tooling for privacy resource management, such as using Grafana to monitor privacy consumption (Q6).



\subsection{Microbenchmark (Q1, Q2, Q3)}
\label{sec:evaluation:microbenchmarks}


Our microbenchmarks evaluate the performance of DPF compared to the two baselines.
We assume pipeline arrival follows the Poisson process.
In the single-block experiment, the pipeline arrival rate is 1 per second. We generate two types of pipelines, mice and elephants, split $75\%$ to $25\%$ by default, with respective demands of $\epsilon = 0.01 \epsilon^G$, and $\epsilon = 0.1 \epsilon^G$. 
In the multi-block experiment, blocks are created every 10 seconds. By default, pipeline's demand $\epsilon$ follows the same distribution as single-block. However, it can either request the last block with probability $0.75$, or the last 10 blocks with probability $0.25$, independently of the requested $\epsilon$. We used a load that emphasizes the differences between the policies, where newly arrived pipelines' average demand is $13.5\times$ of the newly generated blocks. This results in the basic composition experiments using an arrival rate of $12.8$ per second, and the R\'enyi experiments using $234.4$ per second. If not allocated, pipelines timeout after 300 seconds.


\subsubsection{DPF Behavior on a Single Block}
\label{sec:evaluation:microbenchmarks:single-block}

We first evaluate the performance of DPF in the simplest possible setup: with a single
\privacyresource. In this case, the demand vector of each pipeline will only contain one item,
and DPF will prioritize the pipeline with the lowest demand. 

\begin{figure}[t]
	\centering
	\begin{subfigure}{0.5\linewidth}
		\centering
		\includegraphics[width=\linewidth]{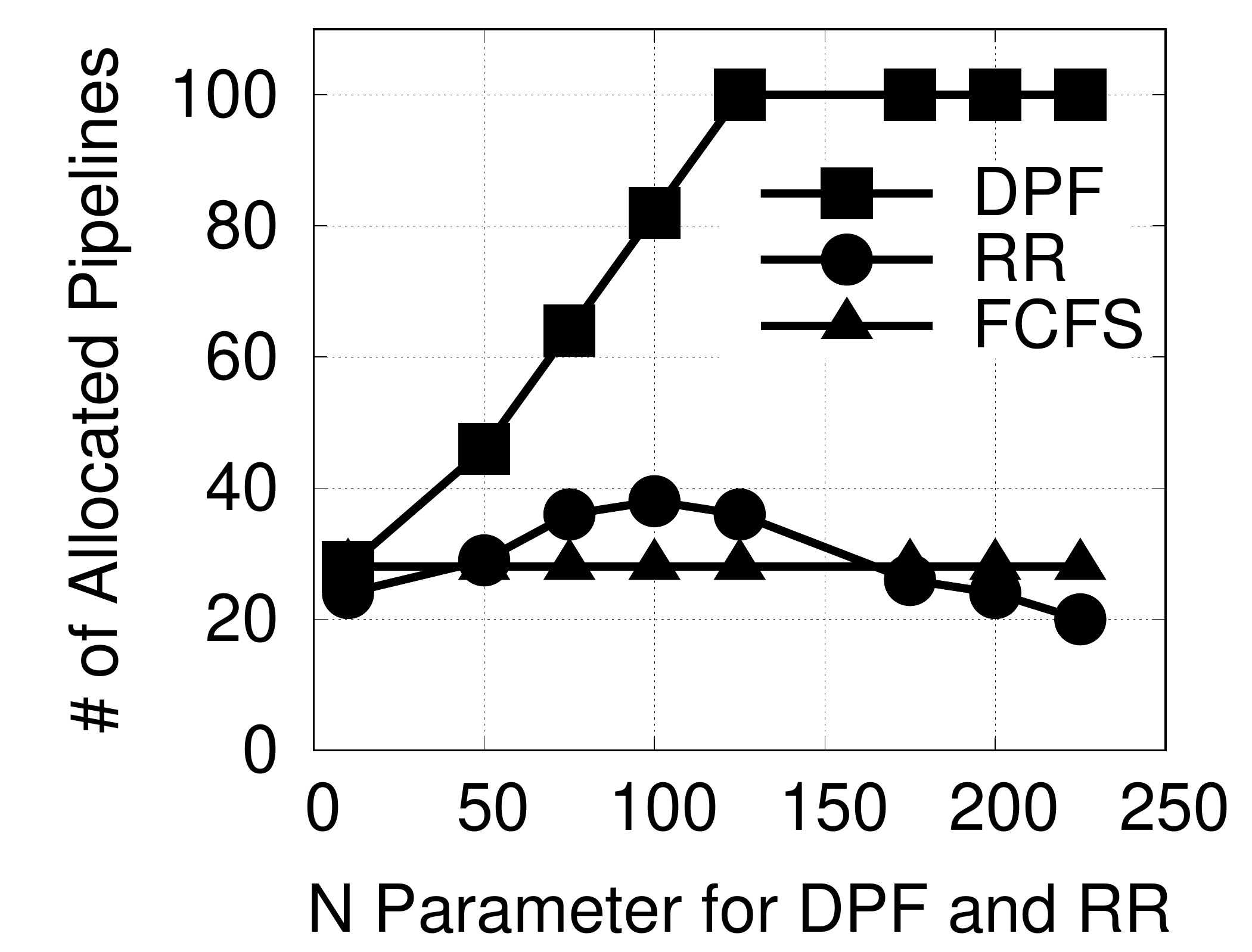}
		\caption{\footnotesize {\bf Number of pipelines allocated.}}
		\label{fig:evaluation:microbenchmark:single-block:completed}
	\end{subfigure}%
	~
	\begin{subfigure}{0.5\linewidth}
		\centering
		\includegraphics[width=\linewidth]{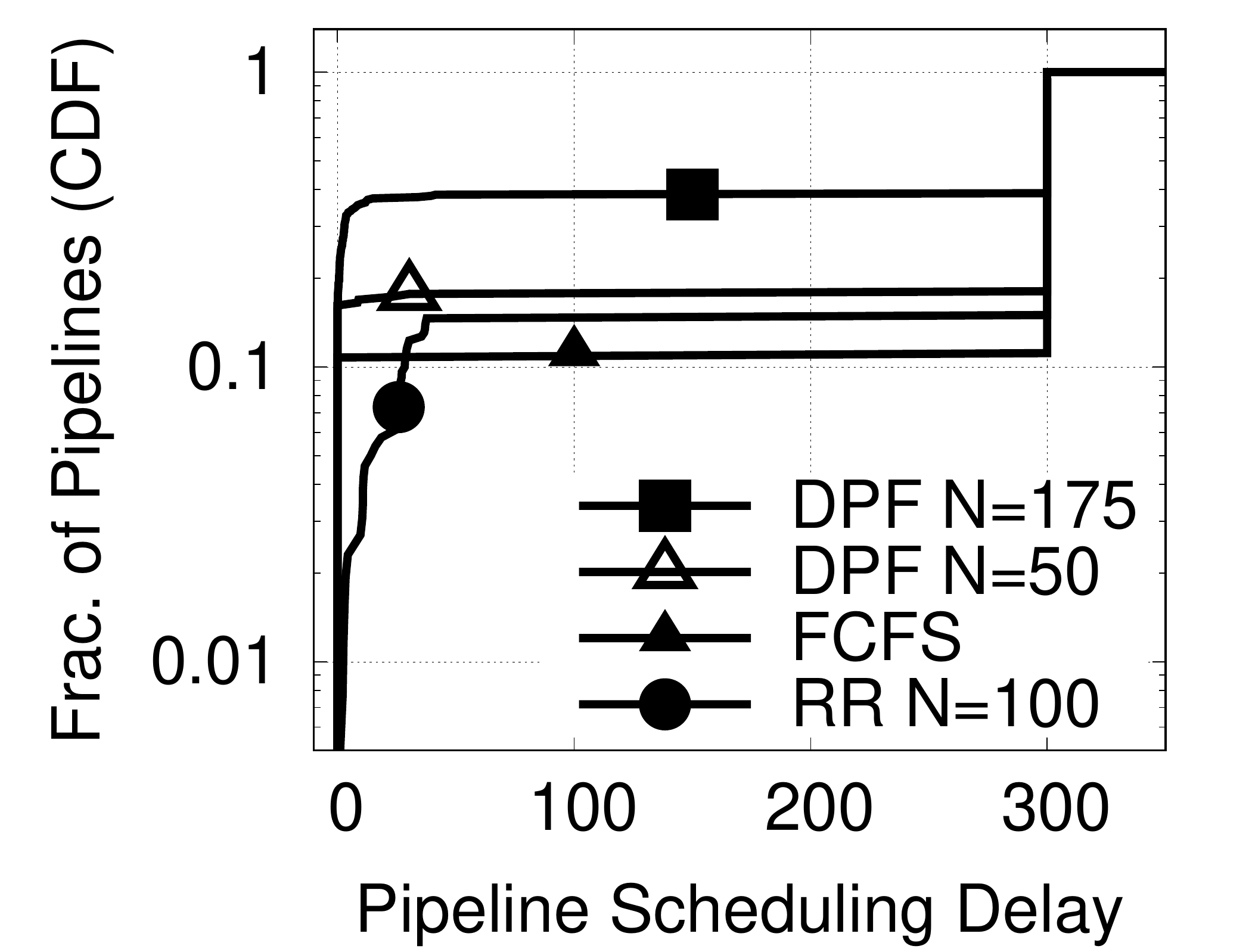}
		\caption{\footnotesize {\bf Scheduling delay.}}
		\label{fig:evaluation:microbenchmark:single-block:wait-time}
	\end{subfigure}
    \caption{\footnotesize {\bf DPF behavior on a single block.}}
	\label{fig:evaluation:microbenchmark:single-block}
\end{figure}

\F\ref{fig:evaluation:microbenchmark:single-block} shows DPF and RR under different
$N$ values, and FCFS.
\F\ref{fig:evaluation:microbenchmark:single-block:completed} shows allocated pipelines.
With FCFS early elephants take away the budget of many mice, only $28$
pipelines are granted.
With RR, a low value of $N$ directly unlocks all DP budget, behaving like FCFS.
When $N$ is high enough to maintain a large number of mice, but low enough to eventually grant them, RR is able to grant up to $38$ pipelines (more than FCFS).
At large $N$ RR's proportional allocation creates multiple partially granted pipelines and only $20$ are granted.
Neither outperforms DPF.
When $N$ is equal to 1, the first pipeline unlocks all the budget and DPF behaves like FCFS.
At higher values of $N$, DPF prefers mice over elephants and a higher number of pipelines get allocated, up to the maximum possible of $100$.
Since DPF never wastes budget on unallocated pipelines it outperforms RR when $N$ is large.

As expected, granting more jobs comes at the cost of increased delay (\F\ref{fig:evaluation:microbenchmark:single-block:wait-time} shows scheduling delay at notable operating points for each policy).
With DPF at $N=50$ some elephants experience scheduling delays before being granted from unlocked budget.
At $N=175$ some mice wait since $\epsilon^{FS}$ is higher than the mice requests, but only mice are granted.

To summarize, DPF is always able to allocate budget to more pipelines than FCFS or RR.
$N$ presents a trade-off between the number of pipelines that are successfully allocated and the scheduling delay the pipelines experience.

\subsubsection{DPF Behavior with Mice Percentage}
\label{sec:evaluation:microbenchmark:mice-percentage}

\begin{figure}[t]
	\centering
	\begin{subfigure}{0.5\linewidth}
		\centering
		\includegraphics[width=\linewidth]{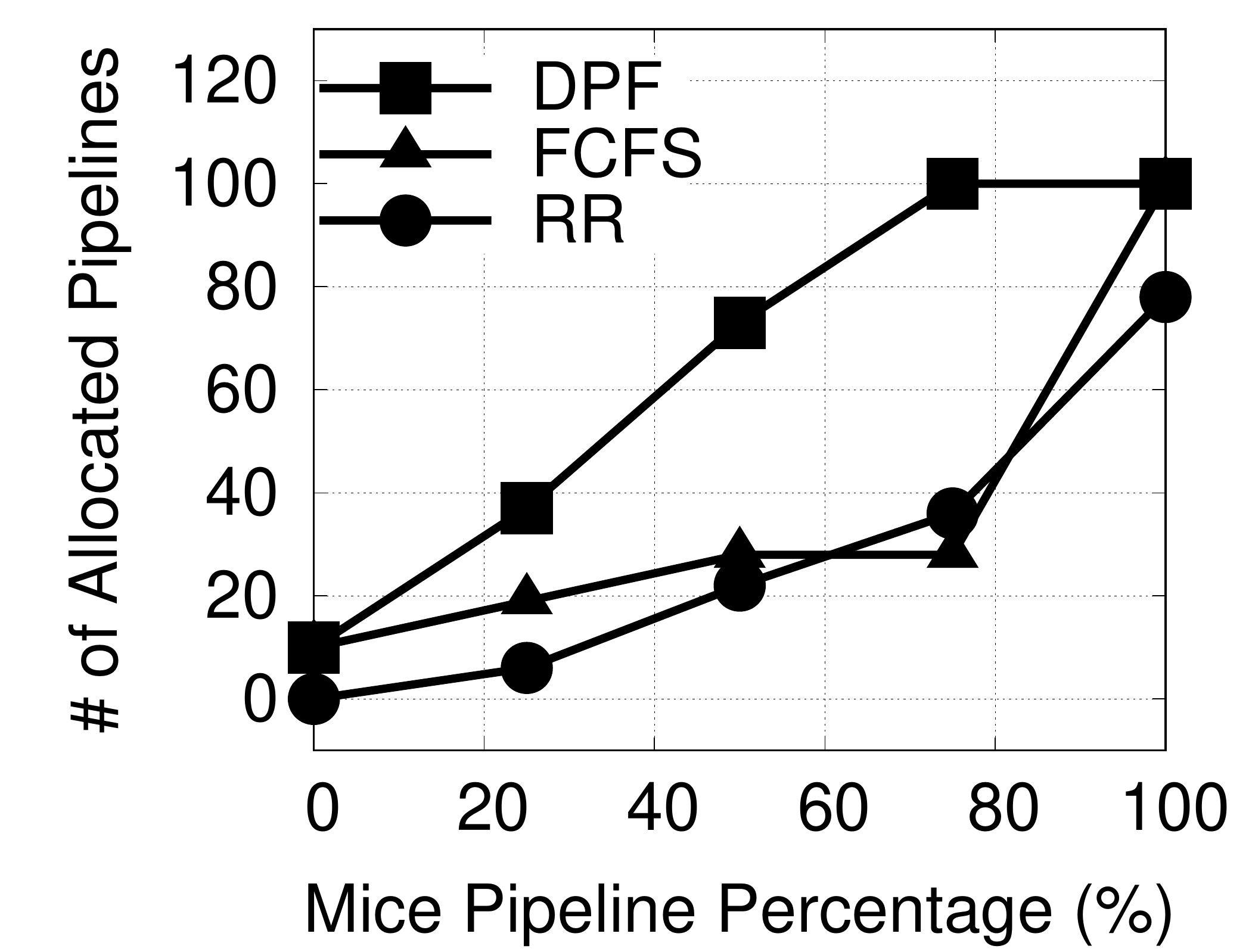}  
		\caption{\footnotesize {\bf Number of pipelines allocated.}}
		\label{fig:evaluation:microbenchmark:mice-percentage:completed}
	\end{subfigure}%
	~
	\begin{subfigure}{0.5\linewidth}
		\centering
		\includegraphics[width=\linewidth]{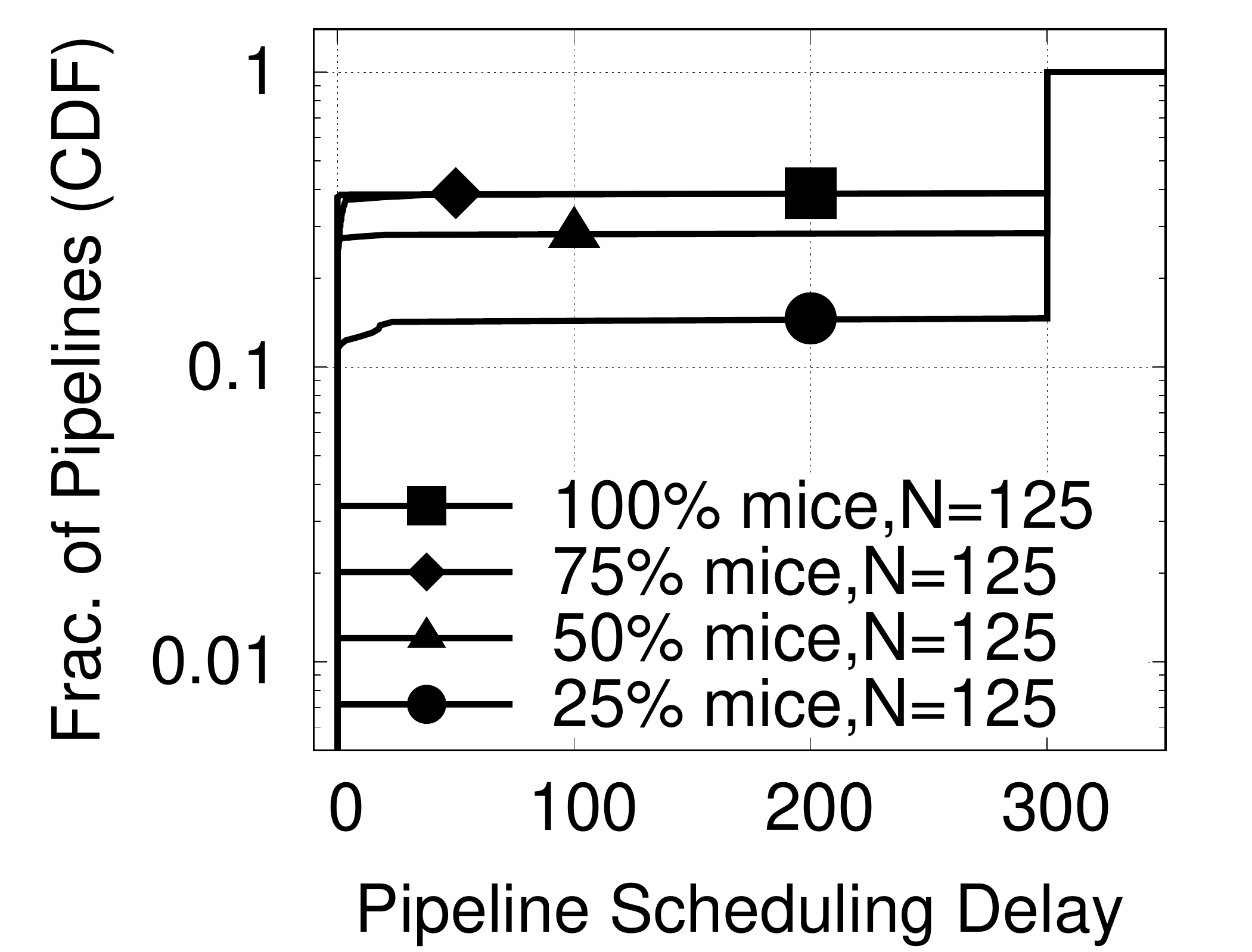}
		\caption{\footnotesize {\bf Scheduling delay.}}
		\label{fig:evaluation:microbenchmark:mice-percentage:wait-time}
	\end{subfigure}
	\caption{\footnotesize {\bf DPF with varied workload mix, single block.} (b) DPF N=125. }
	\label{fig:evaluation:microbenchmark:mice-percentage}
\end{figure}

\F\ref{fig:evaluation:microbenchmark:mice-percentage} compares the three scheduling policies under a variable percentage of mice and elephants.
At either extreme, all pipelines are identical so DPF and FCFS allocate
the same number of pipelines. In this case, the scheduling delay of FCFS is slightly better,
since it always immediately schedules these pipelines. However, when there is a mix of pipelines,
DPF always allocates more pipelines.
RR performance is mixed: for some workloads it is able to allocate slightly more
pipelines than FCFS, since it assigns a higher percentage of budget to mice;
for others it underperforms FCFS, since it wastes budget
on pipelines that are never scheduled.

\subsubsection{DPF Behavior on Multiple Blocks}
\label{sec:evaluation:microbenchmark:multiple-blocks}

\begin{figure}[t]
	\centering
	\begin{subfigure}{0.5\linewidth}
		\centering
		\includegraphics[width=\linewidth]{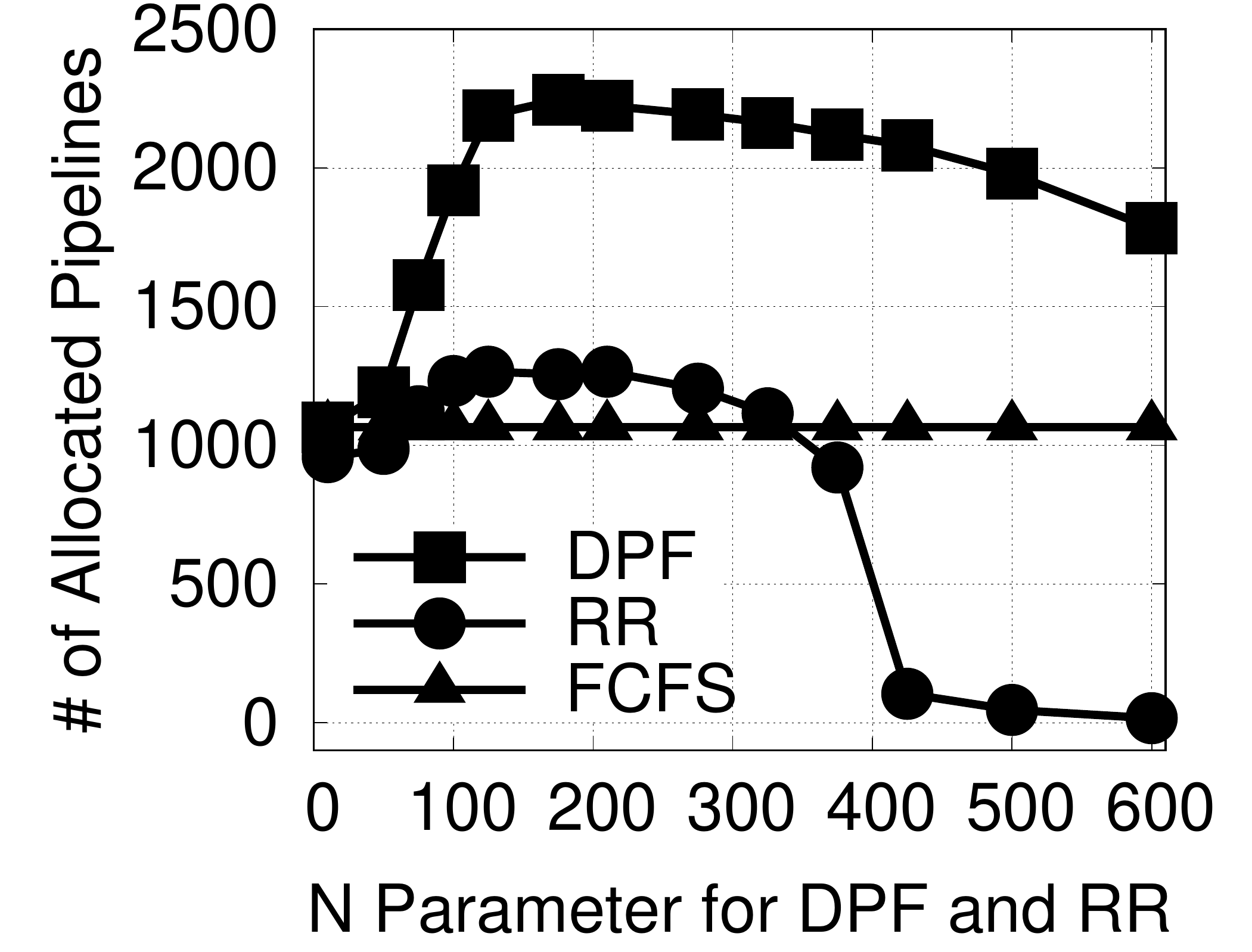}
		\caption{\footnotesize {\bf Number of pipelines allocated.}}
		\label{fig:evaluation:microbenchmark:multiple-blocks:completed}
	\end{subfigure}%
    ~
	\begin{subfigure}{0.5\linewidth}
		\centering
		\includegraphics[width=\linewidth]{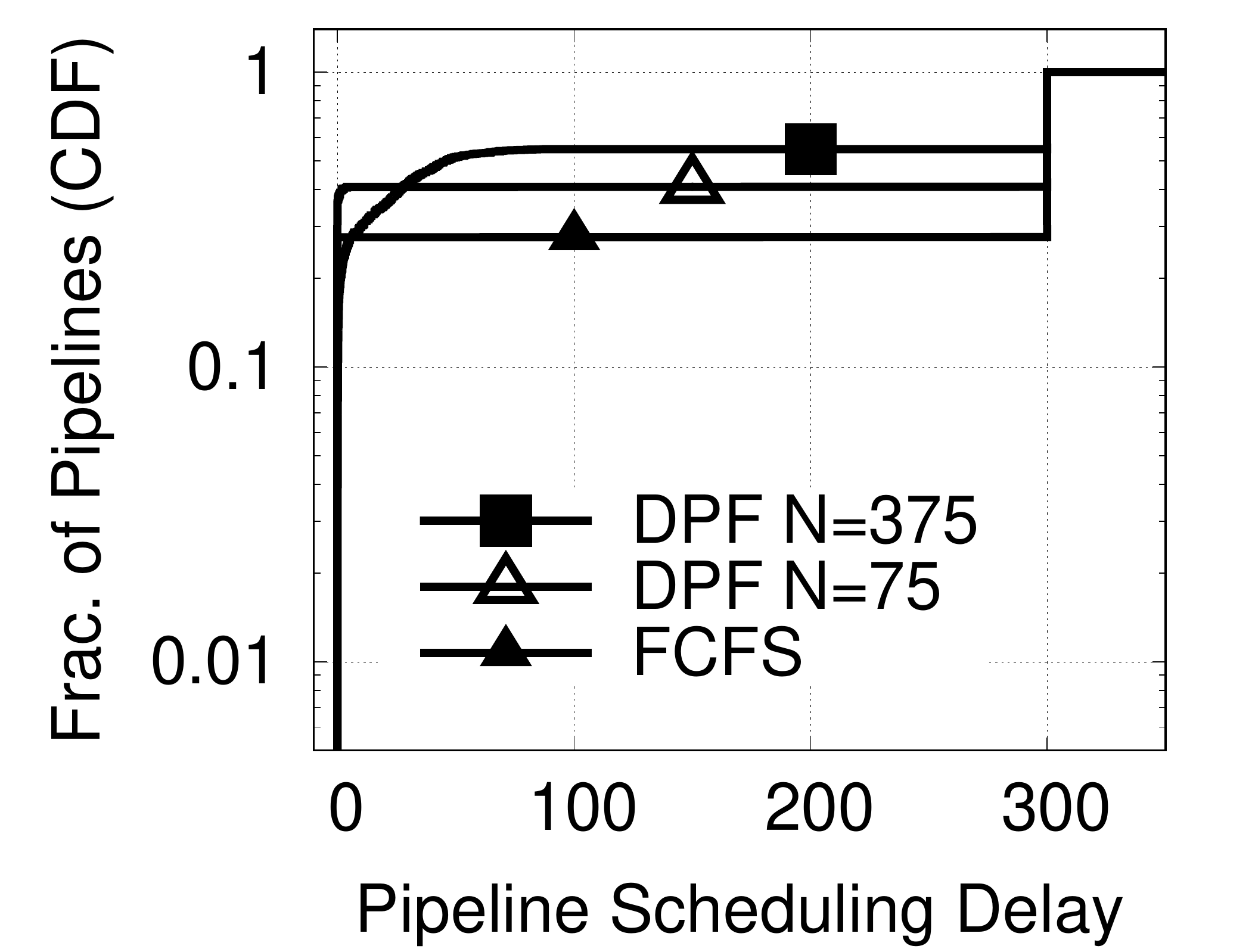}
		\caption{\footnotesize {\bf Scheduling delay.}}
		\label{fig:evaluation:microbenchmark:multiple-blocks:wait-time}
	\end{subfigure}
	\caption{\footnotesize {\bf DPF behavior on multiple blocks.} }
	\label{fig:evaluation:microbenchmark:multiple-blocks}
\end{figure}

\F\ref{fig:evaluation:microbenchmark:multiple-blocks} shows the multi-block experiment results are similar to the single-block experiment.
The main difference is that DPF performance with very large $N$ drops,
because some blocks do not see enough requests to unlock
all their budget.
For RR, proportional allocation helps cross-blocks pipelines to
be granted (small $N$), yielding a small improvement over FCFS and $N=1$ DPF. When
$N>400$, the multiple blocks create more DP budget spread over
ungrantable pipelines, and there is no high allocation peak:
RR grants collapse while DPF shows a $2\times$ increase over FCFS.

\subsubsection{DPF-N vs. DPF-T}
\label{sec:evaluation:microbenchmark:dpfn-dpft}

\begin{figure}[t]
	\centering
	\begin{subfigure}{0.5\linewidth}
		\centering
		\includegraphics[width=\linewidth]{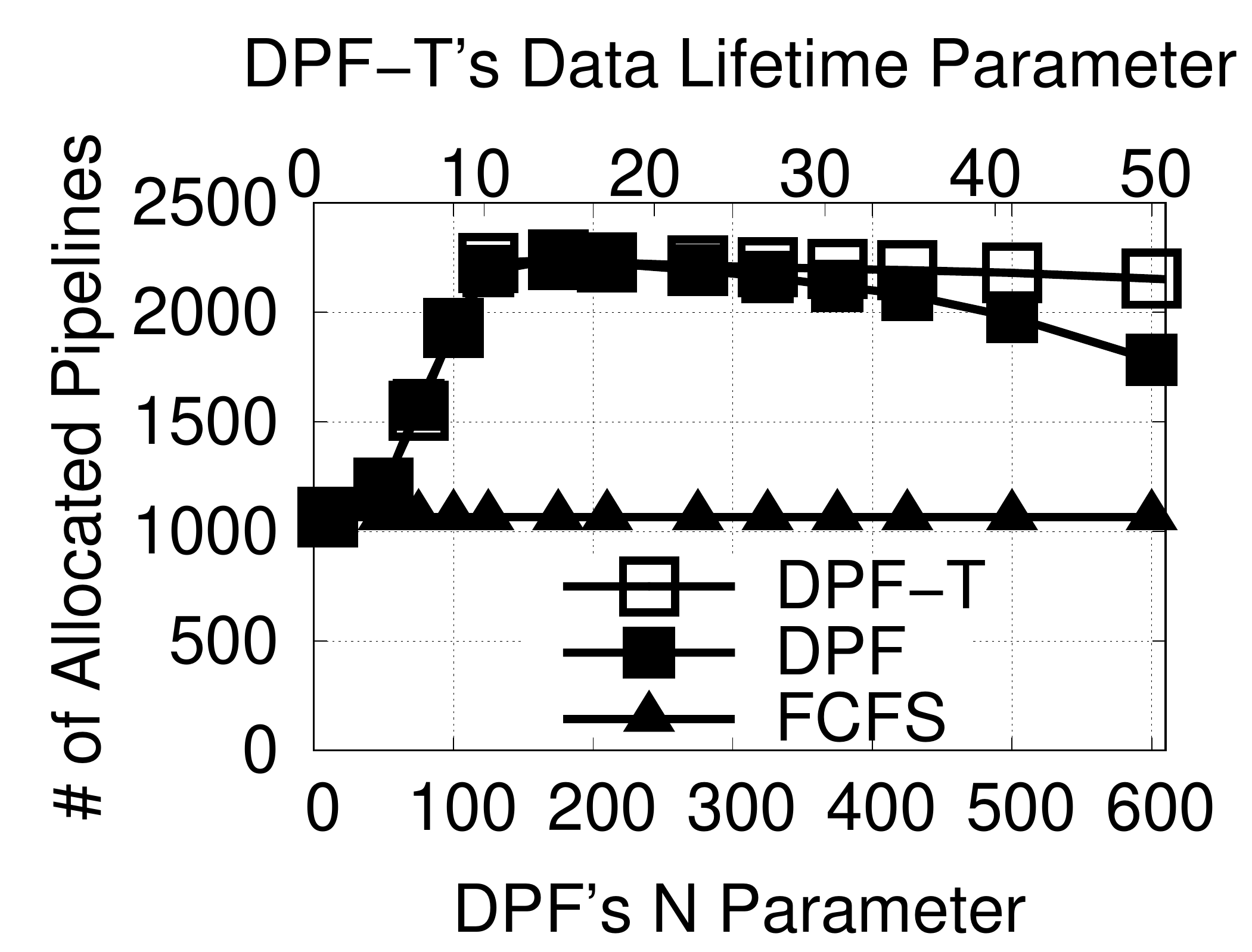}
		\caption{\footnotesize {\bf Number of pipelines allocated.}}
		\label{fig:evaluation:microbenchmark:dpfn-dpft:completed}
	\end{subfigure}%
    ~
	\begin{subfigure}{0.5\linewidth}
		\centering
		\includegraphics[width=\linewidth]{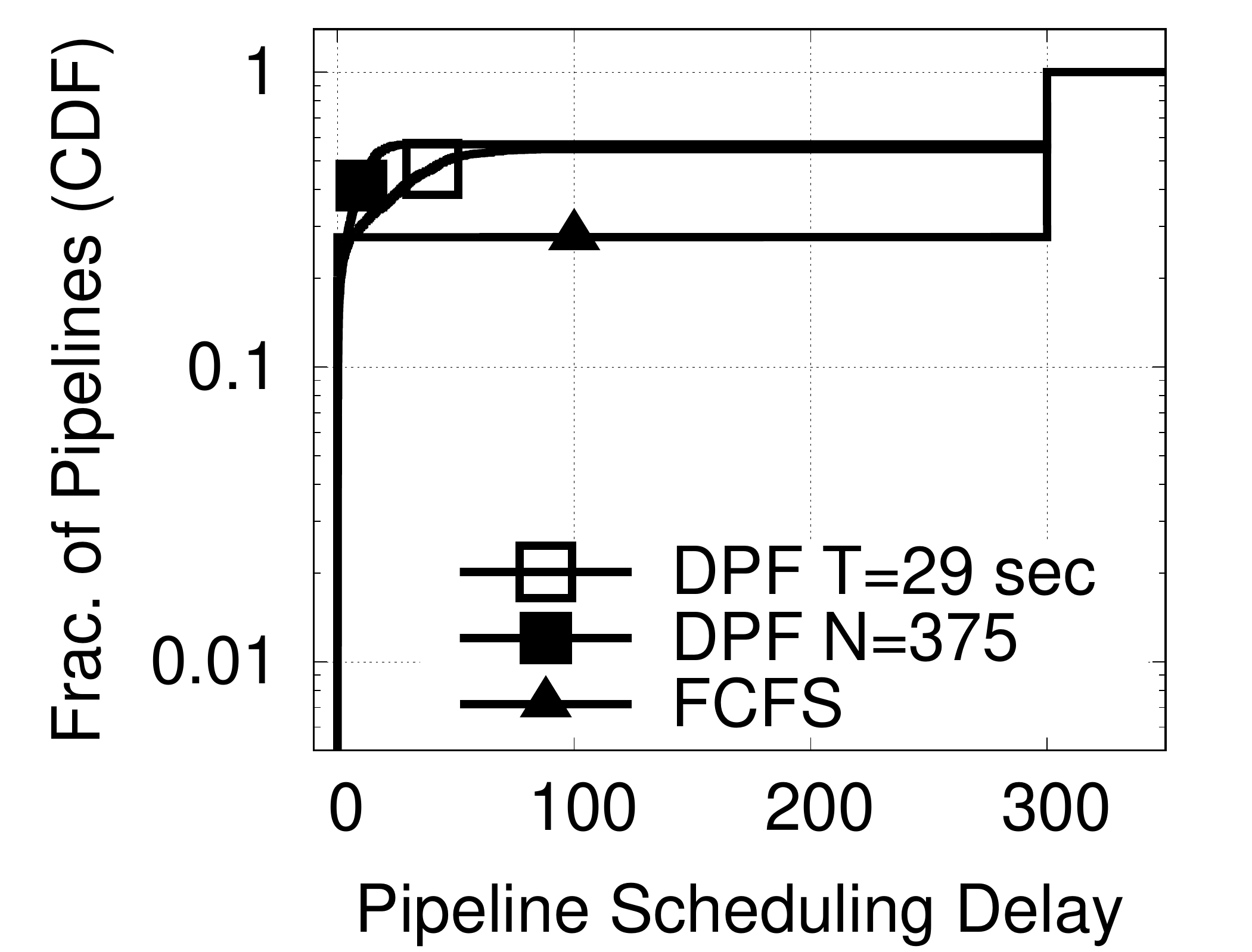}
		\caption{\footnotesize {\bf Scheduling delay.}}
		\label{fig:evaluation:microbenchmark:dpfn-dpft:wait-time}
	\end{subfigure}
	\caption{\footnotesize {\bf DPF and DPF-T behavior on multiple blocks.} }
	\label{fig:evaluation:microbenchmark:dpfn-dpft}
\end{figure}

\F\ref{fig:evaluation:microbenchmark:dpfn-dpft} compares DPF-N, the version used throughout the paper, which unlocks budget based on arriving pipelines,
and DPF-T, which releases budget based on time
(\S\ref{sec:DPF-extensions}).
We observe that on low $N$ and $T$ 
they behave almost identically.  This is because DPF-T will release budget
on less queried blocks, sometimes allowing multi-block pipelines
to be prematurely granted. On large $N$ and $T$ values DPF-T does much better, as all budget is eventually unlocked and some waiting pipelines can
be granted, even when no new request is made to the blocks they demanded.
\F\ref{fig:evaluation:microbenchmark:dpfn-dpft:wait-time} shows the delay
for equivalent $N$ and $T$ values.

\subsubsection{Traditional DP vs. R\'enyi DP}
\label{sec:evaluation:microbenchmark:dpfn-rdp}

\begin{figure}[t]
	\centering
	\begin{subfigure}{0.5\linewidth}
		\centering
		\includegraphics[width=\linewidth]{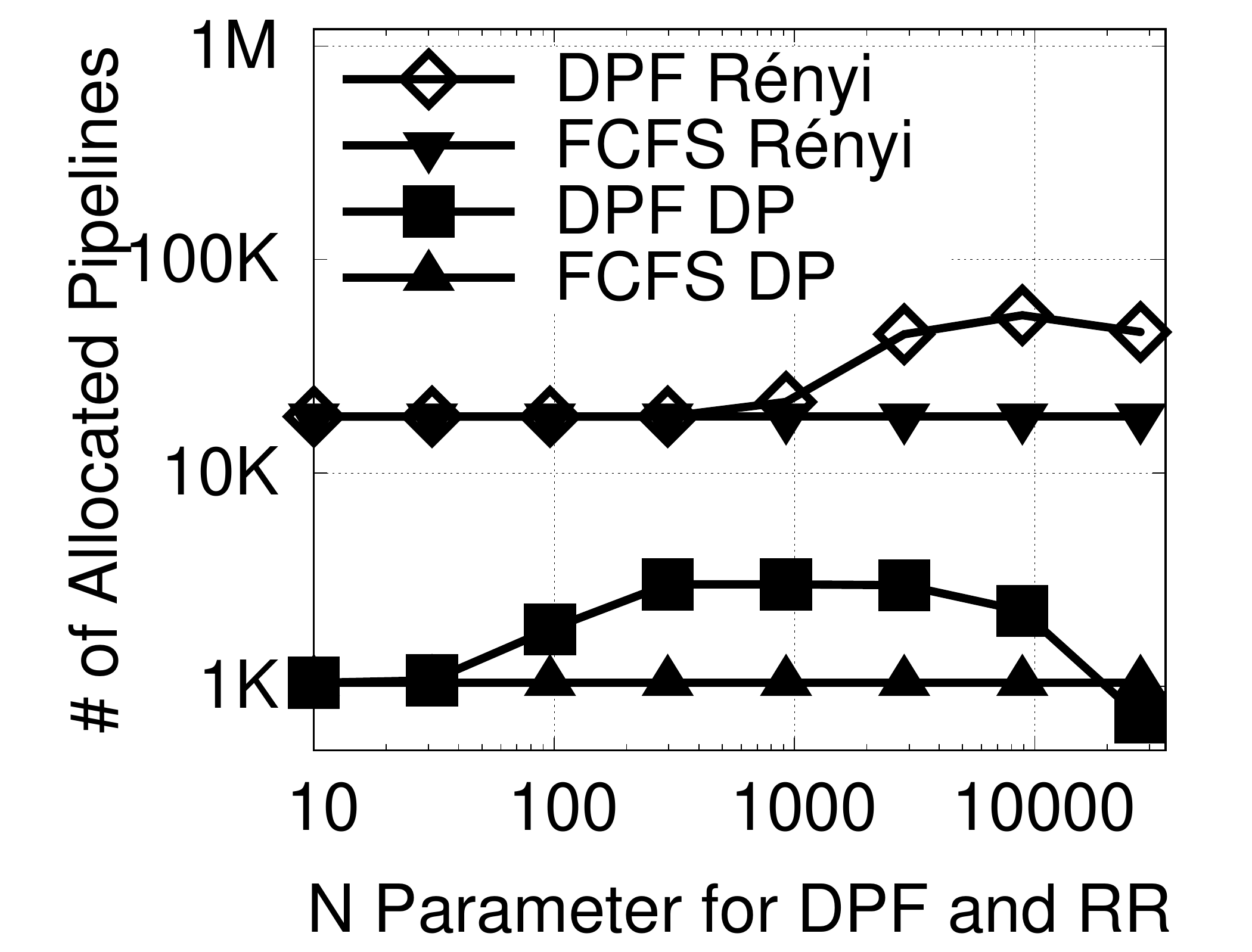}
		\caption{\footnotesize {\bf Number of pipelines allocated.}}
		\label{fig:evaluation:microbenchmark:dpfn-rdp:completed}
	\end{subfigure}%
    ~
	\begin{subfigure}{0.5\linewidth}
		\centering
		\includegraphics[width=\linewidth]{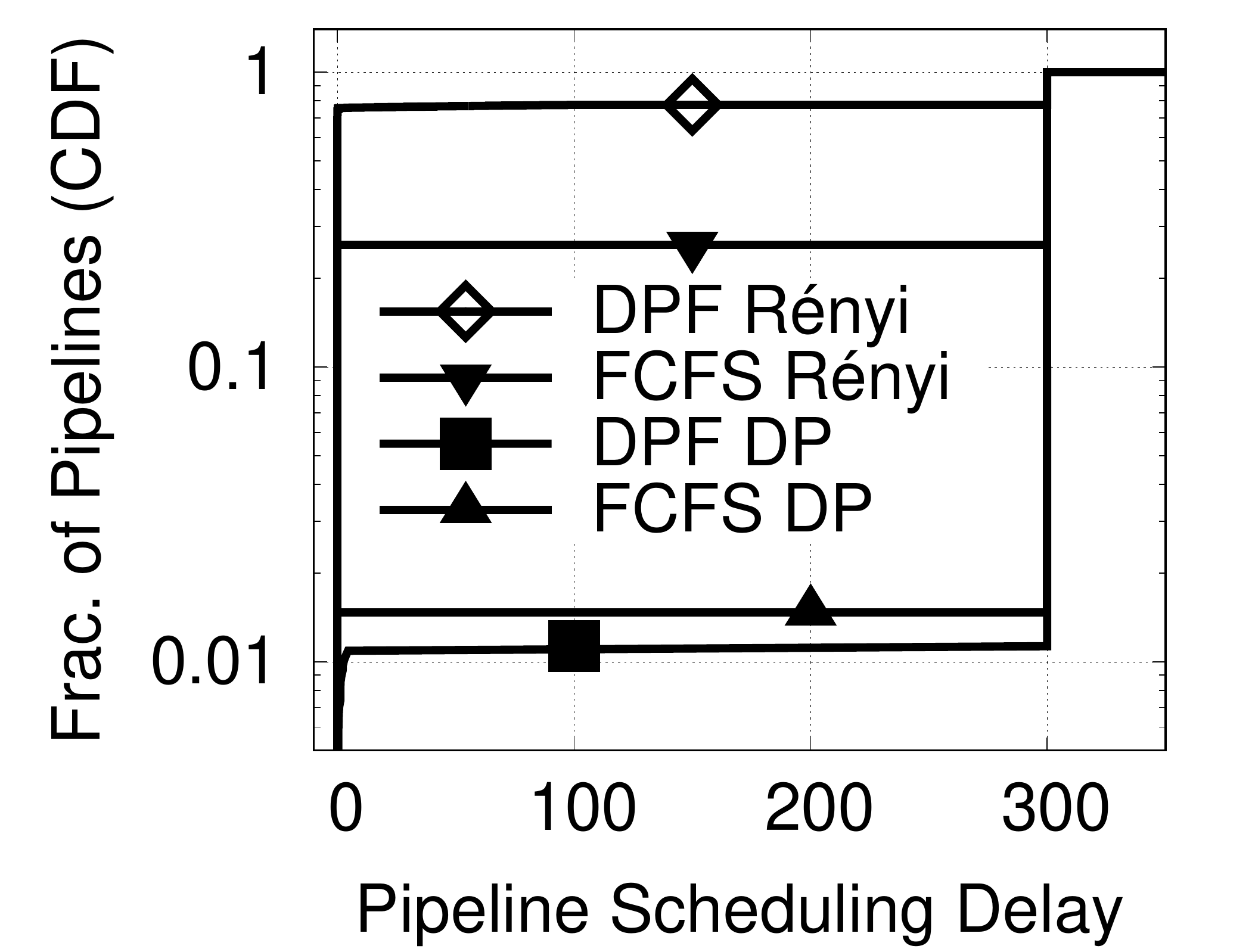}
		\caption{\footnotesize {\bf Scheduling delay.}}
		\label{fig:evaluation:microbenchmark:dpfn-rdp:wait-time}
	\end{subfigure}
	\caption{\footnotesize {\bf Traditional vs. R\'enyi DP, multiple blocks.} (a) Note log axes. Workload is highly amplified to saturate R\'enyi. (b) DPF N=8875.}
	\label{fig:evaluation:microbenchmark:dpfn-rdp}
\end{figure}

\F\ref{fig:evaluation:microbenchmark:dpfn-rdp} compares the DPF algorithm with traditional DP (the default DP composition used in the paper), against R\'enyi DP, including FCFS with both compositions as a baseline. The results show that switching to R\'enyi DP results in much better pipeline allocation: R\'enyi DP allows DPF to allocate more than $17\times$ more pipelines than traditional DP, at their respective peaks. Even FCFS using R\'enyi DP significantly outperforms DPF with traditional DP. Note that DPF provides a benefit at different values of $N$ for the two compositions, since R\'enyi DP requires a higher $N$ value to reach the point where DPF starts prioritizing small pipelines. We conclude that switching to R\'enyi DP leads to much more efficient privacy budget utilization, regardless of the scheduling policy.

\subsection{Macrobenchmark (Q1, Q4, Q5)}
\label{sec:evaluation:macrobenchmarks}


We use a subset of Amazon Reviews~\cite{amazon-reviews} in which users and products have 5 reviews or more, and keep product categories with 1M+ reviews.
Each event has a review, timestamp, user, 1-5 rating, and product in one of eleven categories (\eg, books, clothing).
We keep the reviews from 01-01-13 to 01-01-18, in total 43.4M reviews from 3.7M users.
\T\ref{tab:evaluation:pipelines-hyperparameters} specifies our workload: eight ML pipelines and six summary statistics pipelines.
For ML, we define four types of models for each of two tasks: product classification (assigns a review to its product category) and sentiment analysis (predicts whether a review is positive).
Reviews are embedded using a Wikipedia-trained GloVe~\cite{glove} except for the fine-tuned BERT model.
We run non-DP architecture searches for non-DP and DP pipelines on a 1\% hold-out.

We set an accuracy goal for each pipeline: for summary statistics, $5\%$ relative error; for ML models, an accuracy reachable by User DP (\eg, $60\%$ for LSTM/Product).
Each pipeline demands the minimum amount of private blocks necessary to reach its goal with $\epsilon \in \{0.01, 0.05, 0.1\}$ (``mice,'' \ie statistics) and $\epsilon \in \{0.5, 1, 5\}$ (``elephants,'' \ie ML models). The demands range from 1 to 500 \privacyresources.
Models use $\delta = 10^{-9}$.  The workload draws 75\% mice and 25\% elephants.
Each \privacyresource holds one day of data and has $\epsilon^G=10$.  The experiments replay 50 days of the dataset.
Pipelines register with \sysname at exponentially distributed time intervals, at a rate of $300$ pipelines per day.

\begin{table}[t]
    \footnotesize
    \centering
    \begin{tabular}{|l|c|c|c|}
        \hline
        {\bf Task}        & {\bf Model}                                               & {\bf Architecture$^*$}  & {\bf Training}                      \\
        \hline

                          & Linear                                                    & 75; 100; []             & Optimizer: Adam                     \\
                          &                                                           & 1,111 parameters        & (for DP, non-DP).                   \\
        \cline{2-3}
                          & FF$^{\dagger\dagger}$                                     & 60; 100; [185, 150]     &                                     \\
        {\bf Product}     &                                                           & 48,246 parameters       & DP algo: DP-SGD                     \\
        \cline{2-3}
        {\bf classifi-}   & LSTM                                                      & 30; 100; [40]$^\dagger$ & (Opacus).                           \\
        {\bf cation}      &                                                           & 23,171 parameters       &                                     \\

        \cline{2-3}       & BERT                                                      & L 4; H 256; A 4$^\S$    & Epochs: non-DP,                     \\
                          &                                                           & 858,379 parameters      & event/event-time                    \\

        \cline{1-3}

                          & Linear                                                    & 50; 100; []             & DP: 15; user DP: 60.                \\
                          &                                                           & 101 parameters          &                                     \\

        \cline{2-3}
                          & FF$^{\dagger\dagger}$                                     & 30; 100; [150, 110]     & Batch: non-DP: 256;                 \\
        {\bf Sentiment}   &                                                           & 31,871 parameters       & DP: $\sqrt{N}$ for $N$ train        \\

        \cline{2-3}
        {\bf  analysis}   & LSTM                                                      & 50; 100; [40]$^\dagger$ & samples (per~\cite{abadi2016deep}). \\
                          &                                                           & 22,761 parameters       &                                     \\

        \cline{2-3}       & BERT                                                      & L 4; H 256; A 4$^\S$    & DP clipping: flat,                  \\
                          &                                                           & 855,809  parameters     & max norm = $1$.                     \\

        \hline
                          & \multicolumn{2}{c|}{ Reviews: total \#, per category \# } & Laplace. Bounded                                              \\
        \cline{2-3}
        {\bf  Statistics} & \multicolumn{2}{c|}{Tokens: total \#, avg, stdev }        & user contribution:                                            \\
        \cline{2-3}
                          & \multicolumn{2}{c|}{ Rating: avg  }                       & 20/day, 100 in total                                          \\
        \hline
    \end{tabular}
    \caption{\footnotesize {\bf Macrobenchmark pipelines.}
        $*$: Architecture column: the first line, $x; y; z$, shows the input sequence length ($x$), embedding size ($y$), and the list of hidden layers' size ($z$).  The second line shows the number of trainable parameters.
        $^{\dagger\dagger}$: Fully-connected feed-forward neural network.
        $\dagger$: The LSTM is single directional and has no dropout.
        $\S$: We use a pretrained BERT model and fine-tune the last transformer layer with over 850K trainable parameters.
    }
    \label{tab:evaluation:pipelines-hyperparameters}
\end{table}

\subsubsection{Accuracy of Individual Models with DP Semantic}
\label{sec:evaluation:macrobenchmark:individual-model-accuracy}

\begin{figure*}[t]
    \centering
    \begin{subfigure}{0.24\linewidth}
        \centering
        \includegraphics[width=\linewidth]{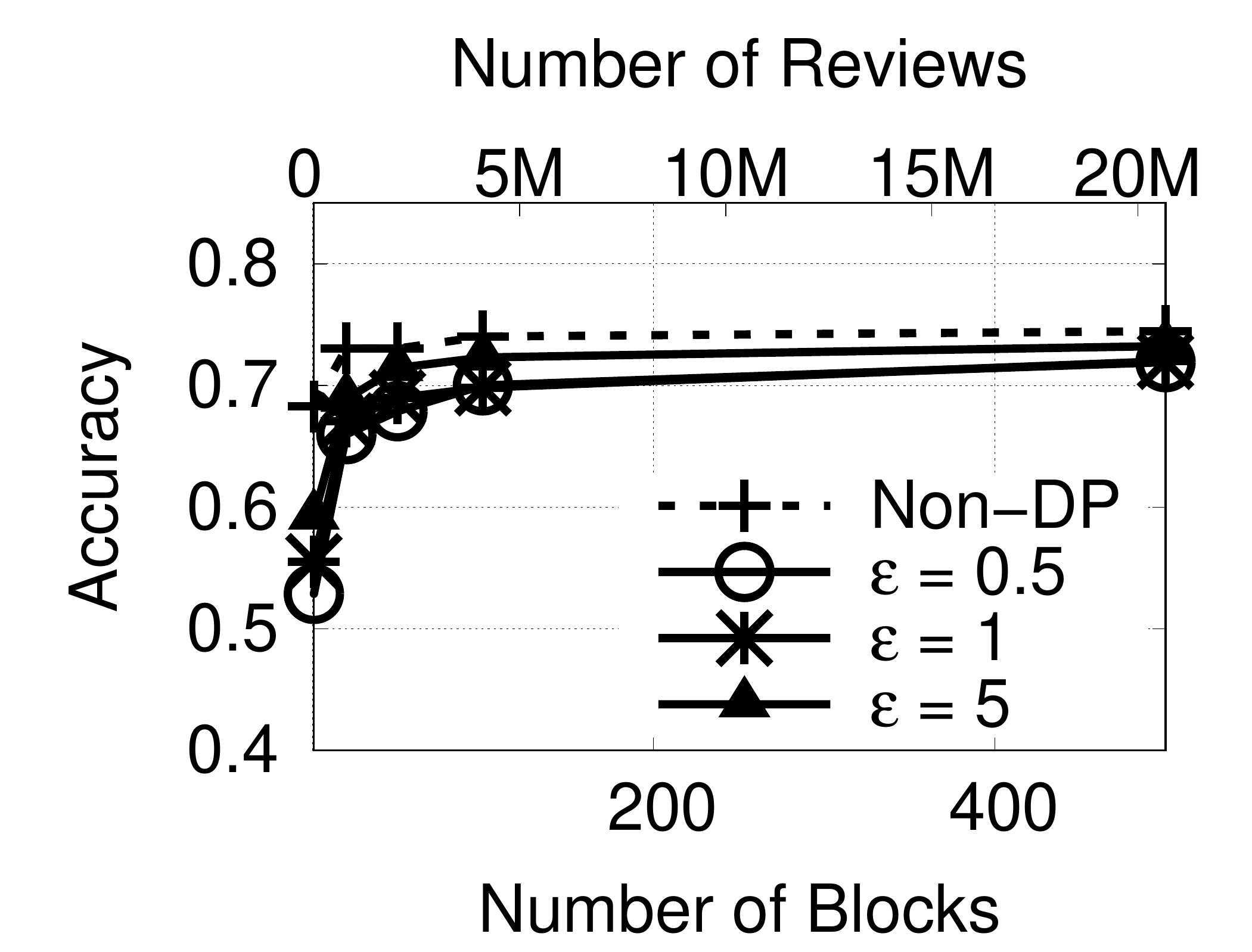}
        \caption{\footnotesize {\bf Product/LSTM: Event DP}}
        \label{fig:evaluation:macrobenchmark:individual-model-accuracy:event}
    \end{subfigure}%
    ~
    \begin{subfigure}{0.24\linewidth}
        \centering
        \includegraphics[width=\linewidth]{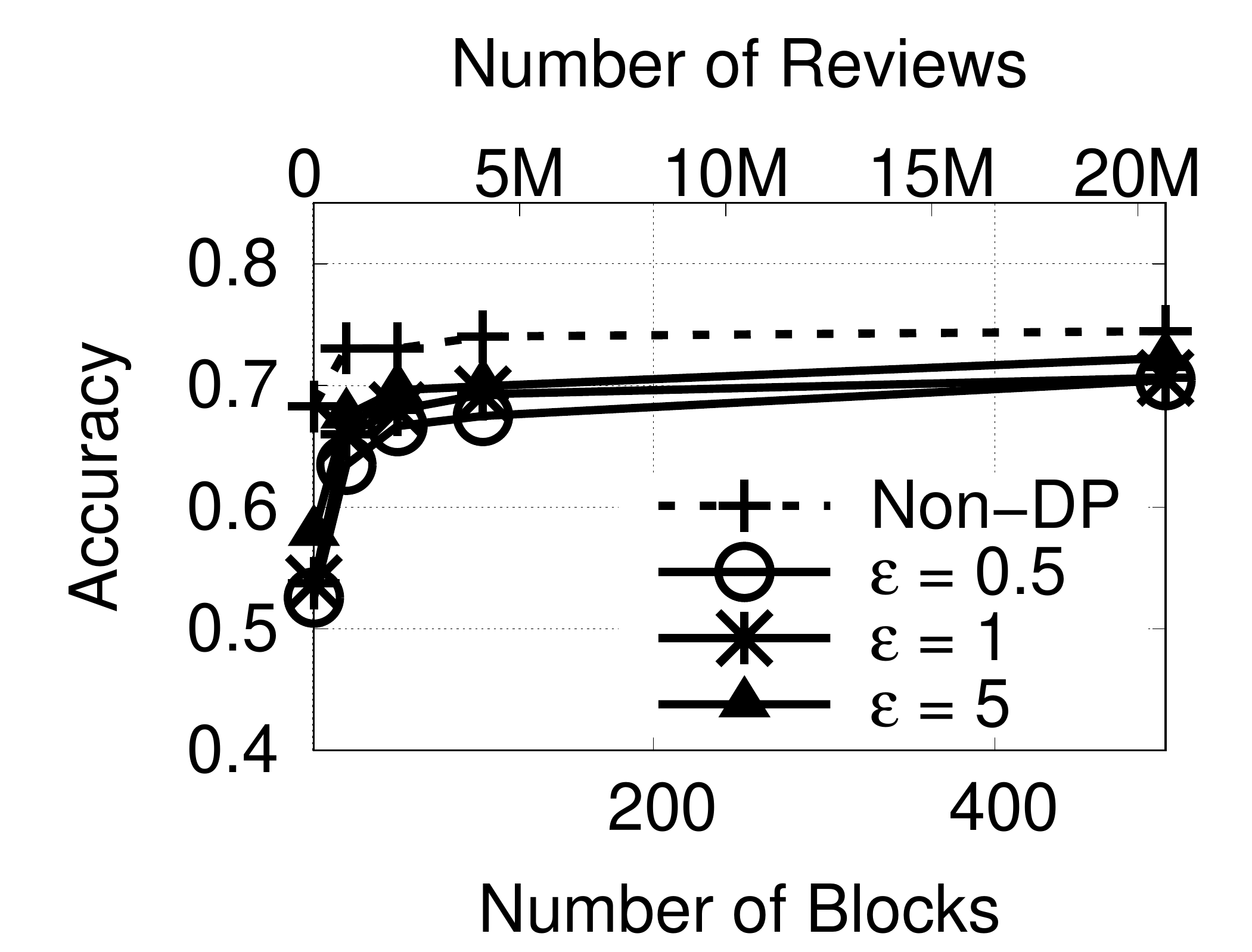}
        \caption{\footnotesize {\bf Product/LSTM: User-Time DP}}
        \label{fig:evaluation:macrobenchmark:individual-model-accuracy:user-time}
    \end{subfigure}%
    ~
    \begin{subfigure}{0.24\linewidth}
        \centering
        \includegraphics[width=\linewidth]{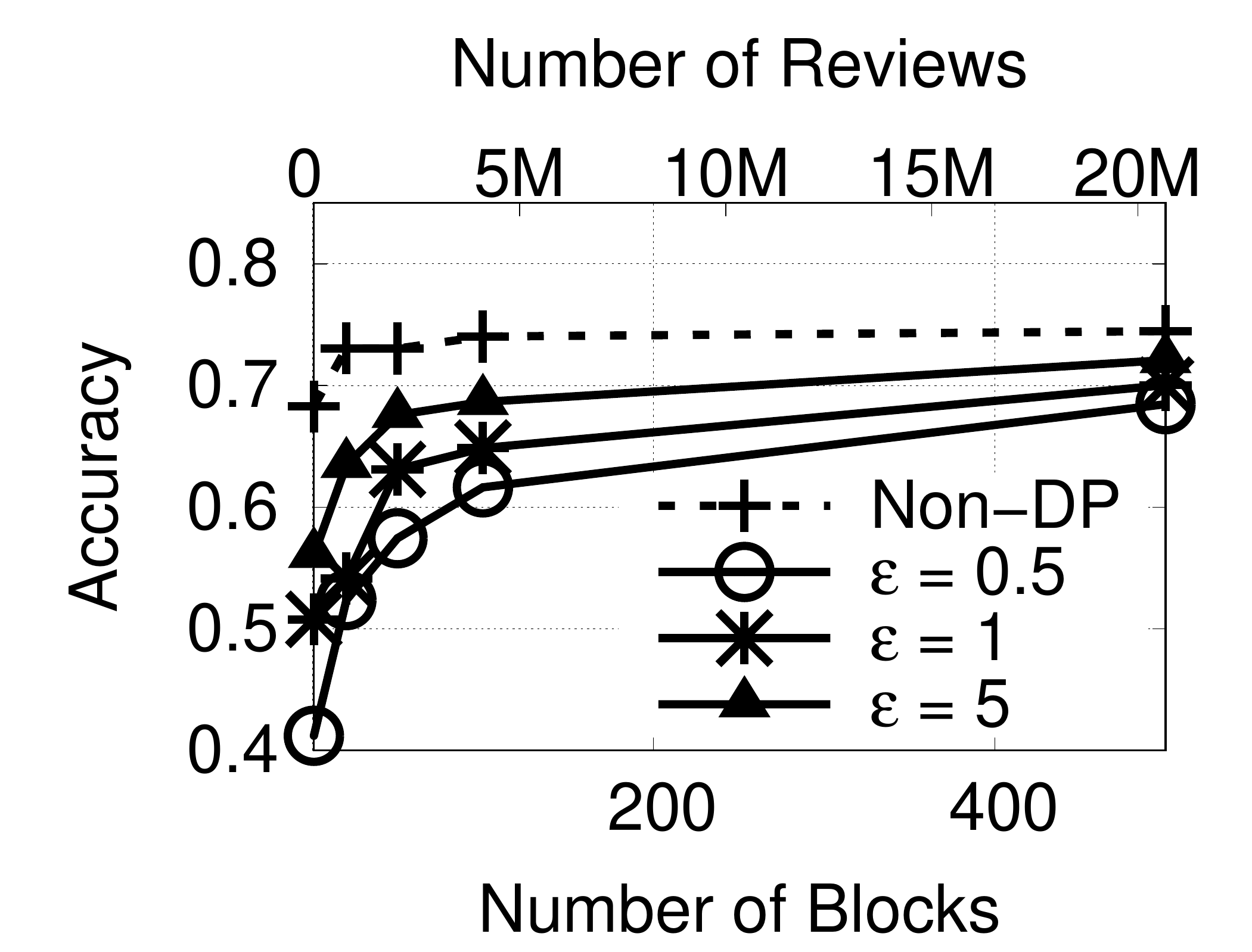}
        \caption{\footnotesize {\bf Product/LSTM: User DP}}
        \label{fig:evaluation:macrobenchmark:individual-model-accuracy:user}
    \end{subfigure}%
    ~
    \begin{subfigure}{0.24\linewidth}
        \centering
        \includegraphics[width=\linewidth]{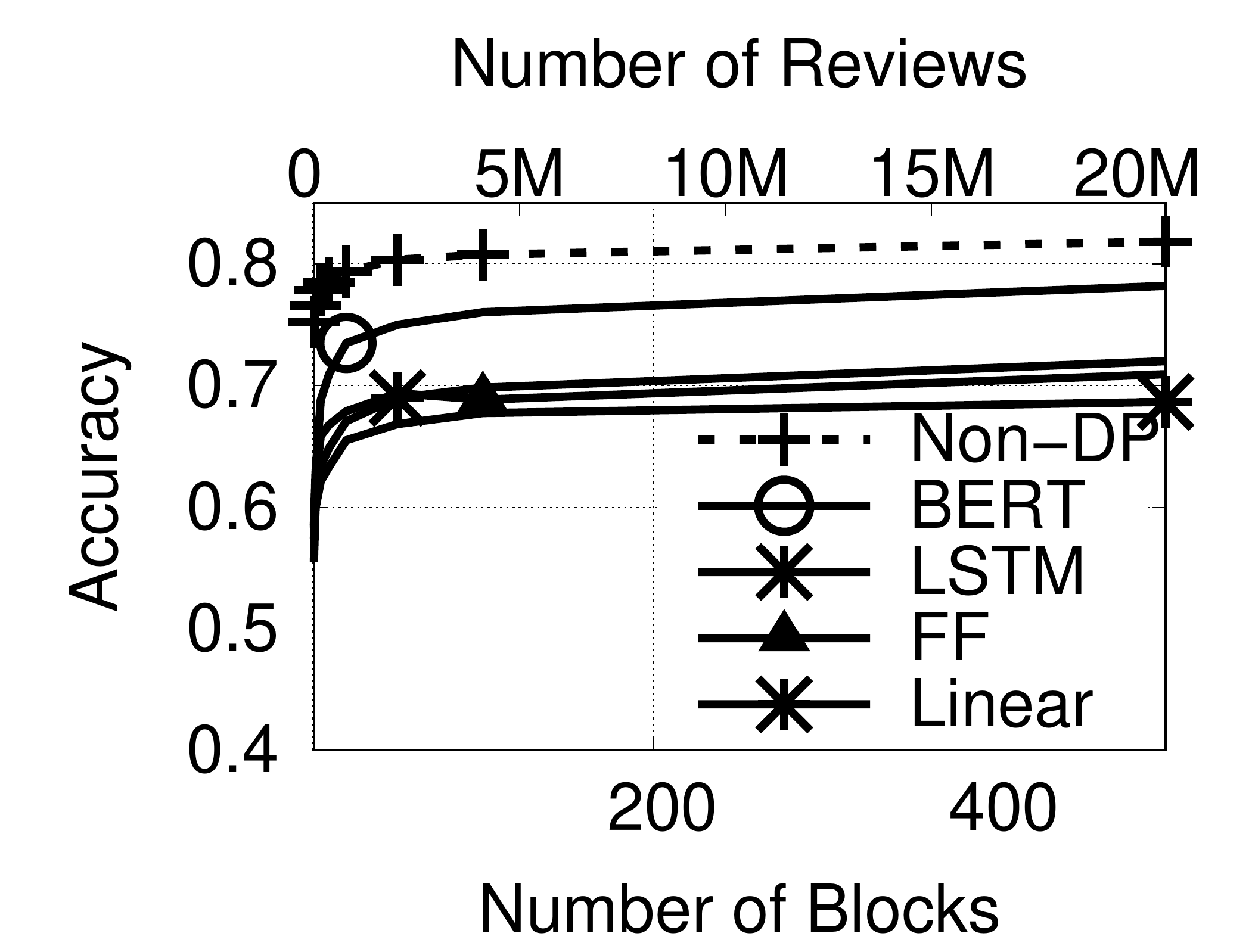}
        \caption{\footnotesize {\bf All Product: Event DP}}
        \label{fig:evaluation:macrobenchmark:individual-model-accuracy:all}
    \end{subfigure}%
    \caption{\footnotesize {\bf Performance of macrobenchmark Product models.} (a)-(c) Accuracy of the product classification LSTM with various DP semantics. (d) Accuracy of all four product classification models with $\epsilon = 1$ and Event DP.  The dotted baseline is non-DP BERT, whose accuracy is highest.  The y axes start at 0.4, the accuracy of the naive classifier for this task (\ie the classifier that returns the most common class). }
    \label{fig:evaluation:macrobenchmark:individual-model-accuracy-and-dpf-with-dp-semantic}
\end{figure*}

\F\ref{fig:evaluation:macrobenchmark:individual-model-accuracy-and-dpf-with-dp-semantic} shows the LSTM's product classification accuracy with increasing data, with no DP and for $\epsilon \in \{0.5, 1, 5\}$ for each DP semantic.
Other pipelines show similar trends.
We make two observations.
First, DP semantic has a large impact on accuracy for a given DP budget and data size.
As expected, Event DP, the weakest semantic, provides the highest accuracy: 73\%, 72\%, and 72\%, for DP budgets of 5, 1, and 0.5 respectively, on 20M datapoints.
The larger budgets get close to the non-DP baseline, at 77\%.
User DP requires larger budgets: the largest reaches 72\% while the smallest yields 68\%.
User-time DP's behavior is closer to, but lower than, Event DP, with accuracies of 72\%, 71\%, and 70\%.

Second, increasing data or budget improves accuracy: the DP models approach the baseline slowly, but can reach it given enough data and DP budget.
The relationship between accuracy, data, and budget however is non linear.
For event DP with 20M datapoints, increasing the budget from 0.5 to 5 increases accuracy from 72\% to 73\%, while at 2.5M datapoints the same increase goes from 68\% to 71\%.
This relationship also depends on DP semantics, with low budget models being disproportionately impacted by smaller amounts of data and budget.
For user DP for instance, the accuracies go from 68\% to 72\% for 20M datapoints, and from 57\% to 68\% for 2.5M.

\subsubsection{DPF Behavior with Macrobenchmark}
\label{sec:evaluation:macrobenchmark:dpf}

\begin{figure}[t]
    \centering
    \begin{subfigure}{0.45\linewidth}
        \centering
        \includegraphics[width=\linewidth]{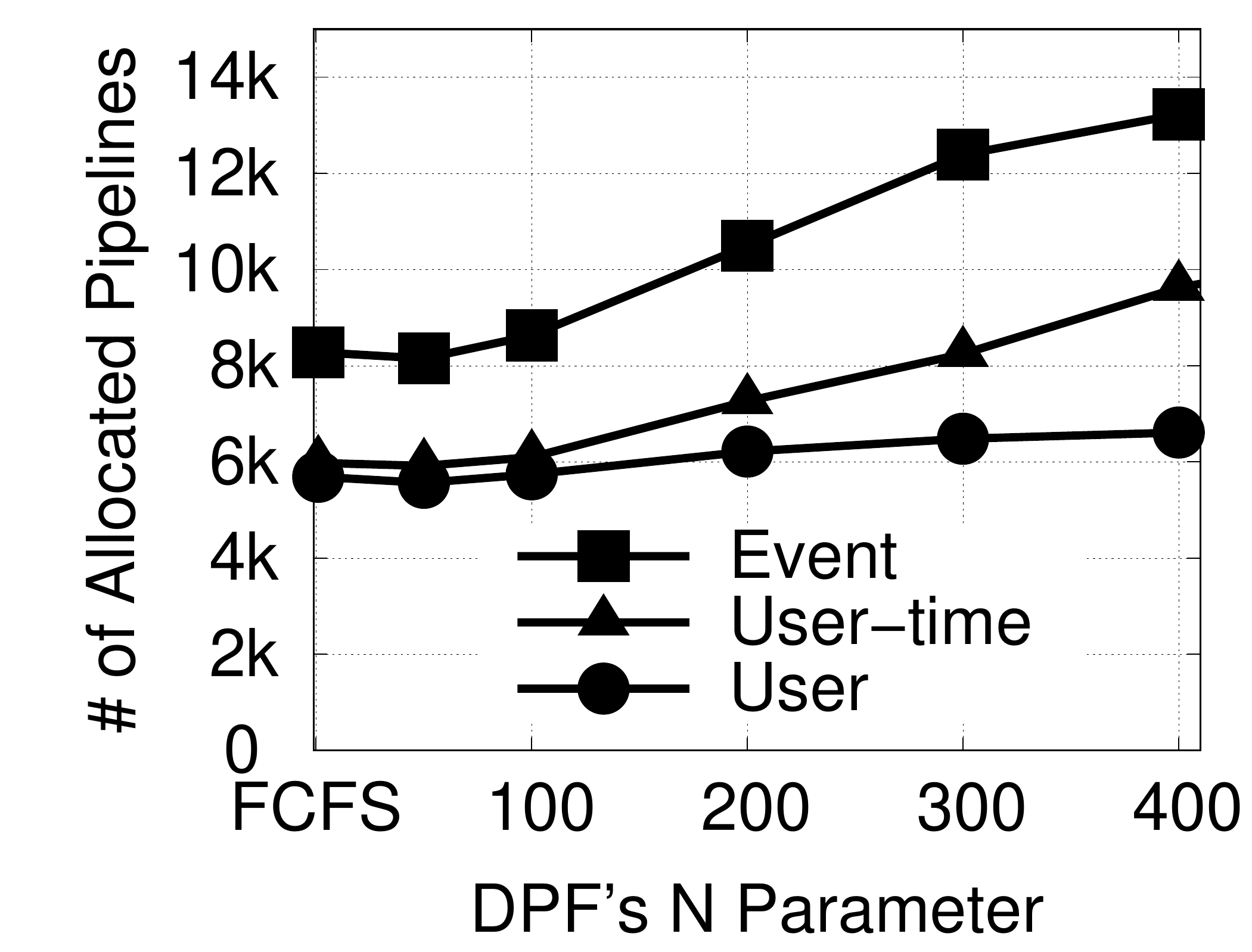}
        \caption{\footnotesize {\bf Allocated for 3 DP semantics.}}
        \label{fig:evaluation:macrobenchmark:dpf-on-event-dp-workload:completed}
    \end{subfigure}%
    ~
    \begin{subfigure}{0.45\linewidth}
        \centering
        \includegraphics[width=\linewidth]{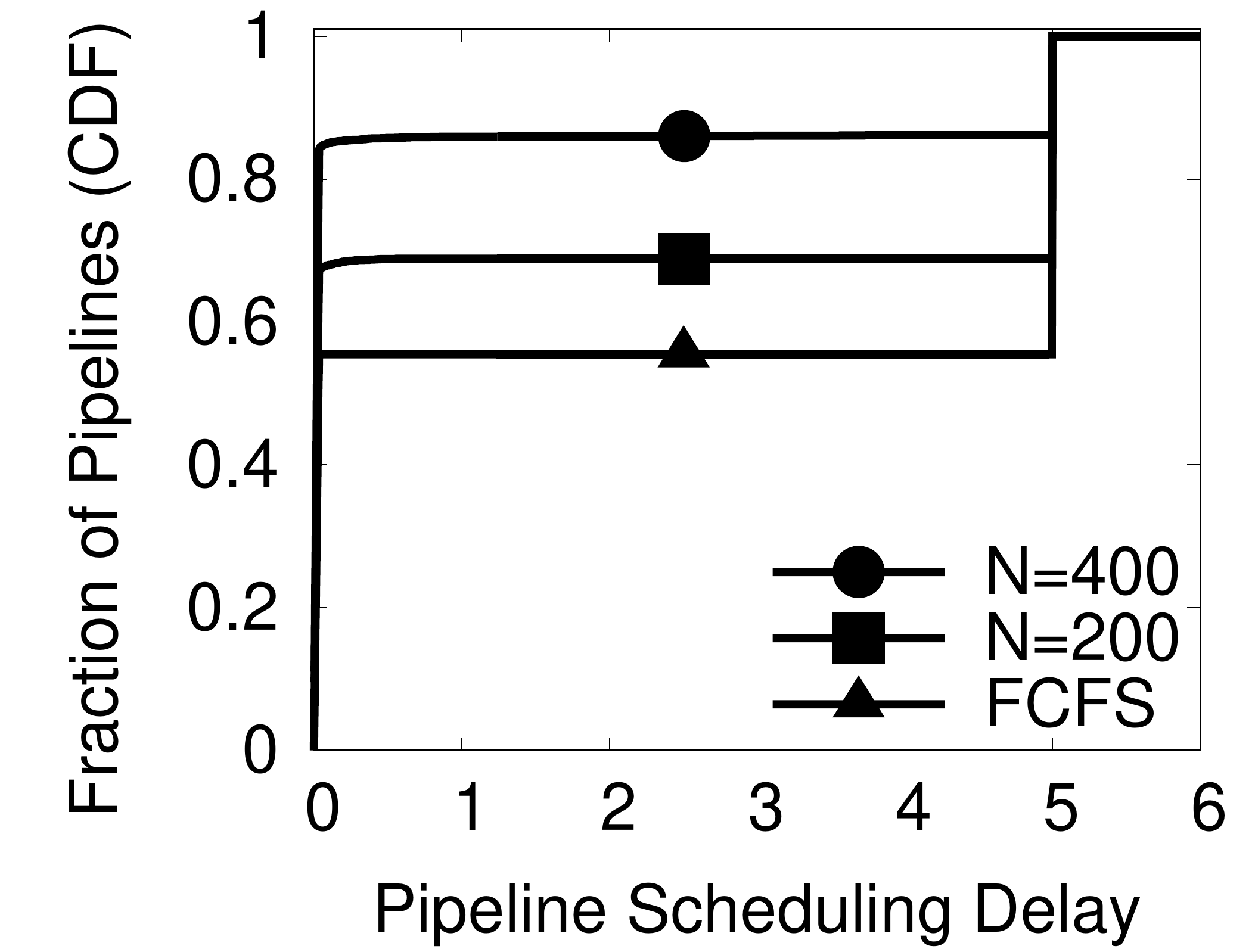}
        \caption{\footnotesize {\bf Event DP Scheduling delay.}}
        \label{fig:evaluation:macrobenchmark:dpf-on-event-dp-workload:wait-time}
    \end{subfigure}
    \caption{\footnotesize {\bf DPF on macrobenchmark.} $\epsilon^G = 10$, $\delta^G = 10^{-7}$.}
    \label{fig:evaluation:macrobenchmark:dpf-on-event-dp-workload}
\end{figure}

\F\ref{fig:evaluation:macrobenchmark:dpf-on-event-dp-workload} shows the performance
of {\em DPF with R\'enyi DP} under our end-to-end workload.
\F\ref{fig:evaluation:macrobenchmark:dpf-on-event-dp-workload:completed}
shows the number of granted pipelines under the different DP semantics.
We make two observations. First, as expected stronger DP semantics require more \privacyresource and DP budget, so fewer pipelines are granted in total: event, user-time, and user DP can grant 13.8k, 10.4k, and 6.7k pipelines, respectively.
Second, as before, increasing $N$ helps DPF prioritize later mice over current elephants, increasing the total number of pipelines granted by 67\% (event), 75\% (user-time) and 17\% (user) compared to low $N$ and FCFS.
\F\ref{fig:evaluation:macrobenchmark:dpf-on-event-dp-workload:wait-time} shows the scheduling
delay of user DP for $N$ values of 200 and 400. We see that increase in pipelines granted comes at a reasonable cost in delay.

\begin{figure}[t]
    \centering
    \includegraphics[width=0.45\linewidth]{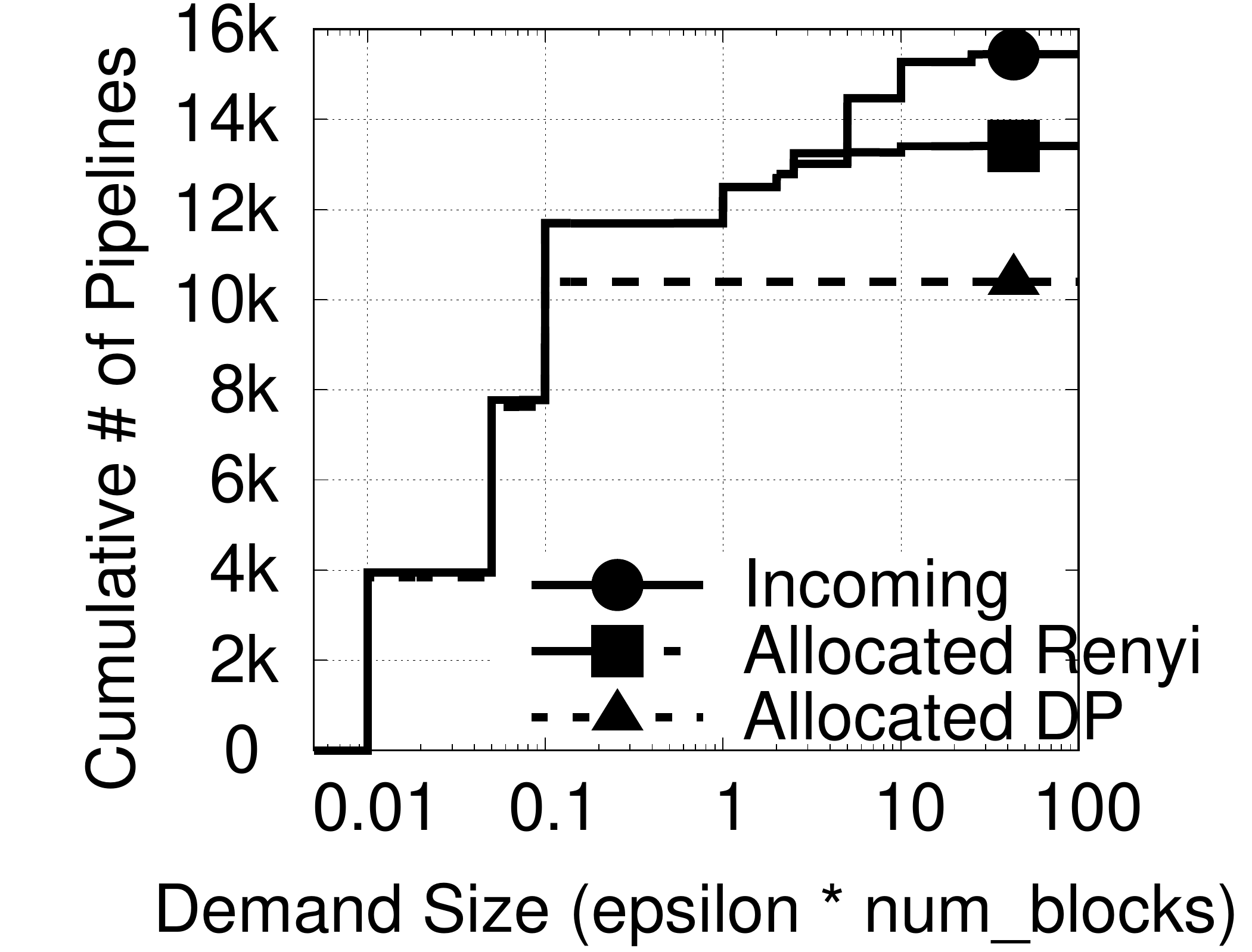}
    \caption{\footnotesize {\bf Distribution of allocated pipeline sizes.} Event DP, DPF N=400.}
    \label{fig:evaluation:macrobenchmark:dpf-on-event-dp-workload:job-sizes}
\end{figure}

\F\ref{fig:evaluation:macrobenchmark:dpf-on-event-dp-workload:job-sizes} shows the cumulative number of incoming pipelines below a given DP size in our workload, as well as those granted under DP and R\'enyi DP.
The DP size of a pipeline is the sum of $\epsilon$-DP budget over all requested blocks,
and is a measure of the total amount of budget requested by the pipeline.
The R\'enyi DP allocates about $29\%$ more pipelines than DP.
This difference is {\em quantitatively} smaller than we obtained in our microbenchmark.
However, there is a big {\em qualitative} difference that this graph also illustrates: while DP only grants mice (cumulative budget below $0.1$), R\'enyi DP is able to also run some elephants: it grants all pipelines with a cumulative budget below $2$ and some pipelines up to $10$.
This confirms that R\'enyi DP is very valuable in realistic workload settings.

\subsection{Kubernetes Tool Reuse (Q6)}
\label{sec:evaluation:grafana}

\begin{figure}[t]
    \centering
    \includegraphics[width=0.8\linewidth]{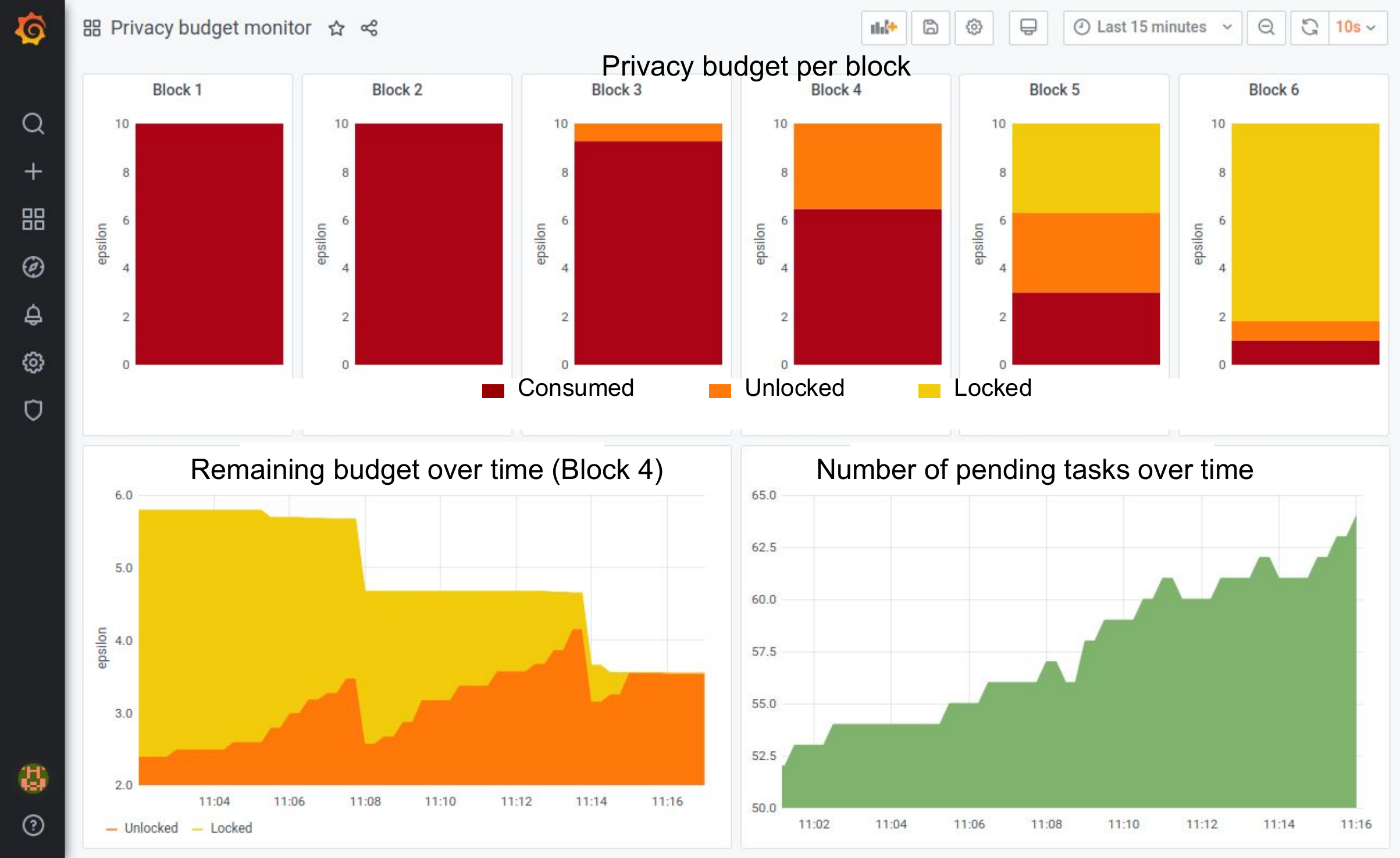}
    \caption{\footnotesize {\bf Screenshot of Grafana-Kubernetes Privacy Dashboard.} 
    }
    \label{fig:evaluation:grafana}
\end{figure}

To illustrate the value of integrating with Kubernetes, we extended the Grafana-Kubernetes resource utilization monitor to track privacy usage (screenshot depicted in \F\ref{fig:evaluation:grafana}) with only 150 lines of code.
We envision a suite of tools for monitoring privacy, on par with compute resources. 

\vspace{-0.3cm}
\section{Related Work}                      
\vspace{-0.3cm}
\label{sec:related-work}
To our knowledge, there is no work on scheduling DP, but our work builds upon a vast literature in each of these two topics.

\heading{Scheduling.}
Decades of work exist on scheduling compute resources, such as CPU, network, memory and storage~\cite{drf,carbyne,hug,pisces,faircloud,fairride,quincy,memshare,demers1989analysis,goyal1996start,parekh1993generalized,caprita2005group,axboe2004linux,liu2003opportunistic,graphene,tetris,parkes2015beyond}.
Typically, schedulers aim for max-min fairness, achieving both high system-wide utilization and high utility for each tenant.
However, compute resources are replenishable, while privacy budget is not: the particular budget consumed by task $i$ will never be available for another task in the future, whereas a CPU core granted to task $i$ can be granted to another task after $i$ finishes.
%
%

The two closets to our work are Dynamic DRF~\cite{dynamicdrf} and \textsc{SequentialMinMax}~\cite{parkes2015beyond}. Dynamic DRF provides
fairness guarantees for agents arriving over time, consuming a fixed set of non-replenishable resources.
Unfortunately, the all-or-nothing utility function of \privacyresources violates Dynamic DRF's Pareto efficiency,
since Dynamic DRF would waste budget on tasks that may never get fully allocated.
\textsc{SequentialMinMax} is an algorithm focused on ``indivisible'' jobs, or jobs that have an all-or-nothing utility, and thus,
similar to DPF, it only assigns resources in a sequential fashion and all-or-nothing fashion ordered by the dominant resource share.
However, unlike DPF, \textsc{SequentialMinMax} has static jobs, it assumes all resources are replenishable, and it does not consider dynamically arriving resources (\privacyresources in our case). Therefore, it provides no mechanism for gradually releasing or unlocking these resources, and would not provide a sharing incentive in our setting.

Even under a static setting, standard DRF~\cite{drf} violates Pareto efficiency with all-or-nothing utility.
\textsc{Carbyne} schedules analytics jobs, which depend on the parallel execution of
multiple tasks and have an all-or-nothing utility~\cite{carbyne}. However, it
assumes replenishable resources.

\heading{Differential privacy.}
There is vast literature on {\em DP algorithms}, which includes versions of most popular ML algorithms (\eg, SGD~\cite{abadi2016deep,yu2019differentially}, Federated Learning~\cite{McMahan2018LearningDP}) and statistics (\eg, contingency tables~\cite{barak2007privacy}, histograms~\cite{xu2012differentially}).  There are also open source implementations available~\cite{ms-harvard-opendp,idm-diffprivlib,google-dp,tensorflow-privacy,opacus}.
This literature is at a lower level than \sysname, and we leverage it extensively in our pipelines.
Some algorithms focus on workloads~\cite{hardt2010multiplicative}, including on a data stream~\cite{cummings2018differential}, but they remain very limited, supporting only linear queries.

A few {\em DP systems} exist, providing DP SQL-like~\cite{Mcsherry:pinq,proserpio2014calibrating} or MapReduce interfaces~\cite{roy2010airavat} to static datasets, as well as support for summary statistics~\cite{mohan2012gupt}.
None focuses on workloads of ML pipelines or supports continuous streams of data.
The only such system is Sage~\cite{sage}, which introduces block composition for event DP, and proposes a procedure to iteratively increase a model's privacy budget until reaching a good accuracy.
However, Sage does not support user and user-time DP, for which we extend block composition, and leaves the question of scheduling unexplored.

\vspace{-0.3cm}
\section{Conclusion}
\vspace{-0.3cm}
\label{sec:conclusions}

For workloads operating on sensitive user data privacy loss should be carefully orchestrated to enforce a global bound on personal data leakage.
This paper presented {\em \sysname}, an extension to the Kubernetes workload orchestrator that adds differential privacy budget as a new native resource to be managed alongside traditional compute resources.
\sysname incorporates a novel scheduling algorithm, {\em DPF}, the first one suitable for the unique characteristics of the privacy resource, including its all-or-nothing utility and non-replenishable nature.
We show that DPF has desirable theoretical properties, outperforms baseline scheduling algorithms, and that native integration of privacy into Kubernetes can facilitate reuse of existing tools to better manage this scarce resource.

\vspace{-0.3cm}
\section*{Acknowledgments}
\vspace{-0.3cm}
We thank Su Ji Park for developing and tuning baseline models for Amazon Reviews.
We thank our shepherd, Malte Schwarzkopf, and the anonymous reviewers for the valuable comments.
This work was funded by the U.S. Department of Energy (DOE) under award DE-SC-0001234; by the U.S. Army Research Office (ARO) under award W911NF-21-1-0078; by Google Research and Cloud awards; and by Sloan, Microsoft, Google, and Facebook awards.

\renewcommand{\baselinestretch}{1}
{
    \bibliographystyle{plain}
    \bibliography{bib/abbrev,bib/conferences,bib/dummy,bib/refs}
}

\clearpage
\appendix

\vspace{-0.3cm}
\section{Artifact Appendix}
\vspace{-0.3cm}
\label{apx:artifact}
\subsection{Abstract}

\begin{figure*}[t]
    \centering
    \begin{subfigure}{0.24\linewidth}
        \centering
        \includegraphics[width=\linewidth]{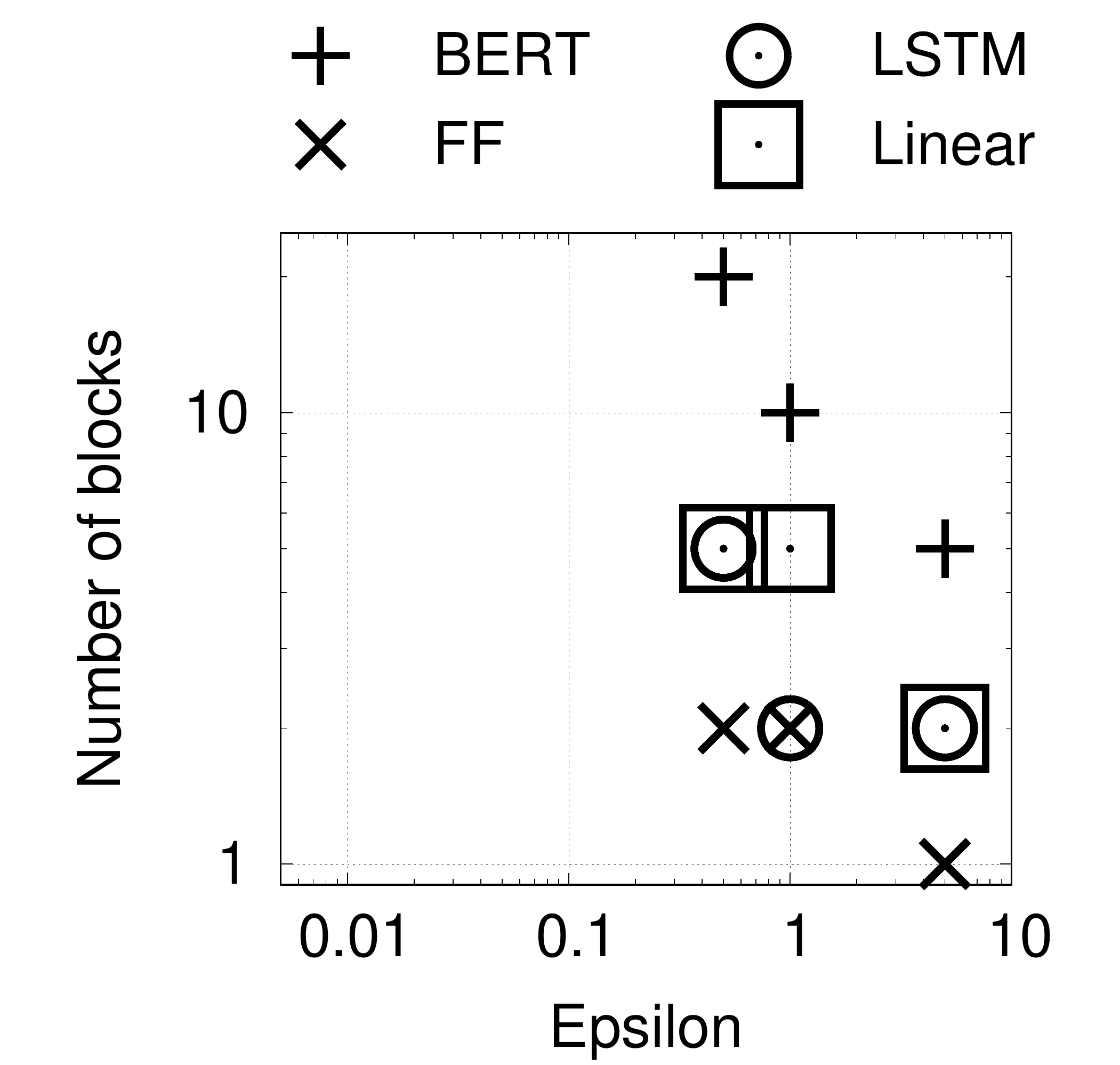}
        \caption{\footnotesize {\bf Product classification demands}}
        \label{fig:appendix:macrobenchmark:workload-mix:product}
    \end{subfigure}%
    ~
    \begin{subfigure}{0.24\linewidth}
        \centering
        \includegraphics[width=\linewidth]{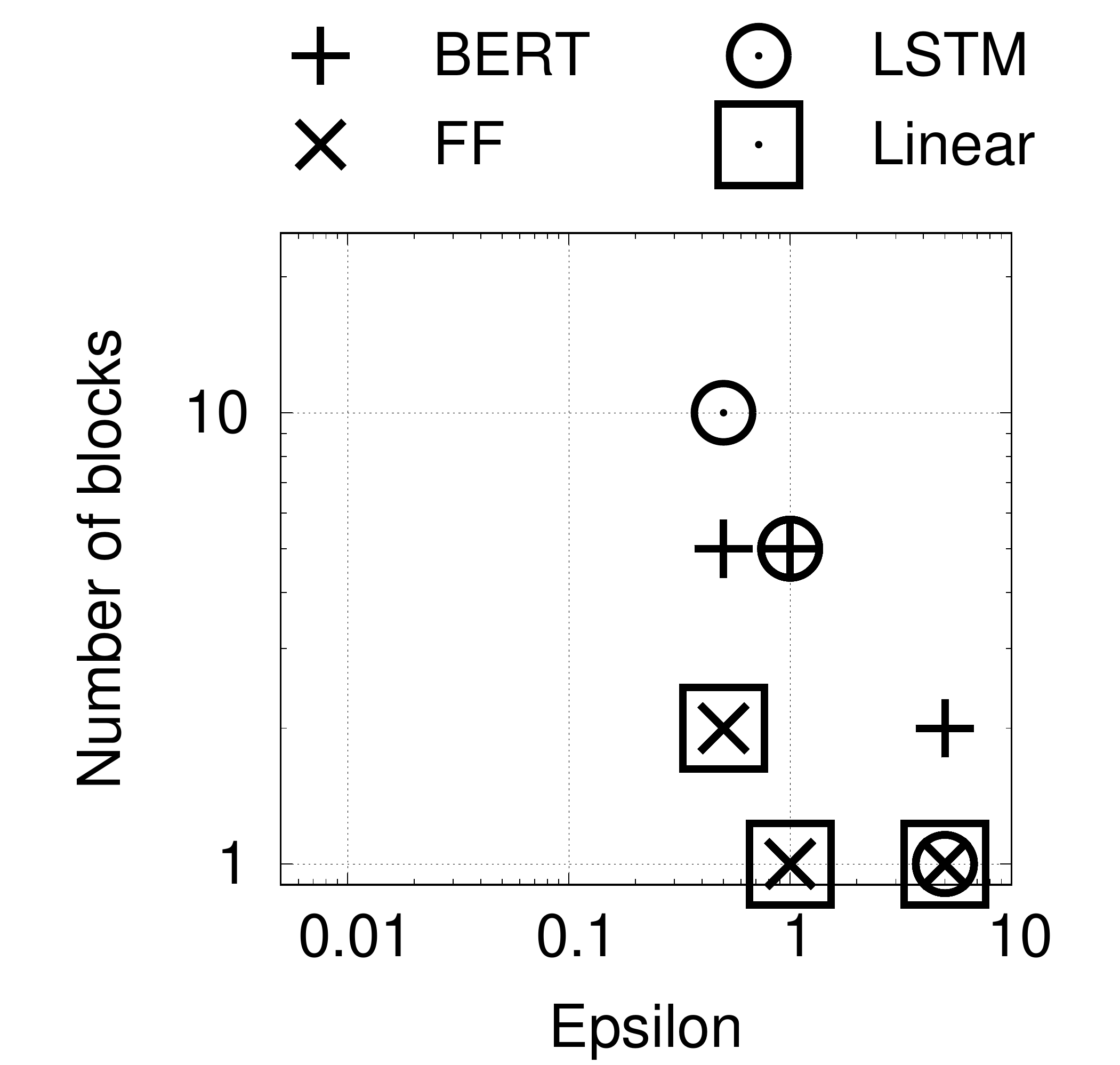}
        \caption{\footnotesize {\bf Sentiment analysis demands}}
        \label{fig:appendix:macrobenchmark:workload-mix:sentiment}
    \end{subfigure}%
    ~
    \begin{subfigure}{0.24\linewidth}
        \centering
        \includegraphics[width=\linewidth]{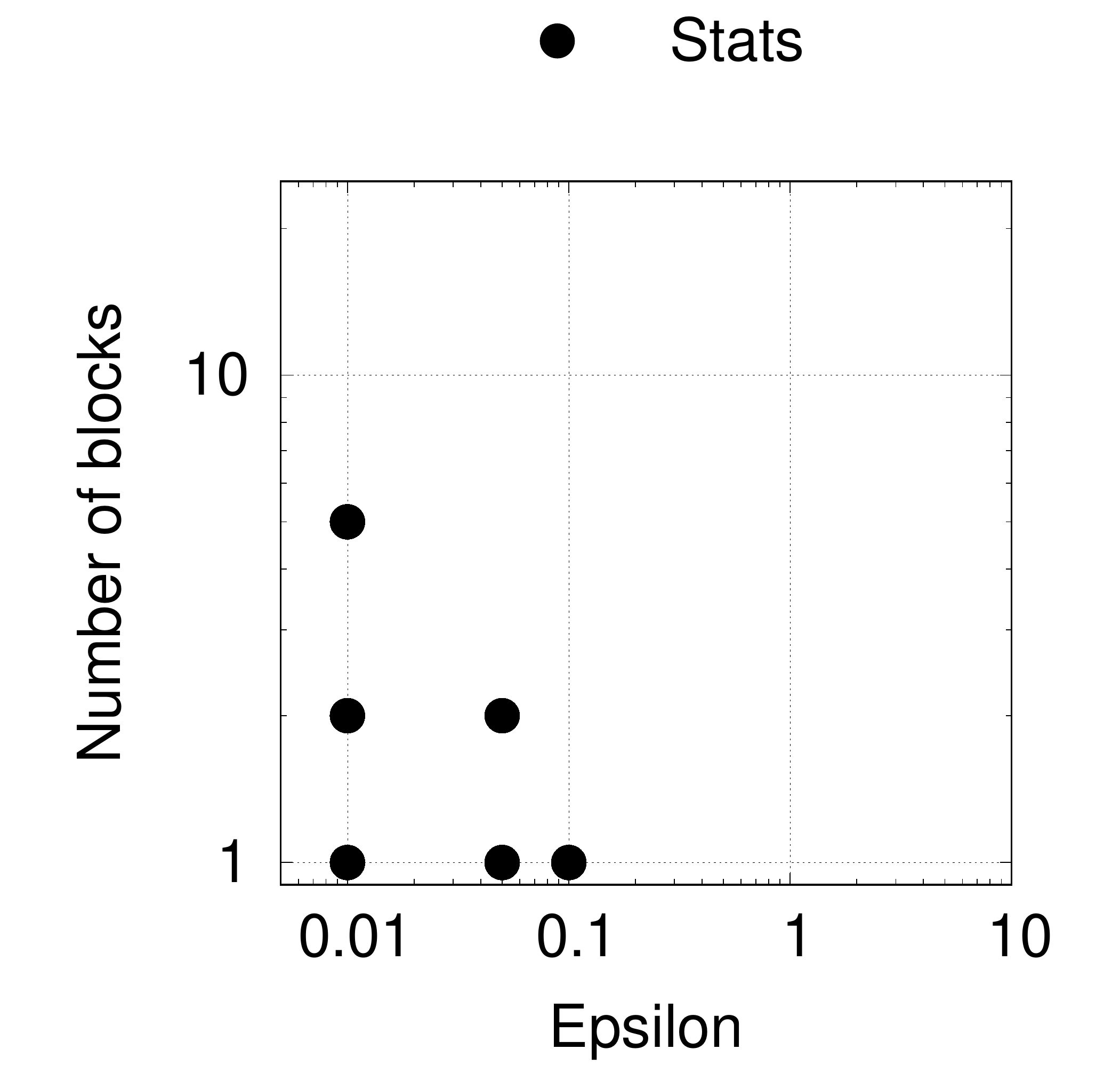}
        \caption{\footnotesize {\bf Statistics demands}}
        \label{fig:appendix:macrobenchmark:workload-mix:stats}
    \end{subfigure}%
    ~
    \begin{subfigure}{0.24\linewidth}
        \centering
        \includegraphics[width=\linewidth]{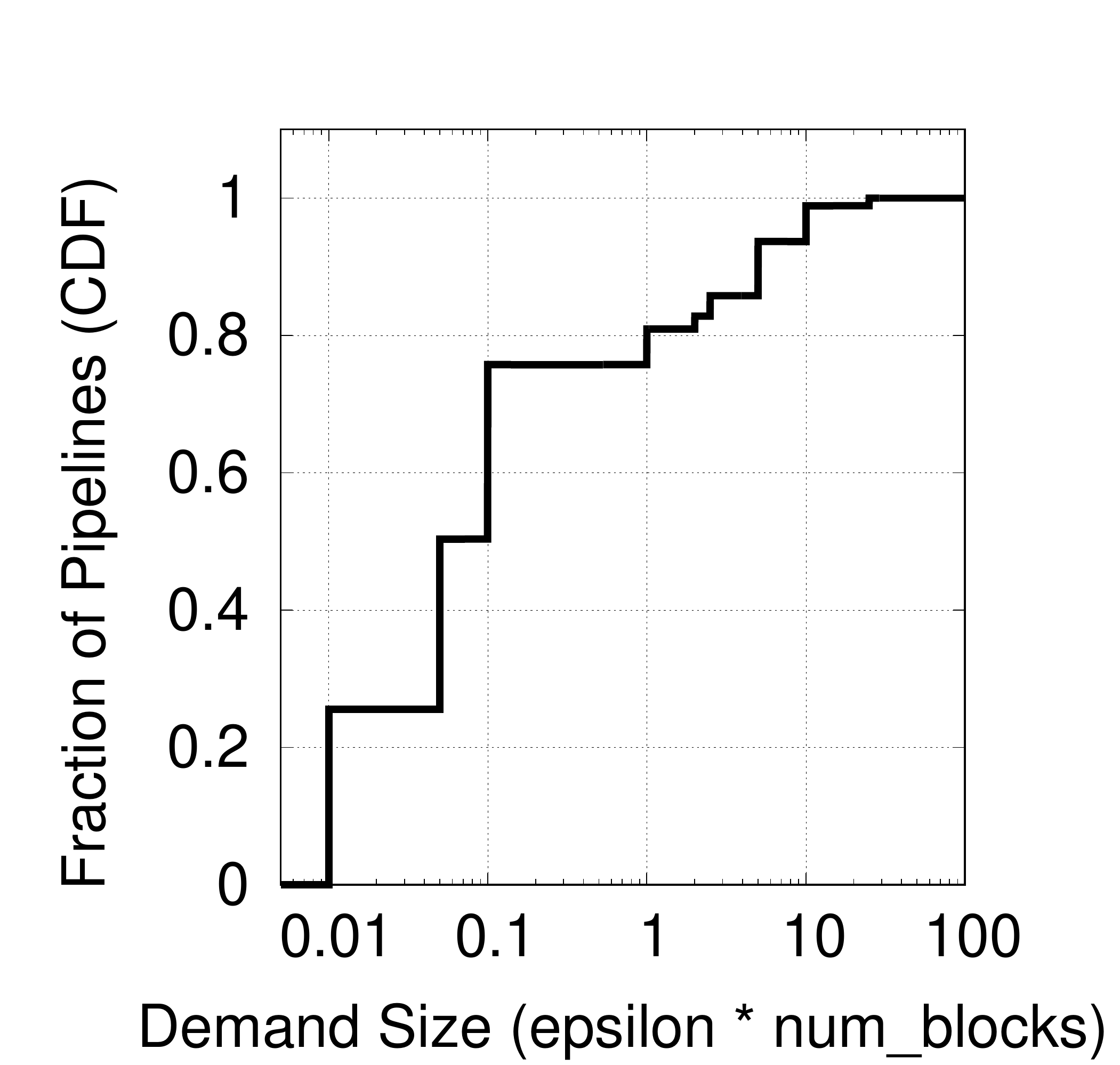}
        \caption{\footnotesize {\bf Distribution of the demands}}
        \label{fig:appendix:macrobenchmark:workload-mix:cdf}
    \end{subfigure}%
    \caption{\footnotesize {\bf Pipeline demands for the Event-DP workload.}}
    \label{fig:appendix:macrobenchmark:workload-mix}

\end{figure*}


Our open-source artifact contains the main parts of the \sysname system, a scheduling simulator as well as experimental setups to reproduce our evaluation results.

\subsection{Scope}


The artifact allows to validate the microbenchmark
(\F\ref{fig:evaluation:microbenchmark:single-block},
\F\ref{fig:evaluation:microbenchmark:mice-percentage},
\F\ref{fig:evaluation:microbenchmark:multiple-blocks},
\F\ref{fig:evaluation:microbenchmark:dpfn-dpft} and
\F\ref{fig:evaluation:microbenchmark:dpfn-rdp})
and the macrobenchmark (\F\ref{fig:evaluation:macrobenchmark:individual-model-accuracy-and-dpf-with-dp-semantic} and \F\ref{fig:evaluation:macrobenchmark:dpf-on-event-dp-workload}).

The privacy resource implementation and the DPF scheduler can be reused on any Kubernetes cluster, as well as modified to study other aspects, such as different scheduling algorithms, or the co-scheduling of privacy budgets with computational resources.

\subsection{Contents}


We release the following parts of the \sysname system: the privacy resource implementation (for both DP and RDP); the DPF scheduler (DPF-T and DPF-N); and an example of Kubeflow pipeline using \sysname.

We also release the discrete-event simulator, which we leverage to study and prototype scheduling algorithms of privacy and computational resources.

We also provide command line interfaces to reproduce: the microbenchmark; the DP workloads (dataset, models and parameters) used for the macrobenchmark; and the evaluation of the DPF scheduler on the macrobenchmark workloads.

The artifact does not contain: the Grafana dashboard; data ingestion pipelines and other data management infrastructure; nor a cloud-agnostic deployment for Kubeflow pipelines.  We can make these components available upon request, but at the time of this publication they are fairly specific to our Kubernetes cluster.


\subsection{Hosting} The artifact is available at \url{https://github.com/columbia/privatekube/releases/tag/v1.0}.


\subsection{Requirements}


This artifact requires a Kubernetes cluster. The documentation explains how to set up a small cluster on a laptop and details the other requirements. Optionally, an NVIDIA GPU can speed up the evaluation.

The privacy resource implementation, the scheduler and the macrobenchmark do not require anything else.
The Kubeflow components and the Kubeflow pipeline example require a Google Cloud Platform Kubernetes cluster with Kubeflow enabled.

It is highly recommended to reproduce the microbenchmark with a beefy machine. It normally takes us several hours to finish it with two 32-core CPUs.



\subsection{Additional Evaluation Results}
\label{sec:artifact-appendix:additional-evaluation-results}

The released artifact supports evaluation of \sysname and DPF beyond the results included in the paper.
We include here a few of the results that we omitted in the paper.

\begin{figure}[H]
    \centering
    \begin{subfigure}{0.5\linewidth}
        \centering
        \includegraphics[width=\linewidth]{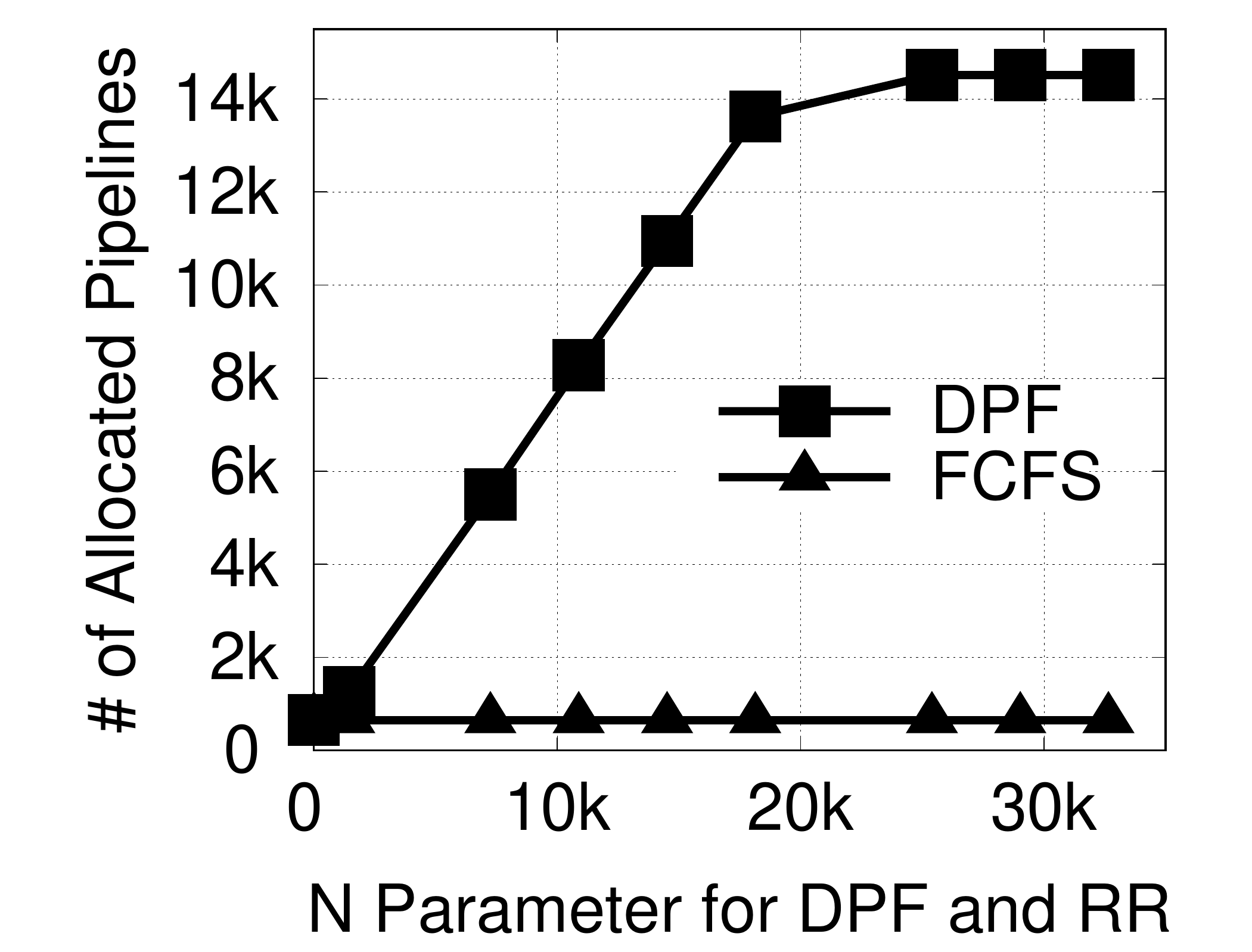}
        \caption{\footnotesize {\bf Number of pipelines allocated.}}
        \label{fig:appendix:microbenchmark:single-block:completed}
    \end{subfigure}%
    ~
    \begin{subfigure}{0.5\linewidth}
        \centering
        \includegraphics[width=\linewidth]{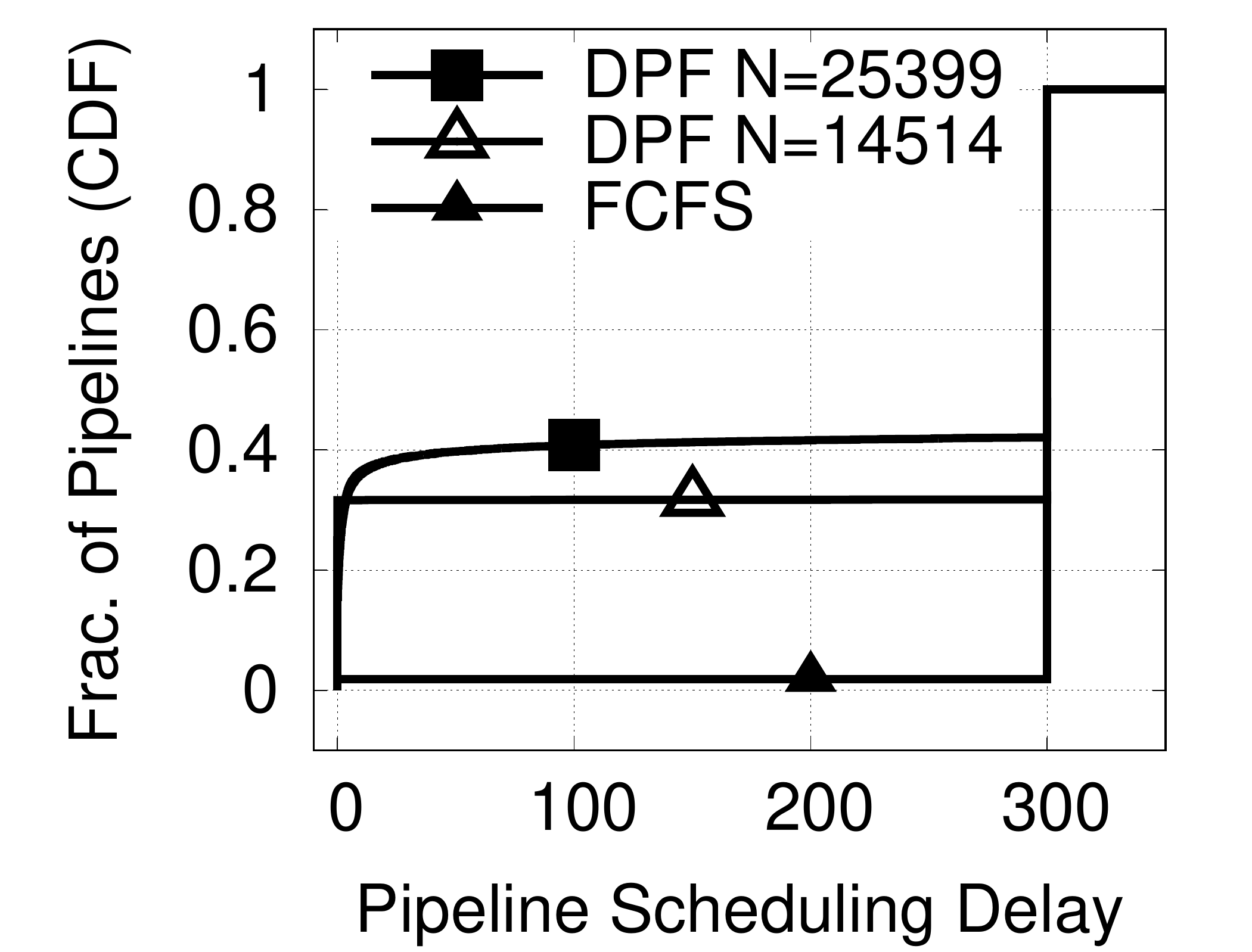}
        \caption{\footnotesize {\bf Scheduling delay.}}
        \label{fig:appendix:microbenchmark:single-block:wait-time}
    \end{subfigure}
    \caption{\footnotesize {\bf R\'enyi DPF behavior on a single block.}}
    \label{fig:appendix:microbenchmark:single-block}
\end{figure}

\begin{figure}[H]
    \centering
    \begin{subfigure}{0.5\linewidth}
        \centering
        \includegraphics[width=\linewidth]{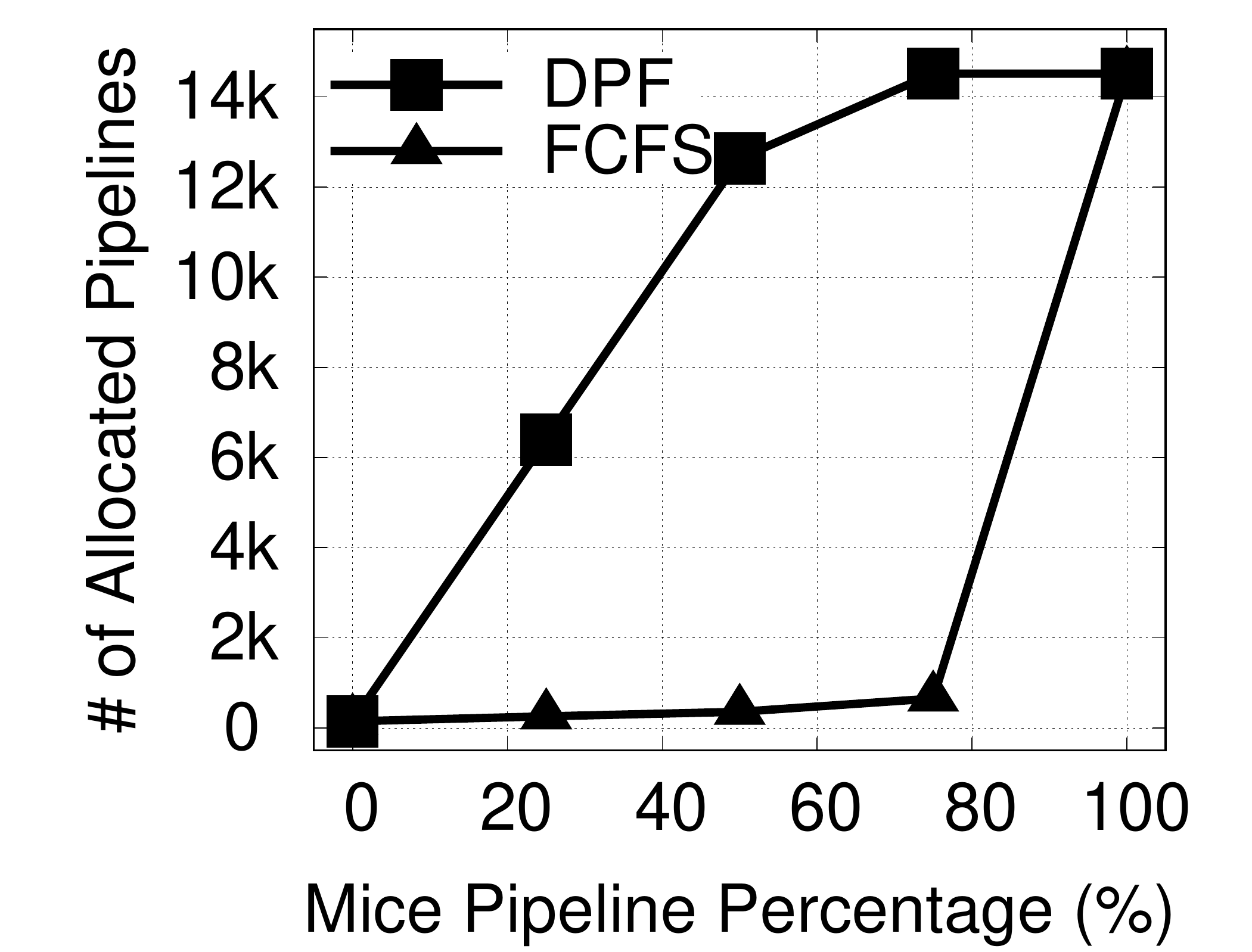}
        \caption{\footnotesize {\bf Number of pipelines allocated.}}
        \label{fig:appendix:microbenchmark:mice-percentage:completed}
    \end{subfigure}%
    ~
    \begin{subfigure}{0.5\linewidth}
        \centering
        \includegraphics[width=\linewidth]{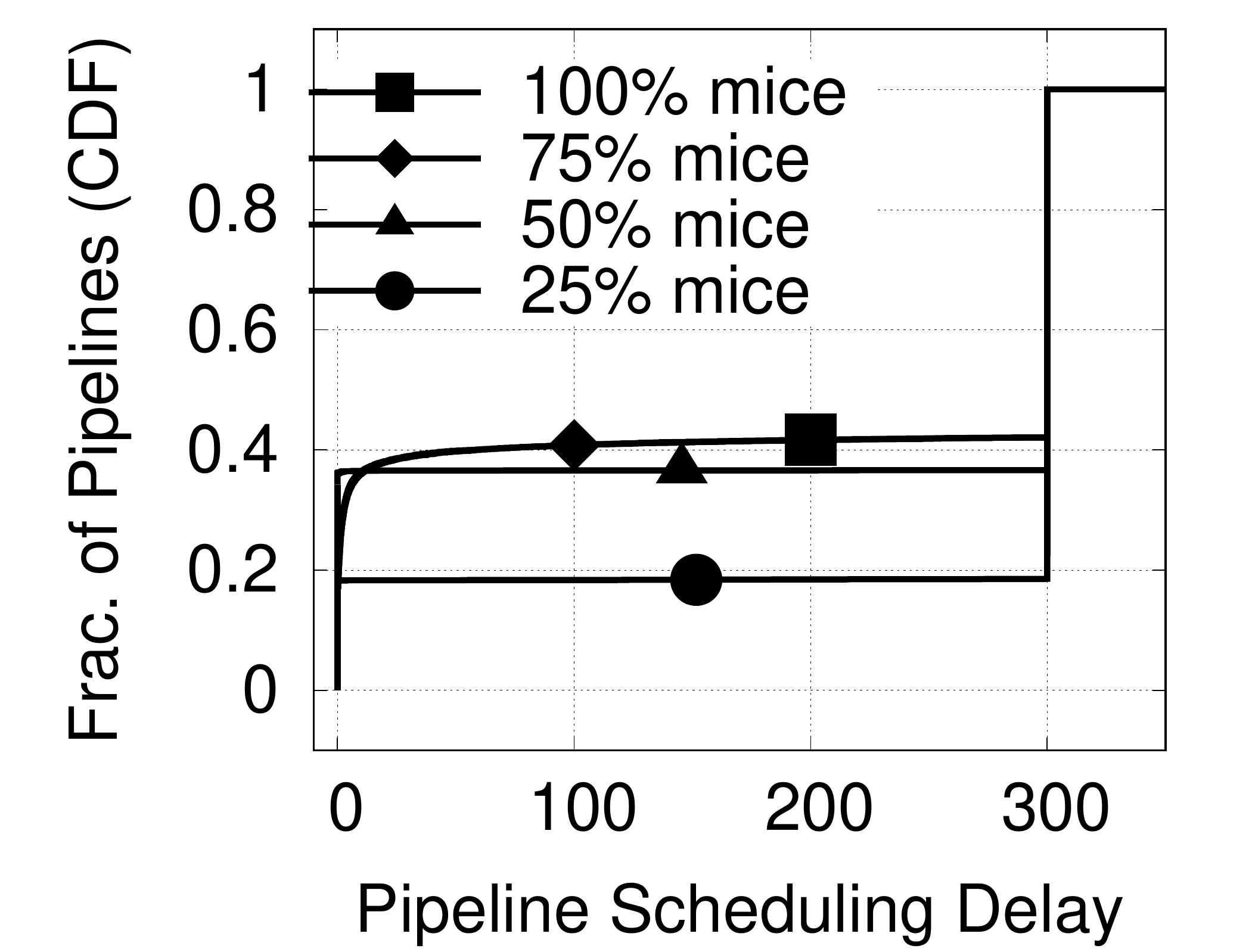}
        \caption{\footnotesize {\bf Scheduling delay.}}
        \label{fig:appendix:microbenchmark:mice-percentage:wait-time}
    \end{subfigure}
    \caption{\footnotesize {\bf R\'enyi DPF behavior with variable workload mix, single block.} DPF N=25,399.}
    \label{fig:appendix:microbenchmark:mice-percentage}
\end{figure}

\begin{figure}[t]
    \centering
    \begin{subfigure}{0.5\linewidth}
        \centering
        \includegraphics[width=\linewidth]{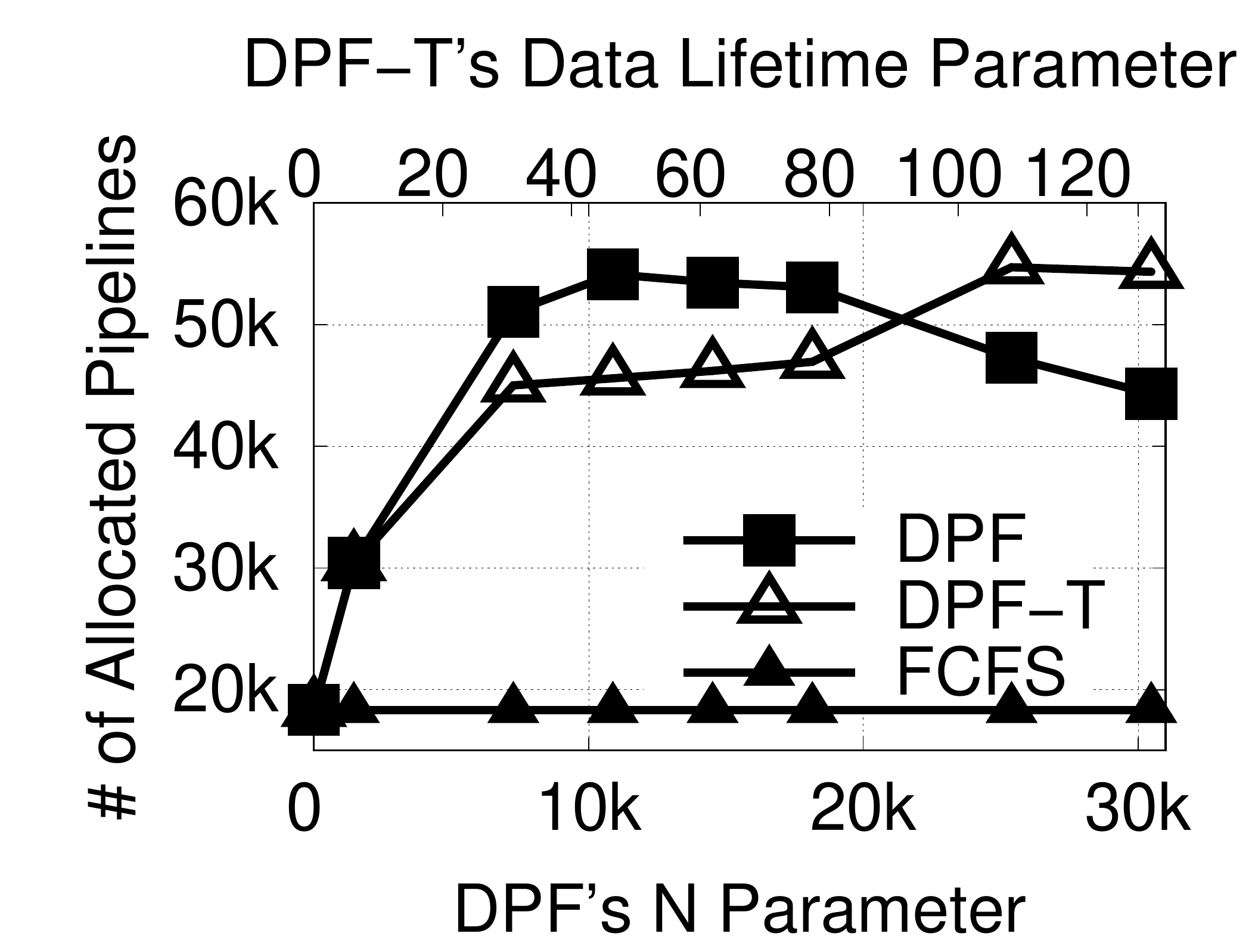}
        \caption{\footnotesize {\bf Number of pipelines allocated.}}
        \label{fig:appendix:microbenchmark:multiple-block:completed}
    \end{subfigure}%
    ~
    \begin{subfigure}{0.5\linewidth}
        \centering
        \includegraphics[width=\linewidth]{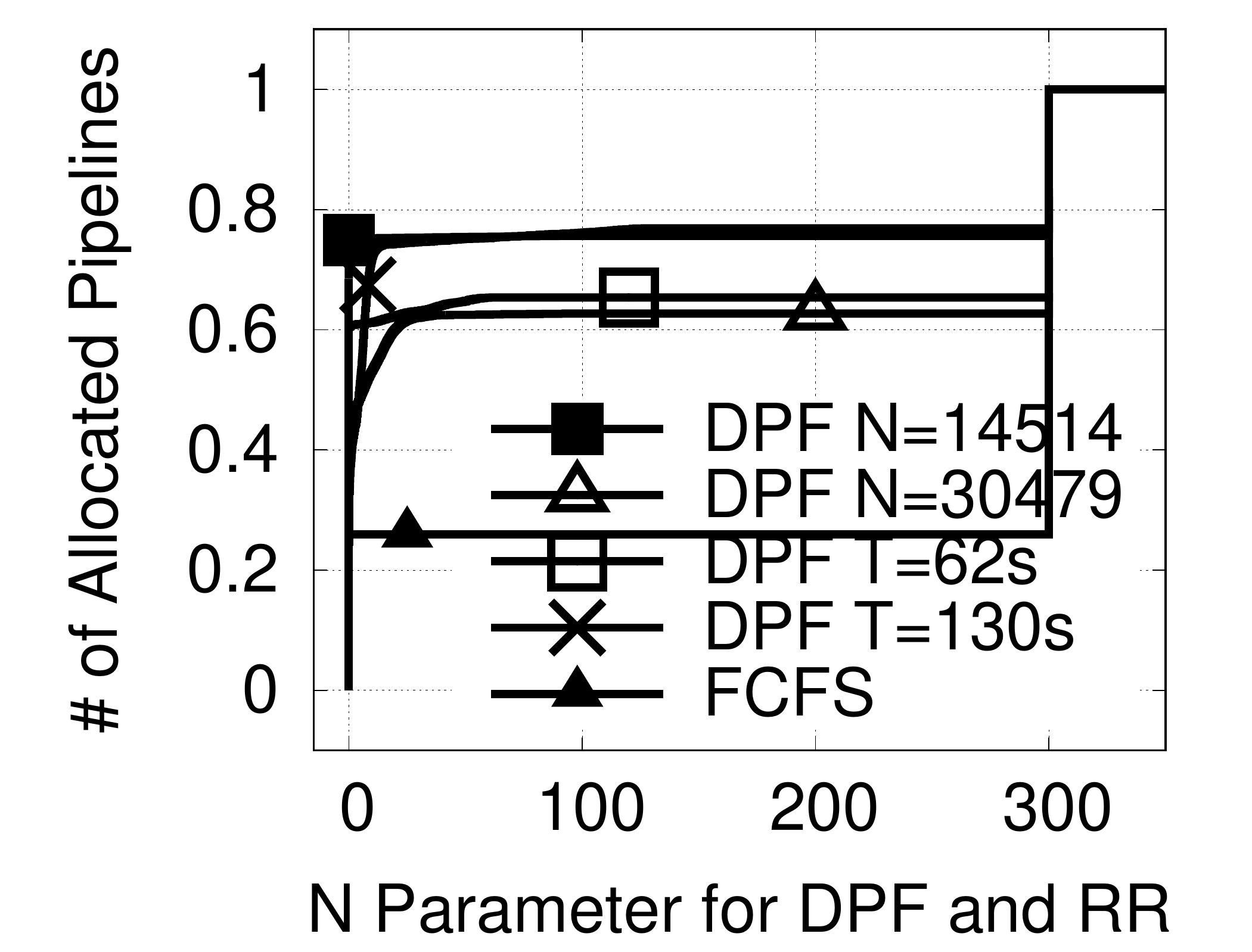}
        \caption{\footnotesize {\bf Scheduling delay.}}
        \label{fig:appendix:microbenchmark:multiple-block:wait-time}
    \end{subfigure}
    \caption{\footnotesize {\bf R\'enyi DPF and DPF-T behaviors on multiple blocks.}}
    \label{fig:appendix:microbenchmark:multiple-blocks}
\end{figure}

\heading{Additional Microbenchmark Results.}
\S\ref{sec:evaluation:microbenchmarks} explores in detail the behavior of DPF with basic composition on one or multiple blocks, and under varied mice::elephant ratios.  Our artifact allows exploration of these behaviors for DPF with R\'enyi composition, as well.  For thoroughness, we include the corresponding graphs here:

\F\ref{fig:appendix:microbenchmark:single-block} (R\'enyi version of \F\ref{fig:evaluation:microbenchmark:single-block}) shows that, when the load is amplified appropriately (as described in \S\ref{sec:evaluation:microbenchmark:dpfn-rdp}), R\'enyi DP can allocate more than $14 \times$ more pipelines than traditional DP for the optimal values of $N$, in the single block setting.

\F\ref{fig:appendix:microbenchmark:mice-percentage} (R\'enyi version of \F\ref{fig:evaluation:microbenchmark:mice-percentage}) shows that increasing the mice percentage has a similar impact on the number of allocated pipelines for DPF under R\'enyi DP and traditional DP. Similar to the basic composition results, FCFS also behaves the same as DPF when the percentage of Mice is either 0\% or 100\%.

\F\ref{fig:appendix:microbenchmark:multiple-blocks} (R\'enyi version of \F\ref{fig:evaluation:microbenchmark:dpfn-dpft}) shows that, similarly to the traditional DP case, DPF performs better for large $N$ and $T$. In addition, $T$ outperforms $N$ for large $N$ values, since all budget is eventually locked.

\heading{Additional Macrobenchmark Results.}
\S\ref{sec:evaluation:macrobenchmarks} shows the results from our macrobenchmark evaluation of the R\'enyi DP instantiation of our system.  Our artifact allows evaluation of the macrobenchmark against the traditional DP instantiation as well.
For completeness, we include here some of the omitted macrobenchmark results:

\begin{figure}[t]
    \centering
    \begin{subfigure}{0.5\linewidth}
        \centering
        \includegraphics[width=\linewidth]{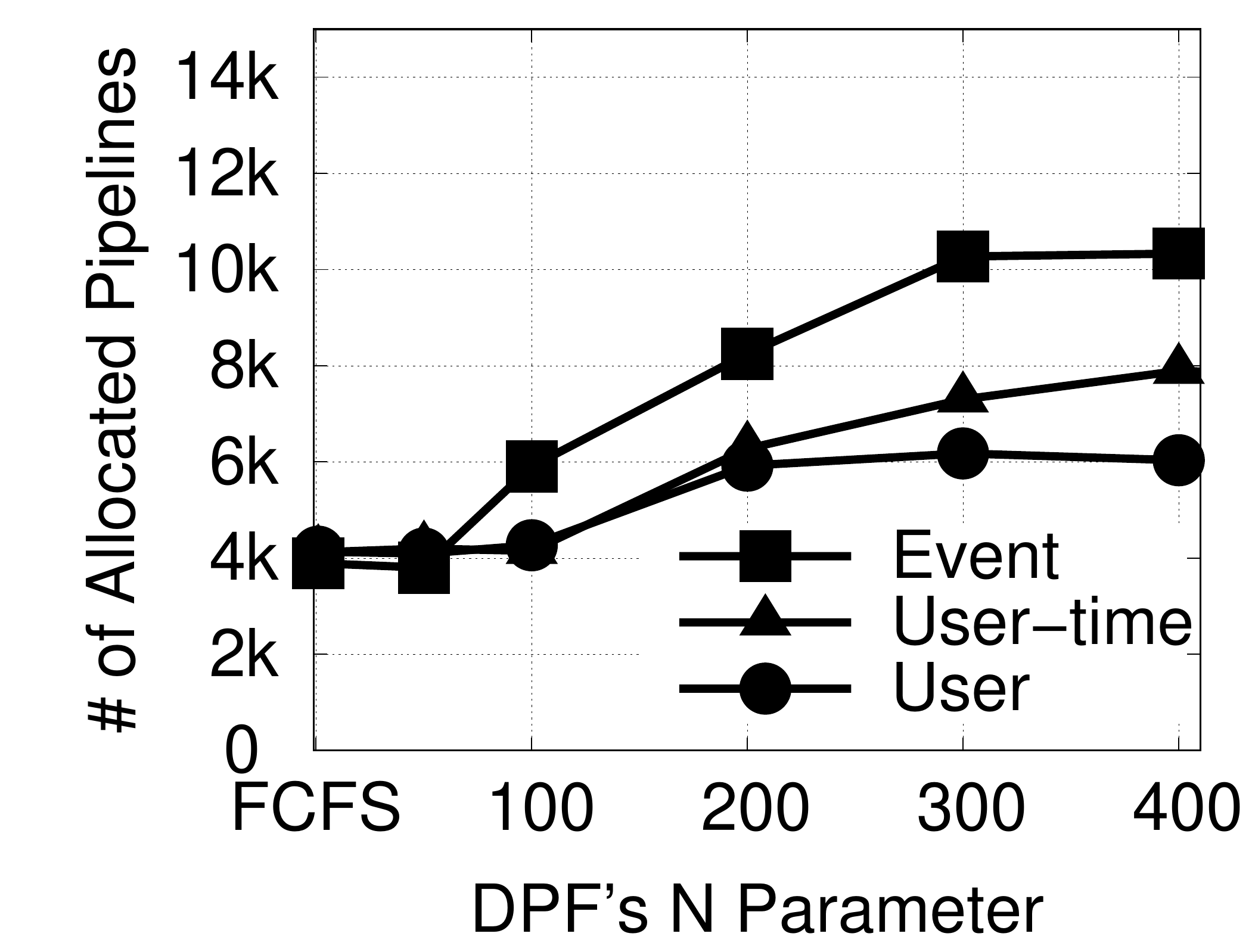}
        \caption{\footnotesize {\bf Allocated for 3 DP semantics}}
        \label{fig:appendix:macrobenchmark:dpf-on-event-dp-workload:completed}
    \end{subfigure}%
    ~
    \begin{subfigure}{0.5\linewidth}
        \centering
        \includegraphics[width=\linewidth]{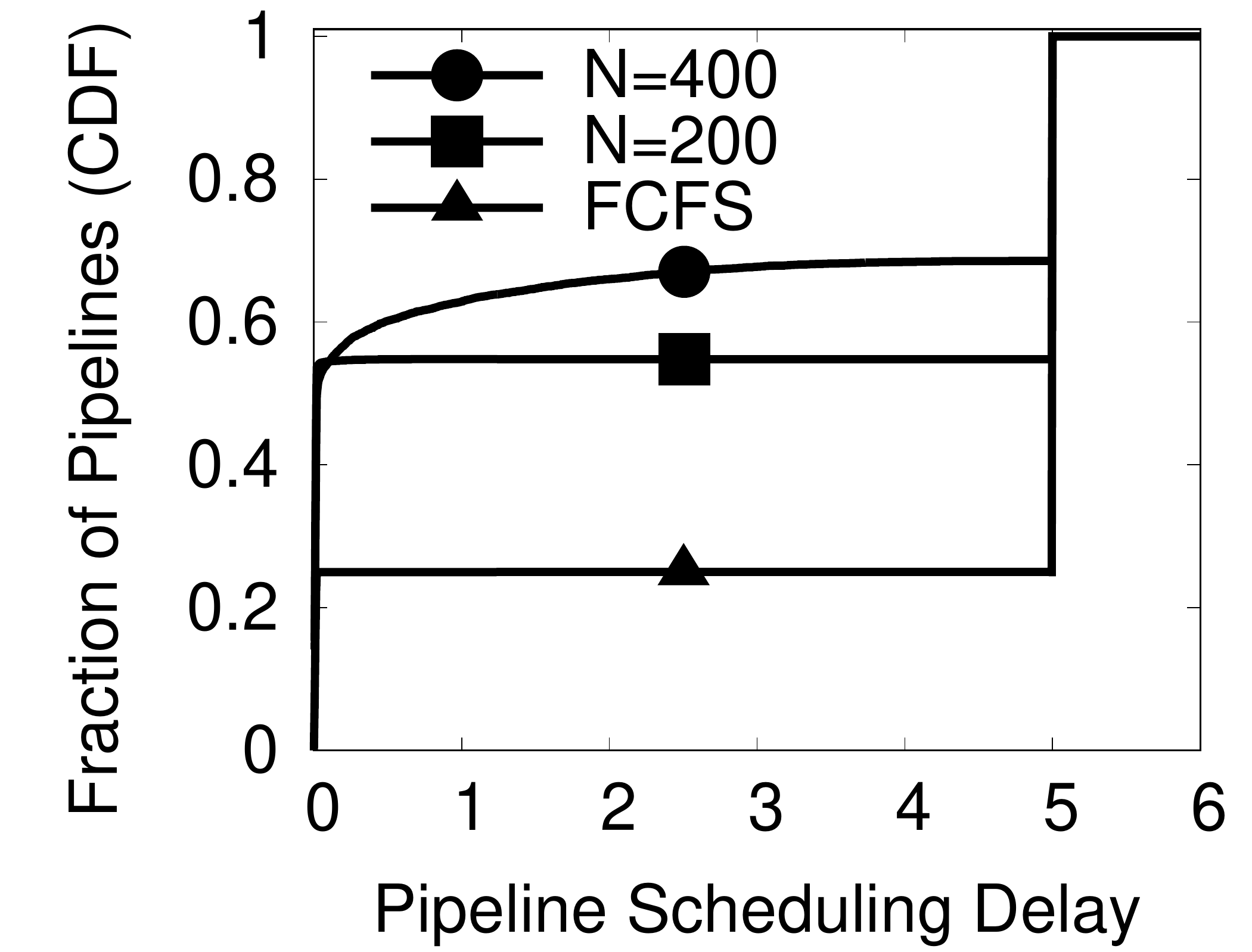}
        \caption{\footnotesize {\bf Scheduling delay for Event DP}}
        \label{fig:appendix:macrobenchmark:dpf-on-event-dp-workload:wait-time}
    \end{subfigure}
    \caption{\footnotesize {\bf DPF behavior on the macrobenchmark workload with basic composition.} The global privacy guarantee is $\epsilon^G = 10$, $\delta^G = 10^{-7}$.}
    \label{fig:appendix:macrobenchmark:dpf-on-event-dp-workload}
\end{figure}

First, in the body of the paper, we provided an analytical description of how we chose privacy demands for our macrobenchmark workload.
\F\ref{fig:appendix:macrobenchmark:workload-mix} plots the distribution of these demands for the pipelines in the Event-DP workload. The $x$-axis of \F\ref{fig:appendix:macrobenchmark:workload-mix:product}, \ref{fig:appendix:macrobenchmark:workload-mix:sentiment}, \ref{fig:appendix:macrobenchmark:workload-mix:stats} represents the $\epsilon$ demand in terms of traditional DP for product classification, sentiment analysis and statistics pipelines.
Each $\epsilon$ also corresponds to the best possible DP-$\epsilon$ for the R\'enyi DP version of a given pipeline.
We can see that the demands are scattered across a wide range of sizes, both in terms of blocks and epsilon, and with finer granularity than the microbenchmark's clear-cut mice and elephants.
Finally, \F\ref{fig:appendix:macrobenchmark:workload-mix:cdf} shows how these varied demands are combined to form a workload.
This workload gives the incoming load in \F\ref{fig:evaluation:macrobenchmark:dpf-on-event-dp-workload} and \F\ref{fig:evaluation:macrobenchmark:dpf-on-event-dp-workload:job-sizes}, which evaluate \sysname's performance with R\'enyi DP.

Second, under the same workload, we add here the results from our evaluation of \sysname on {\em traditional DP} with basic composition. \F\ref{fig:appendix:macrobenchmark:dpf-on-event-dp-workload} (basic composition version of \F\ref{fig:evaluation:macrobenchmark:dpf-on-event-dp-workload}) shows the performance of DPF for the three DP semantics.
We observe the same overall behavior as with R\'enyi DP: stronger semantics can allocate less pipelines, and larger values of $N$ increase the number of granted pipelines.
As expected, R\'enyi DP allocates more pipelines than traditional DP.
However, as illustrated in \F\ref{fig:evaluation:macrobenchmark:dpf-on-event-dp-workload:job-sizes}, the pipelines allocated by R\'enyi DP are qualitatively different from the pipelines allocated by traditional DP.
This effect explains why the gap in the number of allocated pipelines is smaller than in the microbenchmark, in particular when the workload contains larger pipelines (such as under User-DP).

\vbox{}
Appendices~\ref{apx:rdp}, ~\ref{apx:rdp-dpf-proofs}, and~\ref{apx:counter} are the new additions over the OSDI 2021 version of this paper.
\vbox{}

\vspace{-0.3cm}
\section{An Intro to R\'enyi Differential Privacy}
\label{apx:rdp}
\vspace{-0.3cm}
\paragraph{Definition and Intuition.} R\'enyi Differential Privacy (RDP) \cite{8049725} is a relaxation of pure DP that always implies $(\epsilon, \delta)$-DP, though the converse is not true.
Where DP bounds the ratio of probabilities for any possible output of an algorithm on two adjacent datasets $D$ and $D'$ (\S\ref{sec:dp}), RDP bounds the the R\'enyi divergence of order $\alpha$ between the entire distribution over possible outputs.
The R\'enyi divergence of order $\alpha \in ]1, +\infty]$ between the distribution of outputs, noted $D_\alpha(Q(D) || Q(D'))$, is defined as:
\[
    D_\alpha(Q(D) || Q(D')) \triangleq \frac{1}{\alpha -1} \log \E_{x \sim Q(D')} \Big( Q(D) / Q(D') \Big)^\alpha .
\]
Given this definition, a randomized algorithm $Q$ is $(\alpha, \epsilon^{RDP})$-RDP if:
\[
    D_\alpha(Q(D) || Q(D')) \leq \epsilon^{RDP} .
\]

Intuitively when $\alpha$ is close to one, the R\'enyi divergence is the log expectation of the probability ratio (the quantity bounded by DP). Because of the expectation, likely outputs of the algorithm have more weights, and unlikely outputs can have larger ratios without making the R\'enyi divergence too large.
This is why RDP, which bounds the R\'enyi divergence, is a relaxation of DP, which bounds the probability ratio of every outputs, even unlikely ones.
A large $\alpha$ amplifies probability ratios above $1$, and even unlikely events with large probability ratios increase the R\'enyi divergence.
At the limit of $\alpha = \infty$, the R\'enyi divergence is just the maximum ratio of probabilities regardless of the likelihood of the output, implying that $(\infty, \epsilon)$-RDP is pure $(\epsilon, 0)$-DP.

\paragraph{RDP Curve and Translation to DP.} Like for DP, making a computation RDP requires adding noise.
The specific noise distribution used determines the value of $\epsilon^{RDP}$ for all possible $\alpha$ values,
and the privacy guarantees of an RDP mechanism is characterized by its RDP curve $\epsilon^{RDP}(\alpha)$.
As an example, using the Gaussian mechanism with noise standard-deviation $\sigma$ is $\epsilon^{RDP}(\alpha) = \frac{\alpha}{2\sigma^2}$.
We note that for the Gaussian mechanism $\alpha = \infty$, implying $\epsilon^{RDP} = \infty$, which confirms that the Gaussian mechanism cannot enforce pure DP.
For lower values of the R\'enyi divergence order $\alpha$, Proposition 3 from \cite{8049725} provides an {\em RDP to DP translation formula}:
\[
    \big(\alpha, \epsilon^{RDP}\big)\textrm{-RDP} \Rightarrow  \big(\epsilon^{RDP} + \frac{\log(1/\delta)}{\alpha-1}, \delta\big)\textrm{-DP} .
\]
This formula is true for any $\alpha \in ]1, +\infty]$. When translating RDP to DP, we can thus choose the order $\alpha$ yielding the best DP guarantee.

This is how \sysname ensures $(\epsilon^G, \delta^G)$-DP from RDP in \S\ref{sec:dpf-extensions:renyi}.
Remember that for all blocks $j$, there exists $\alpha$ such that $0 \leq \epsilon_j^U(\alpha) \leq \epsilon_j^G(\alpha) = \epsilon^G - \frac{\log(1/\delta)}{\alpha - 1}$. Applying the RDP to DP translation formula directly yields that $(\epsilon^G, \delta^G)$-DP is preserved for each block.

We can also see that a mechanism's RDP curve characterizes the final DP guarantee for all $\delta$.
This parameter is fixed a priori to initialize $\epsilon_j^G(\alpha)$, but never impacts composition, effectively removing the need to perform composition over $\delta$ in \sysname.

\paragraph{Strong Composition.} When composing a sequence of $(\alpha, \epsilon^{RDP}_i)$-RDP computations, the total output is $(\alpha, \sum \epsilon^{RDP}_i)$-RDP.
Hence, the RDP budget at each order $\alpha$ sums under composition, providing an intuitive notion of additive privacy budget, similar to that of basic composition.
Unlike basic composition which scales as $\epsilon k$ however, summing RDP budgets yields strong composition results scaling in $O(\epsilon \sqrt{k})$.
%
As an example,
composing $k$ Gaussian mechanisms yields $\epsilon^{RDP}(\alpha) = k\alpha / 2\sigma^2$,
identical to a Gaussian mechanism with noise $\sigma/\sqrt{k}$ (``$\sqrt{k}$ times less private''), and not $\sigma/k$ as with basic composition.

Like the RDP to DP translation formula, additive composition holds for all $\alpha \in ]1, +\infty]$.
When composing RDP computations with heterogeneous RDP curves, which is always the case in \sysname, it is unclear a priori which order $\alpha$ will give the best guarantee after applying the composition and translation formulae.
We thus keep track of the entire composed RDP curve, and chose the $\alpha$ value yielding the best guarantee a posteriori.
In practice, \cite{8049725} has shown that composing over the whole RDP curve $\epsilon(\alpha)$ is unnecessary. Using a small set of discrete values is sufficient ---$\{1.5,1.75,2,2.5,3,4,5,6,8,16,32,64,+\infty\}$ are typical values--- and the final results are not very sensitive to the choice of this set.
RDP composition has also been shown to be valid under adaptively chosen privacy budgets \cite{feldman2020individual, adaptive_rdp}, fulfilling a key requirement from block composition \cite{sage}.

\vspace{-0.3cm}
\section{Proofs for DPF with R\'enyi DP}
\vspace{-0.3cm}
\label{apx:rdp-dpf-proofs}


\begin{definition}[{\em dominant share}] We fix $A$, a set of RDP orders, a $(\epsilon^G, \delta^G)$, a global DP guarantee to enforce, and an initial block budget curve $\epsilon_j^G(\alpha) = \epsilon^G - \frac{\log(1/\delta^G)}{\alpha -1}$ for each block $j$.

    The dominant share of a pipeline $i$ with demand vector $d_i$ is:
    $$\textrm{DominantShare}_i = \max_{\alpha \in A, j : d_{i,j} > 0} \frac{d_{i,j}(\alpha)}{\epsilon_i^G(\alpha)} $$

\end{definition}

\begin{definition}[{\em fair demand pipeline}]
    A fair demand pipeline has two properties:
    (a) the pipeline is within the first N pipelines that requested some budget for all its requested blocks,
    and (b) its demand for each one of the blocks is smaller or equal to the RDP fair share (\ie for pipeline $i$, $\forall j, \forall \alpha \in A: d_{i,j}(\alpha) \leq \epsilon^{FS}$, where $\epsilon_j^{FS}(\alpha) = \epsilon_j^G(\alpha)/N$).

\end{definition}

\begin{theorem}[{\em sharing incentive}]
    \label{thm:rdp-sharing-incentive}

    A fair demand pipeline is granted immediately.
\end{theorem}
\begin{compactproof}


    Consider a fair demand pipeline $i$ with demand $d_i$.  We proceed by induction over the number of waiting pipelines.

        {\em Base case:} no waiting pipelines. Consider any $\alpha \in A$ and $j$ such that $d_{i,j}(\alpha) > 0$.  $\epsilon^{FS}_j(\alpha) \leq \epsilon^U_{j}(\alpha)$ since $\epsilon^{FS}_j(\alpha)$ is unlocked by $d_i$. $d_i$ is fair so $d_{i,j}(\alpha) \leq \epsilon^{FS}_j(\alpha) \leq \epsilon^U_{j}(\alpha)$. The pipeline is granted, and there is no fair waiting pipeline.

        {\em Induction step:} Consider any waiting pipeline $k$ with demand $d_k$, and its dominant share $\textrm{DominantShare}_k$. Consider any $\alpha \in A$ and $j$ such that $d_{i,j}(\alpha) > 0$.

    By the induction assumption no fair pipeline is waiting, so $\textrm{DominantShare}_k > \epsilon^{FS}_j(\alpha) \geq \textrm{DominantShare}_i$.
    For the same reason as before, $d_{i,j}(\alpha) \leq \epsilon_j^{FS}(\alpha) \leq \epsilon^U_{j}(\alpha)$, so in particular the condition for the \textsc{CanRun} function of Algorithm~\ref{alg:dpf-rdp} holds and $d_i$ can be granted.
    $d_i$ is ordered first so it is granted.
\end{compactproof}

\begin{theorem}[{\em strategy-proofness}]
    \label{thm:rdp-strategy-proofness}
    A pipeline has no incentive to misreport its demand.
\end{theorem}
\begin{compactproof}


    A pipeline has no incentive to ask for more budget than its real demand, because: (a) its utility would
    not increase if it obtains more budget than it needs, (b) its dominant share will be greater or equal so it can only become less likely to get scheduled.
    A pipeline also has no incentive to ask for less budget than its real demand, because its utility will drop
    to zero if it is not allocated its demanded budget.
\end{compactproof}

\begin{theorem}[{\em dynamic envy-freeness}]
    \label{thm:rdp-dynamic-envy-freeness}
    A pipeline present at time $t$ cannot envy the allocation of another pipeline present at time $t$, except if their $\textrm{dominantshare}$s are identical.
\end{theorem}
\begin{compactproof}


    Consider pipeline $i$ in the system at time $t$. There are two cases.
    Case 1: $i$ was granted. Its utility cannot improve due to all-or-nothing utility, there is no envy.
    Case 2: $i$ is waiting. Consider any non-identical pipeline $j$ present at time $t$.
    \begin{itemize}
        \item If $j$ is also waiting, $i$ does not envy $j$.
        \item If $j$ has been granted before $i$ entered the system, $i$ doesn't envy $j$.
        \item If $j$ has been granted after $i$ entered the system: \begin{itemize}
                  \item Either $\textrm{DominantShare}_j < \textrm{DominantShare}_i$ ; or
                  \item $\textrm{DominantShare}_j > \textrm{DominantShare}_i$ but $i$ could not be granted while $j$ could.
              \end{itemize}
    \end{itemize}

    If $\textrm{DominantShare}_j < \textrm{DominantShare}_i$: since the dominant share is computed as the maximum over all blocks and \textit{all} $\alpha \in A$, there exists $\alpha \in A$ and a block $b$ such that $d_{i,b}(\alpha) > d_{j,b}(\alpha)$. In both cases $i$ cannot be granted from $j$'s allocation, which would thus give $i$ a utility of zero. $i$ cannot envy $j$.
\end{compactproof}

\begin{theorem}[{\em Pareto efficiency}]
    \label{thm:rdp-Pareto-efficiency}
    No allocation from unlocked budget can increase a pipeline's utility without decreasing another pipeline's utility.
\end{theorem}
\begin{compactproof}


    Consider pipeline $i$.
    If $d_i$ was already allocated, its utility cannot improve due to all-or-nothing utility.
    If $i$ is waiting, it cannot be allocated from unlocked budget as DPF grants pipelines until no pipeline can be allocated.
    Allocating $d_i$ would require extra budget, which can only come from another allocated pipeline. Since each allocated pipeline has exactly its requested budget: this would decrease its utility from one to zero, which is not Pareto-improving.
\end{compactproof}

\vspace{-0.3cm}
\section{User DP with a Streaming Counter}
\vspace{-0.3cm}
\label{apx:counter}
This section provides more details about the streaming counter described in \S\ref{sec:varied-dp-semantics}. We leverage the algorithm of \cite{Chan2011PCR} (Algorithm 2), which provides an event-level continuous counter with low error and under low DP budget.
The high-level idea of the algorithm and its application in \sysname is summarized in Fig.~\ref{f:dp_count} and explained in \S\ref{appendix:count-algo} and \S\ref{appendix:count-lb}.
\S\ref{appendix:count-rdp} shows how to use the counter with R\'enyi DP.

\begin{figure}[h]
    \centering
    \includegraphics[width=\linewidth]{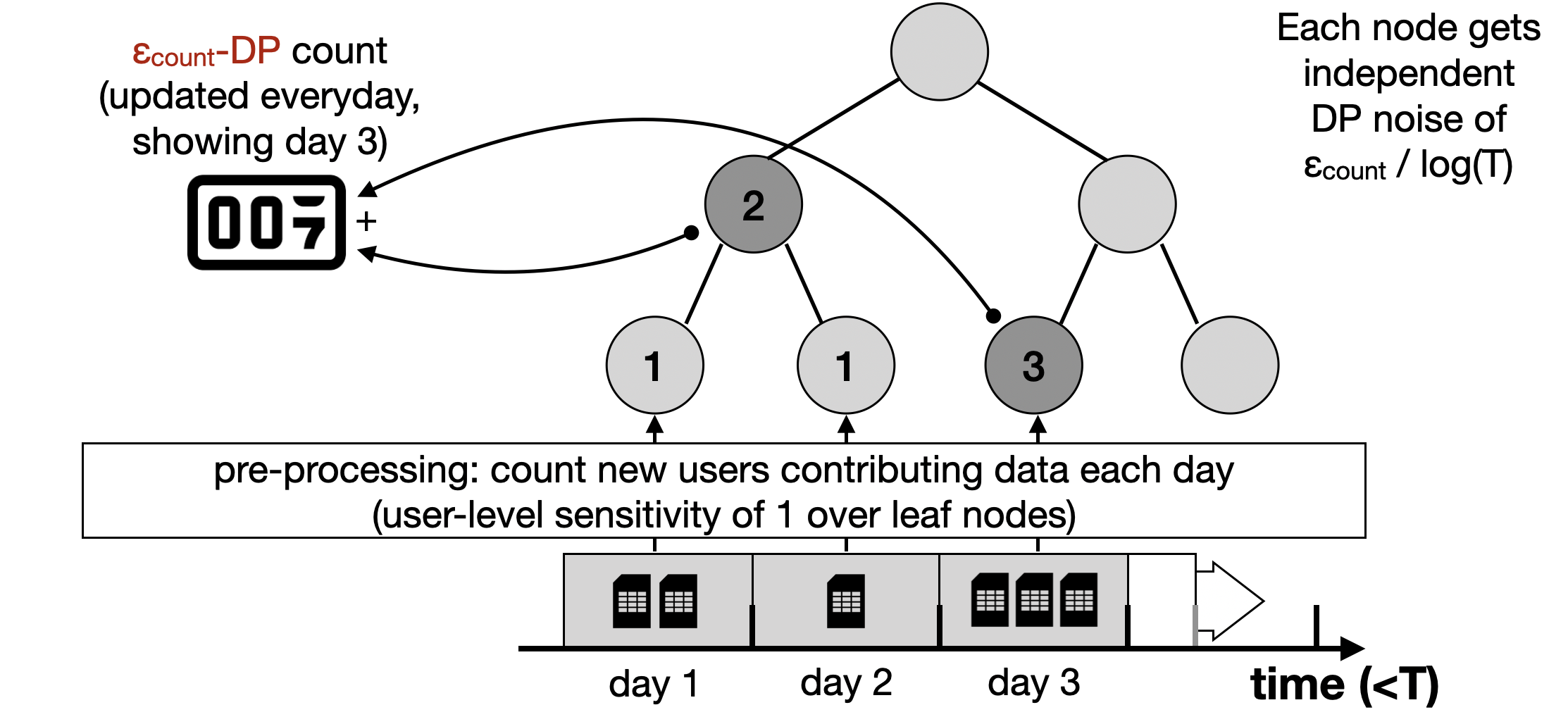}
    \caption{\footnotesize {\bf Continuous DP user count.}
        Adapted from~\cite{Chan2011PCR} to counting users with user-level $\epsilon_{count}$-DP. Each day while $t<T$, the number of new users is computed (leaf nodes). A binary tree keeps counts at different resolutions, each node contains the count over every day it spans. Each node, when complete, receives independent noise to ensure $\frac{\epsilon_{count}}{\log_2(T)}$-DP. The counter outputs the sum of the minimum number of complete nodes that span the desired time span (here the first three days).
    }
    \label{f:dp_count}
\end{figure}

\subsection{Algorithm}
\label{appendix:count-algo}

The first step is to map the time-based data stream of event-level observations to a user count.
To this end, time is divided in intervals (e.g., days).
At each time interval, \sysname computes the number of new users that contributed data in this interval.
Note that while this computation is stateful (it requires keeping track of all existing users), it has sensitivity 1 to the addition or removal of any user, over the entire stream. That is, adding or removing a user will change the value of only one interval, by only one unit.
Ensuring event-level DP over this stream of counts is thus user-level differentially private, and we can leverage the event-level DP algorithm from \cite{Chan2011PCR} to do this.
The key idea to maintain a continuously updated counter in \cite{Chan2011PCR} is to build a binary tree of counts, as illustrated in Fig.~\ref{f:dp_count}.
Each node holds the total count over all leaves it spans, with DP noise added independently of the noise of children and parents.
For a tree spanning less than $T$ time periods (here we assume $T$ is a power of $2$ but the extension is straightforward), changing one leaf affects $\log_2(T)$ nodes, so adding DP noise to ensure $\frac{\epsilon_{count}}{\log(T)}$-DP for each node ensures the final counter is $\epsilon_{count}$-DP.
\cite{Chan2011PCR} also shows a continuous counter not bounded by $T$, which adds only slightly more noise. However, using daily updates \sysname can operate for almost ninety years ($2^{15}$ days) by adding noise scaled by $\frac{\epsilon_{count}}{15}$, which we deem sufficient.

\subsection{Lower bound on the count}
\label{appendix:count-lb}

Of course the count is noisy. If it is lower than the true count (in green on Fig.~\ref{f:user_splitting}) some users with data will not be used but their DP budget can be used later.
If the DP count is higher than the true count (in blue), a few empty users will be queried, and their DP budget consumed.
To avoid the latter scenario, it may be useful to compute a high probability lower bound on the true count. By using this lower bound, we know that with high probability, the true count is always under-estimated.
Computing a count up to $t <T$ by summing the minimum number of nodes requires at most $\log_2(T)$ nodes (actually less than $\log_2(t)$).
Since \sysname uses the Laplace mechanism in the counter with parameter $\frac{\epsilon_{count}}{\log(T)}$, summing $\log_2(T)$ yields an error lower than $\frac{4}{\epsilon}\log^{1.5}_2(T)\log(\frac{T}{\beta})$ at each time interval (see \cite{dwork2010differential} Theorem 4.1), with probability at least $1-\beta$.
Replacing the counter with $\textrm{count} - \frac{4}{\epsilon}\log^{1.5}_2(T)\log(\frac{T}{\beta})$ ensures that the count is always underestimated with probability at least $1-\beta$. For instance, using a budget of $\epsilon_{count}=0.1$, and $T=2^{15}$ or about ninety years of daily updates, subtracting $17.5k$ from the count ensures it is never over-estimated with probability at least $0.999$.

\subsection{R\'enyi DP accounting}
\label{appendix:count-rdp}

\sysname uses R\'enyi DP accounting to make better use of the privacy budget (see \S\ref{sec:dpf-extensions:renyi}).
Since \cite{Chan2011PCR} uses traditional DP only, we can use the translation formula from pure DP to R\'enyi DP from \cite{8049725}, stating that an $\epsilon_{count}$-DP mechanism is $(\alpha, 2\epsilon^2_{count}\alpha)$-RDP.
This is what we use in \S\ref{sec:dpf-extensions:renyi} for simplicity, but we can provide a tighter, more direct analysis of the Binary Mechanism (Algorithm 2 from  \cite{Chan2011PCR}) as follows:

\begin{theorem}[RDP Curve of the Binary Mechanism]
    \label{th:rdp-count}

    For $T \in \mathbb{N}$, $\epsilon >0$ and $\alpha > 1$, the Binary Mechanism preserves $(\alpha, \epsilon(\alpha))$-RDP with $\epsilon(\alpha) =$

    $$  \frac{\log T}{\alpha-1}\log \left[ \frac{\alpha}{2\alpha -1}\exp\left(\frac{(\alpha-1)\epsilon}{\log T}\right) +  \frac{\alpha-1}{2\alpha - 1} \exp\left(-\frac{\alpha\epsilon}{\log T}\right) \right] $$

    This curve corresponds to the sum of $\log T$ curves for the Laplace Mechanism with noise $\frac{log T}{\epsilon}$.
\end{theorem}

\begin{proof}

    Take $\alpha > 1$.
    Consider an item arriving at $t \in [T]$. Let's note $P$ the distribution of the outputs of the Binary Mechanism when $\sigma(t) = 0$, and $Q$ the distribution when $\sigma(t) = 1$. The output space $\mathcal{T}$ is the set of binary trees $t$ of depth $\log T$.

    The R\'enyi divergence of order $\alpha$ is:

    $$D_\alpha(Q\parallel P) = \frac{1}{\alpha - 1} \log \int_{t \in \mathcal{T}} P(t)^\alpha Q(t)^{1-\alpha}$$

    As observed in \cite{Chan2011PCR}, at most $\log T$ nodes can be affected (with sensitivity 1) if $\sigma(t)$ is flipped. Since each node is noised i.i.d, we have:

    $$D_\alpha(Q\parallel P) = \frac{1}{\alpha - 1} \sum_{i=1}^n \log \int_{-\infty}^{+\infty} p_i(x)^\alpha q_i(x)^{1-\alpha}dx$$

    where $p_i, q_i$ are the probability density functions of the $i$th affected node.

    Since the noise for each node is drawn from $\operatorname{Lap}(\frac{log T}{\epsilon})$, the RDP curve of the Laplace Mechanism \cite{8049725} gives:

    $$\begin{array}{l}
            \frac{1}{\alpha - 1} \sum_{i=1}^n \log \int_{-\infty}^{+\infty} p_i(x)^\alpha q_i(x)^{1-\alpha}dx
            \\
            \le \frac{\log T}{\alpha-1}\log \left[ \frac{\alpha}{2\alpha -1}\exp\left(\frac{(\alpha-1)\epsilon}{\log T}\right) +  \frac{\alpha-1}{2\alpha - 1} \exp\left(-\frac{\alpha\epsilon}{\log T}\right) \right]
        \end{array}
    $$

\end{proof}


\end{document}